%% file: main_theory.tex
\documentclass[12pt]{article}
\usepackage{times,amssymb,amsthm,amsmath,amsfonts,eurosym,geometry,ulem,graphicx,caption,color,setspace,sectsty,comment,footmisc,caption,pdflscape,array,threeparttable,tikz,pgfplots,subcaption,caption,natbib,tabularx} 
\PassOptionsToPackage{draft}{graphicx}
\newsavebox{\measurebox}
\usepackage[scale=2]{ccicons}
\usepackage{pgfplots}
\usepgfplotslibrary{dateplot}
\usetikzlibrary{intersections,shapes.multipart}
\usepackage{pifont}
\usetikzlibrary{tikzmark}
\usetikzlibrary{shapes,arrows,backgrounds,positioning}
\usetikzlibrary{intersections,shapes.multipart}
\usepackage{tikz}
\usetikzlibrary{fit,calc}

\usepackage{longtable}
 \usepackage{booktabs} 

\usepackage{pdflscape}
\usepackage{afterpage}
\usepackage{capt-of}

\usepackage{footmisc}
\usepackage[colorlinks=true,
citecolor=blue,
linkcolor=blue,
anchorcolor=blue,
urlcolor=blue]{hyperref}

\newtheorem{assumption}{Assumptions}

\makeatletter
\newtheorem*{assumption*}{Assumptions}
\makeatother
\newtheorem{lemma}{Lemma}
\newtheorem{corollary}{Corollary}
\newtheorem{proposition}{Proposition}
\newtheorem{theorem}{Theorem}
\theoremstyle{definition}
\newtheorem{definition}{Definition}
\theoremstyle{example}

\newcommand{\E}{\mathbb{E}}
\newcommand{\Var}{\mathrm{Var}}
\newcommand{\R}{\mathbb{R}}

\newcommand{\dd}{\mathrm{d}}

\newcommand{\norm}[1]{\left\lVert #1 \right\rVert}

\newcommand{\Ran}{\mathcal R}
\newcommand{\Null}{\mathcal N}
\newcommand{\calT}{\mathcal T}
\newcommand{\calG}{\mathcal G}
\newcommand{\calN}{\mathcal N}
\newcommand{\calV}{\mathcal V}
\newcommand{\calM}{\mathcal M}
\newcommand{\Hpi}{\mathcal H_\pi}
\newcommand{\psib}{\psi_{\mathrm b}}
\newcommand{\psieff}{\psi_{\mathrm{eff}}}
\newcommand{\Vb}{V_{\mathrm b}}
\newcommand{\Ve}{V_{\mathrm{eff}}}
\newcommand{\ip}[2]{\left\langle #1,#2\right\rangle}
\newcommand{\Pn}{\mathbb P_n}
\newcommand{\ind}{\mathbf 1}
\newcommand{\doi}[1]{\href{https://doi.org/#1}{\nolinkurl{doi:#1}}}
\usepackage[titletoc,toc,title,page,header]{appendix}
\usepackage{minitoc}

\noptcrule
\usepackage{caption}

\title{Nonparametric Identification and Estimation of Causal Effects on Latent Outcomes}

\date{\today}

\author{Jiawei Fu\footnote{Assistant Professor, Duke University (\url{jiawei.fu@duke.edu})} \quad Donald P. Green\footnote{Burgess Professor of Political Science, Columbia University (\url{dpg2110@columbia.edu}).  \\ 
We thank David Broockman, Alex Coppock, Macartan Humphreys, Josh Kalla, Patrick Liu, Ryan Moore, Libby Jenke, Justin Grimmer, Jacob Montgomery, Kosuke Imai, Teppei Yamamoto, participants of the UC San Diego Political Methodology and American Politics Speaker Series and Columbia Political Methodology Workshop for their helpful comments on earlier drafts. A preliminary R package can be found at \href{https://github.com/Jiawei-Fu/LatentOutcomes}{GitHub}.
}}
\pgfplotsset{compat=1.18} 
\begin{document}

\maketitle
\singlespacing



\begin{abstract}
How can we conduct causal inference when the outcome of interest is latent and imperfectly measured by multiple indicators? We develop a nonparametric framework for identifying and estimating average treatment effects on latent outcomes in randomized experiments. Researchers commonly aggregate multiple outcome indicators into a summary measure before estimating treatment effects. We identify a fundamental but often overlooked comparability problem arising from this practice: adding, removing, or replacing indicators can alter the causal estimand. We formalize this problem and show that commonly used procedures generally fail to preserve comparability across measurement sets. To address this problem, we propose a design-based approach that uses nonparametric bridge functions. We further characterize the semiparametric efficiency bound, construct an efficient estimator, and establish necessary and sufficient conditions under which additional measurements yield strict efficiency gains. We illustrate our approach through an application to the Hajj visa lottery. Our framework provides practical guidance for measurement design, identification, and efficient causal estimation with latent outcomes.


\noindent 
\vspace{.1in}

\end{abstract}

\thispagestyle{empty}

\clearpage
\doparttoc 
\faketableofcontents 

\setcounter{page}{1}

\singlespacing

\section{Introduction}

Randomized experiments are widely regarded as the most credible design for estimating causal effects. Yet in many substantive applications, the outcome of real interest is not directly observed. Researchers often care about latent constructs such as ideology, attitude, state capacity, political trust, social capital, mental health, cognitive ability, or human capital. These quantities are theoretically central, but they are not observed directly. Instead, researchers measure them indirectly through survey items, tests, administrative indicators, behavioral traces, or other proxies. Although the experimental literature has made enormous progress on identification and estimation under random assignment, it has paid far less attention to a basic and pervasive problem: how to conduct causal inference when the outcome itself is latent and must be inferred from multiple imperfect measurements. 


 
This paper is motivated by a fundamental comparability problem in causal inference with latent outcomes that existing approaches do not adequately address. When multiple imperfect indicators are available, researchers typically construct a summary score from a selected set of measurements. Common approaches include simple averages, principal component analysis (PCA), and inverse covariance weighting (ICW) \citep{anderson2008multiple}. Although these methods are often useful for descriptive purposes, they are not generally designed to estimate the average causal effect of treatment on the latent outcome itself. We define a family of scoring rules to be measurement-set comparable if the resulting causal estimand remains invariant across different sets of measurements. We establish that a necessary and sufficient condition for comparability is that the conditional expectations of the scores, given the latent outcome, differ only by additive constants. We show that the commonly used scoring procedures mentioned above generally fail to satisfy this condition. Consequently, the estimated causal effect may depend on the arbitrary choice of which measurements to include in the analysis.

In other words, simply adding, removing, or replacing a measurement of the latent outcome can effectively redefine the outcome and thereby alter the causal estimand. Consequently, apparent differences in estimated causal effects may reflect differences in measurement rather than differences in the underlying treatment effect on the latent outcome. This problem also implies that, even when researchers seek to estimate average treatment effects on the same latent outcome across different studies, the resulting estimates may be expressed on different scales. Such a lack of comparability poses a serious obstacle to the accumulation and synthesis of empirical evidence. Alternative approaches, such as structural equation modeling (SEM) \citep{bollen1989structural} and item response theory (IRT) models \citep{stoetzer2022causal}, explicitly specify the relationship between the latent outcome and its observed indicators. However, these approaches rely on assumptions about the measurement model, raising concerns about potential model misspecification.

This paper develops a general nonparametric framework for causal inference with latent outcomes measured by multiple imperfect indicators and provides semiparametrically efficient estimators. Our solution centers on two ideas: a benchmark measurement and a measurement bridge function. First, the fundamental reason that conventional scoring rules fail to satisfy measurement-set comparability is that a latent outcome has no intrinsic metric. Without imposing a structural measurement model, its empirical scale is determined by the measurements. To ensure that different sets of measurements target the same latent outcome without imposing parametric model assumptions, we need an anchor. Our solution is therefore to select a benchmark measurement $Y_1$ and define the scale of the latent outcome relative to this benchmark. This normalization idea is common in the measurement-error literature, where centering or scaling assumptions are used for identification, and in the SEM literature, where one loading or scaling parameter is often normalized to one. We generalize this idea to a fully nonparametric setting.

The remaining challenge is how to incorporate additional measurements without changing the estimand defined by the benchmark. For each additional measurement, we introduce a nonparametric bridge function $\phi$ that maps the measurement onto the benchmark scale in expectation conditional on the latent outcome. Let $\eta$ denote the latent outcome. Formally, for measurement $j \neq 1$, we seek a function $\phi_j$ satisfying
\(\E[\phi_j(Y_j)\mid\eta]=\E(Y_1\mid\eta).\)
In this way, heterogeneous measurements are first transformed so that they represent the same latent information, on the same benchmark scale, in expectation conditional on $\eta$, and only then are they combined for downstream causal analysis. We derive conditions under which such measurement bridge functions exist.

We then characterize the joint tangent space, derive the efficient influence function and semiparametric efficiency bound in the observed-data model defined by the bridge restrictions, and allow for different conditioning variables across bridges. Building on these results, we establish an exact decomposition of the variance reduction and derive necessary and sufficient conditions under which additional measurements yield a strict efficiency gain. In particular, neither bridge identification nor strong associations among measurements necessarily imply additional information about the target parameter. Finally, we construct an efficient estimator that attains the semiparametric efficiency bound under the stated local regularity, approximation, and nuisance-rate conditions. We refer to the resulting method as the Nonparametric Scaled Index (NSI).

This approach has several attractive features. First, it does not require researchers to specify a parametric latent-variable model. The relationship between each observed measurement and the latent outcome may be linear or nonlinear and need not be known. What matters is whether one can recover a bridge function that aligns the measurement to the benchmark. Second, by starting from the causal estimand rather than a dimension-reduction procedure, the framework provides a clear causal interpretation of the resulting estimates. Third, and most importantly, we characterize how to optimally combine multiple outcome measurements to achieve semiparametric efficiency when the target of interest is a single latent outcome. The framework also yields substantive guidance for empirical practice. It shows that measurement is not merely a nuisance that can be handled after the fact; it is part of the definition of the causal estimand. It therefore matters how researchers design measurements and which measurements they choose to share across studies. In particular, the framework highlights the value of designing measurements that are informative about the latent variable and of including at least one benchmark measurement that can anchor comparisons across settings. These design considerations are especially important when researchers hope to compare experimental findings across populations, interventions, or research teams.

We evaluate these claims through two simulation exercises. The first examines finite-sample bias, variance estimation, confidence interval coverage, and statistical power using five-category measurements and continuous measurements with nonlinear bridge functions. The second employs a paired nonlinear design that holds the latent treatment effect fixed while varying the auxiliary measurement system. This design allows us to assess whether NSI produces consistent treatment-effect estimates across different measurement systems.

We apply NSI to a randomized Hajj visa lottery to examine how winning access to the pilgrimage affects social attitudes and religious beliefs. The Hajj is one of the most important religious events in the lives of Muslims. \citet{clingingsmith2009estimating} measure several latent attitudes and beliefs using multiple survey items. We first show that when researchers use conventional methods such as PCA, equal weighting (as used in the original study), or ICW, causal effect estimates can change substantially---and can even reverse sign---when the set of measurements changes. NSI largely eliminates this sensitivity by placing different measurements on a common benchmark scale. We then apply NSI to estimate the causal effects of winning the visa lottery on several latent outcomes. We find more favorable views of people from other countries and greater acceptance of other ethnic and religious groups.

\paragraph{Related Literature.} We contribute to the literature on causal inference with latent outcomes. Although the statistical literature on experimental design and causal inference has expanded substantially in recent years, the challenges posed by imperfect measurement of latent outcomes have received relatively limited attention, even in otherwise comprehensive textbooks such as \citet{angrist2009mostly}, \citet{pearl2009causality}, \citet{gerber2012field}, \citet{imbens2015causal}, and \citet{hernan2020causal}. While the importance of outcome scales has been recognized in the literature \citep{bond2013evolution,cunha2010estimating,del2022identification,freyberger2026normalizations}, their implications for causal inference have received less systematic attention. \citet{stoetzer2022causal} develop a causal inference framework based on item response models, while \citet{fu2025causal} proposes a linear measurement causal model. However, both approaches rely on specific functional-form assumptions. In contrast, we focus on nonparametric identification and efficient estimation. More importantly, we highlight a fundamental comparability problem in existing practice: changing the indicators used to measure an outcome can alter the causal estimand, even when the underlying latent potential outcomes remain unchanged. More recently, \citet{guo2026proximal} also establish nonparametric identification using three mutually independent measurements conditional on latent outcome, treatment and covariates. To address the comparability problem, their approach requires either a known functional of a proxy distribution that equals the latent value or known support values and an ordering restriction. These normalization ideas are closely related to the literature on measurement-error models with auxiliary data and nonclassical measurement error \citep{chen2005measurement,hu2008instrumental}; \citet{schennach2022measurement} provides a broader review of the measurement-systems literature. When latent outcomes are constructed from unstructured data, the causal inference literature has primarily focused on addressing ``learning-induced interference'' \citep{landy2025causal,egami2022make}.

Multiple measurements have also been used to improve estimation efficiency and statistical testing. \citet{lubotsky2006interpretation} study combinations of proxies for an unobserved explanatory variable, while \citet{anderson2008multiple} develops summary-index tests for multiple experimental outcomes. However, neither approach directly addresses the challenge of preserving a common latent-outcome causal estimand across different sets of measurements. The surrogate-index approach of \citet{athey2026surrogate} is more closely related to our framework in its emphasis on a well-defined causal target. Their approach uses short-term proxies to recover a long-term outcome and studies both identification and efficiency, including potential precision gains when the primary outcome is observed. Our framework differs in that the benchmark itself is an imperfect measurement, and unknown bridge functions align different indicators in conditional expectation given the latent outcome.

Our identification strategy based on bridge functions is closely related to the literature on nonparametric instrumental variables (NPIV) \citep{newey2003instrumental}. \citet{santos2011instrumental} studies the identification of continuous linear functionals in this setting. The proximal causal inference literature further applies NPIV techniques to causal identification and estimation \citep{miao2018identifying,miao2024confounding,tchetgen2024introduction}. We extend these ideas to causal inference with latent outcomes, using bridge functions to link multiple imperfect measurements to a common causal estimand. 

Our efficiency analysis builds on the econometric literature on conditional moment restrictions. \citet{chamberlain1987asymptotic} establishes foundational efficiency bounds, while \citet{ai2003efficient} develop efficient estimation methods in models with unknown functions. \citet{ai2012semiparametric} characterize efficiency bounds under sequential moment restrictions with nested conditioning sets. In the context of NPIV models, \citet{severini2012efficiency} derive efficiency bounds for linear functionals of unknown functions. More generally, \citet{chen2018overidentification} study local overidentification in regular models, including conditional moment restrictions with potentially different conditioning sets, and examine its implications for efficiency gains in estimating smooth parameters. More recently, \citet{chen2025local} apply this perspective to causal inference and data combination, including negative-control models. Building on this literature, we characterize the joint local restrictions induced by multiple unknown measurement bridges and derive the corresponding semiparametric efficiency bound. Our estimation procedure is also related to inference with estimated nuisance functions \citep{ackerberg2012practical}, locally robust estimation \citep{chernozhukov2022locally}, and orthogonal estimation with cross-fitting \citep{chernozhukov2018double}. The estimation of bridge functions through regularized minimax methods is closely related to \citet{bennett2026inference}, who study inference on strongly identified functionals of weakly identified functions. Our approach differs in that the residual corrections are jointly selected through constrained variance minimization guided by the semiparametric efficiency bound.


\section{Framework}

Consider a sample of $n$ units drawn from a superpopulation. Let $Z_i$ denote the binary treatment assignment for unit $i$. The outcome of interest is latent: $\eta_i(1)$ is the latent outcome that would be realized if unit $i$ is assigned to receive the treatment, and $\eta_i(0)$ is the latent outcome that would be realized in the absence of treatment. We assume that the stable unit treatment value assumption (SUTVA) holds \citep{rubin1980randomization}.\footnote{The SUTVA assumption is a standard requirement for any study of cause and effect; in this context, it implies that potential outcomes remain stable regardless of which subjects receive treatment, ruling out, for example, spillovers between subjects. Formally, (1) $\eta_i(z_i,z_{-i})=\eta_i(z_i,z'_{-i}) \forall z_{-i}, z'_{-i}$; (2) If $Z_i=z$, then $\eta_i=\eta_i(z)$, $\forall i$ and $z \in \{0,1\}$.} These latent variables are not directly observed. Examples include ideology and state capacity in political science, preferences and human capital in economics, mental health or cognitive ability in psychology, and socioeconomic status and social capital in sociology. These constructs are conceptually important and constitute the true outcomes of interest.

In order to study them, researchers typically measure latent outcomes by treating them as unobserved constructs and linking them to observable measurements. Let $\mathcal J=\{1,2,...,J\}$ denote a universe of candidate measurements. For any finite set $\mathcal M\subseteq\mathcal J$, write $Y_{i\mathcal M}(z)=\{Y_{ij}(z):j\in\mathcal M\}$ to denote the set of measurements. We assume that measurements are generated by a treatment-invariant measurement system.

\begin{assumption}[Valid Measurement System]
    There is a conditional distribution $Q_{\mathcal M}(\cdot\mid a)$ such that for $z\in\{0,1\}$,
\begin{equation}
\mathcal L\{Y_{i\mathcal M}(z)\mid \eta_i(z)=a\}
=Q_{\mathcal M}(\cdot\mid a).
\label{eq:measurement-kernel}
\end{equation} where $\mathcal L(\cdot)$ denotes the distribution.
\end{assumption} Condition~\eqref{eq:measurement-kernel} rules out a direct effect of treatment on the measurement process after conditioning on the latent outcome. The condition can fail, for example, when treatment changes response styles, reporting incentives, question interpretation, or a second construct that is also captured by an indicator. If the distribution is degenerate, the measurement is error-free; otherwise, it is imperfect. We focus on the latter case. A simple and widely used classical specification is $Y_{ij}= \lambda_j \eta_i + \epsilon_{ij}$, where $\lambda_j$ is the factor loading that captures how the latent variable systematically affects the measurement, and $\epsilon_{ij}$ is the measurement error following a normal distribution $N(0,\sigma^2_{\epsilon_j})$. In this case $Y_{ij}|\eta_i \sim P_j(\cdot|\eta_i)=N(\lambda_j\eta_i,\sigma^2_{\epsilon_j})$. Another widely used model is item response theory: $P(Y_{ij}=1|\eta_i)=\frac{1}{1+exp(-a_j(\eta_i-b_j))}$, where $\eta_i$ is interpreted as latent ability, and the parameters $a_j$ and $b_j$ capture item discrimination and difficulty, respectively. In our framework, we do not impose a specific functional form; instead, the measurements can have any unknown relationship with the latent outcome.

We are interested in the causal effect of treatment on the latent variable, which \citet{stoetzer2022causal} and \citet{fu2025causal} call the latent treatment effect (LTE). The individual-level LTE is defined as $\tau_i=\eta_{i}(1)-\eta_{i}(0)$. Because it is impossible to observe both potential latent outcomes simultaneously, we focus on the average latent treatment effect (ALTE): $\tau=\mathbb{E}[\eta_{i}(1)-\eta_{i}(0)]$. For the theoretical analysis, we maintain the assumption that treatment is assigned independently of their potential
latent outcomes and potential measurements:
\(
 Z_i\perp\!\!\!\perp
 \{\eta_i(0),\eta_i(1),Y_{i\mathcal J}(0),Y_{i\mathcal J}(1)\}.
\)

Figure \ref{fig:dgp} illustrates the key elements of the framework. Treatment $Z_i$ affects the latent outcome $\eta_i$. The latent outcome can be written as $\eta_i= \mathbb{E}\eta_i(0) + \tau Z_i + \zeta_i$, where $\zeta_i=\eta_i(0) - \mathbb{E}\eta^0_i+Z_i[(\eta_i(1)-\mathbb{E}\eta_i(1))-(\eta_i(0)-\mathbb{E}\eta_i(0))]$. This term has mean zero by construction, $\mathbb{E}\zeta_i=0$. The measurements $Y_{ij}$ are connected to the same latent outcome. As long as the measurement is not a deterministic function of the latent outcome, we expect the presence of measurement error $\epsilon_{ij}$. In some settings, researchers collect multiple outcome measures but are interested in estimating the causal effect on each measure separately rather than on a common latent outcome. Such settings are beyond the scope of this paper. \citet{vanderweele2022statistical} develop methods for distinguishing between these two types of models. For ease of exposition, we omit $i$ when it does not create confusion. 

\begin{figure}[!h]
    \centering
    \includegraphics[width=0.8\linewidth]{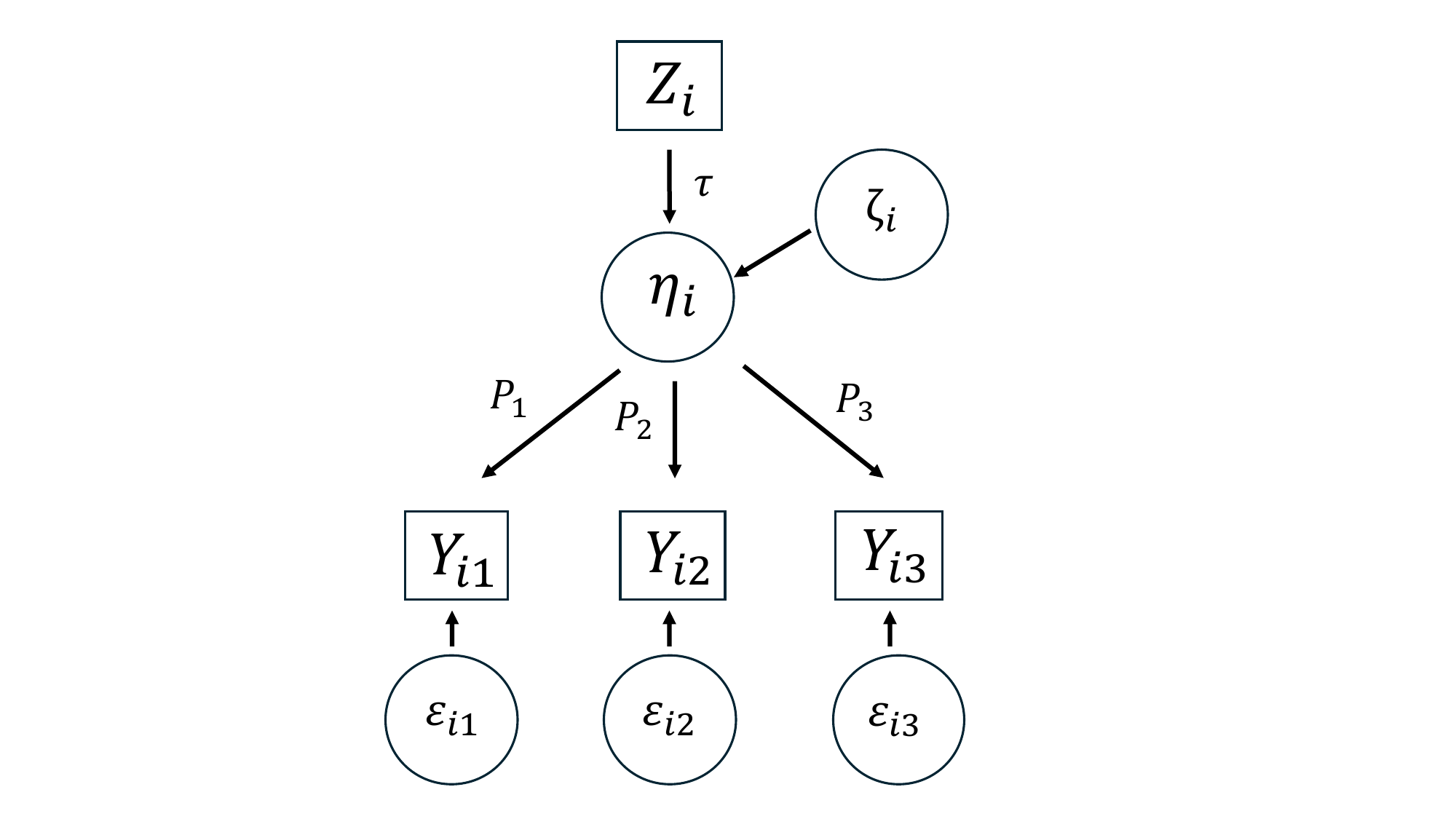}
    \caption{Graphical Depiction of an Experimental Design in which a Latent Outcome is Measured by Three Outcomes, Each Measured with Error.}
    \label{fig:dgp}
\end{figure}

\section{Illustration: Causal Effect of Winning Access to Pilgrimage}
\label{sec:hajj-running-example}

We use a real-world application to illustrate the framework and comparability challenges of causal inference on latent outcomes. Every year, around two million Muslim men and women from around the world gather in Mecca, Saudi Arabia, the holiest city in Islam, to perform the Hajj. The Hajj is an annual Islamic pilgrimage and a religious duty that every Muslim who is physically and financially able is required to undertake at least once in their lifetime. For many Muslims, it is among the most significant religious experiences of their lives. Against the backdrop of concerns about religious orthodoxy and extremism, an important question is whether and how participation in the Hajj causally affects individuals' attitudes, beliefs, and views.

To answer this question, \citet{clingingsmith2009estimating} exploit a randomized lottery used to allocate Pakistani visas for the January 2006 Hajj. Because the demand for Hajj visas greatly exceeds the available quota, many countries use lotteries to allocate visas for the pilgrimage. The original study uses lottery success as an instrument for Hajj attendance, whereas we use lottery assignment itself as the treatment. This intention-to-treat perspective is substantively meaningful for two reasons. First, winning the lottery may itself be perceived as a blessing or recognition from Allah and may therefore affect individuals' attitudes and beliefs. Second, winning the lottery changes individuals' expectations about the Hajj and induces preparations for the pilgrimage, which may also affect subsequent outcomes. 

For illustrative purposes, we focus on one latent outcome proposed in the original study: \emph{beliefs in localized Muslim practices}. This outcome captures individuals' endorsement of locally rooted religious beliefs and customs, and is measured using ten questions concerning religious authorities, rituals, and associated social practices. Table~\ref{tab:hajj-running-measurements} lists these measurements separately. 

\begin{table}[htbp]
\centering
\caption{Beliefs in localized Muslim practices and their survey measurements}
\label{tab:hajj-running-measurements}
\begin{threeparttable}
\renewcommand{\baselinestretch}{1}
\small
\begin{tabularx}{\linewidth}{@{}>{\raggedright\arraybackslash}p{0.26\linewidth}>{\raggedright\arraybackslash}X@{}}
\toprule
Latent outcome and direction &  Measurements \\
\midrule
\textbf{Beliefs in localized Muslim practices}\newline
10 items\par\smallskip
\textit{Higher values:} greater endorsement of localized beliefs and customs.
& \begin{tabular}[t]{@{}>{\raggedright\arraybackslash}p{\linewidth}@{}}
1. Whether a cap is required when praying salat (benchmark). \\[2pt]
2. Views concerning pirs (religious guides). \\[2pt]
3. Visiting genuine pirs to seek their advice. \\[2pt]
4. Visiting mazars (shrines). \\[2pt]
5. Doing dam or obtaining a ta'wiz. \\[2pt]
6. Observing chaliswaan after a death. \\[2pt]
7. Participating in maulud mehfils. \\[2pt]
8. Whether dowry is mandatory according to Islam. \\[2pt]
9. Priority for widows' remarriage relative to unmarried women. \\[2pt]
10. Whether someone can intercede with God on one's behalf.
\end{tabular} \\
\bottomrule
\end{tabularx}
\begin{tablenotes}[flushleft]
\footnotesize
\item \textit{Notes:} Each numbered entry is a separate measurement. The Supplementary Information (SI) \ref{si:app} gives exact wording,
coding, and response rates.
\end{tablenotes}
\end{threeparttable}
\end{table}

Let $\eta_i$ denote applicant $i$'s latent beliefs and $Y_{ij}$ the response to measurement $j$. The scientific target is the effect of winning a Hajj visa on $\eta_i$, whereas the observed data consist only of the $Y_{ij}$. In practice, the design and analysis involve considerable researcher discretion. First, researchers may add, remove, or replace individual items for a variety of reasons, such as new theoretical considerations, expert recommendations, or concerns about data quality. There is no strict theoretical basis for requiring this latent outcome to be measured using exactly these ten items. Thus, the realized questionnaire may contain all ten items, while an alternative questionnaire may use a different battery.

Second, a common practice is to perform dimension reduction by transforming or combining multiple measurements into a single summary score for subsequent causal analysis. However, the choice of how to aggregate these measurements involves considerable discretion. For example, principal component analysis (PCA) uses the first principal component of the standardized responses, while inverse covariance weighting (ICW) constructs a summary score using the inverse covariance matrix of the responses. The original study employs the average effect size (AES), which, in our sample, is algebraically equivalent to standardizing the individual measurements, averaging them into a single composite outcome, and estimating the treatment effect on that outcome. We now examine how these seemingly innocuous choices of measurement and aggregation methods can affect causal inference on the latent outcome.


Table~\ref{tab:hajj-running-beliefs} presents one selected pair of batteries for each method, together with a paired test of the change in its estimate. AES and ICW show substantial magnitude changes. AES is $-0.090$ for battery A and $-0.142$ for B, a 59\% increase in absolute magnitude. The paired difference B minus A is $-0.053$, with $p=0.047$.  Similarly, the ICW estimate is $-0.286$ using battery E and $-0.136$ using battery F, with the former more than twice as large in absolute magnitude as the latter. If the magnitude of the estimated effect is of primary interest, researchers may therefore reach substantially different conclusions depending on the set of measurements and aggregation method they use.

\input{tables/table_running_beliefs.tex}

In some cases, measurement choice can even reverse the sign of the estimated effect. Consider PCA. Battery C includes advice from pirs, maulud participation, dowry, and widows' remarriage (items 3, 7, 8, and 9). Removing the advice-from-pirs item yields battery D. The PCA estimate changes from $-0.144$, with a 90\% interval of $[-0.253,-0.036]$, to $0.112$, with an interval of $[0.014,0.211]$. Thus, if the direction of the effect is of primary interest, researchers may reach opposite conclusions depending solely on which set of measurements is used.

This application illustrates a fundamental problem with current practice in causal inference with multiple measurements. Although what we directly observe is a difference in the estimates, the underlying estimands are in fact different. A latent outcome is not directly observed and has no intrinsic metric; instead, its empirical representation depends on the observed measurements used to capture it. Consequently, a causal effect estimated from these measurements necessarily depends on how they are selected and combined. When the observed outcome is constructed as a function of multiple measurements, the resulting estimand is simply the causal effect of the treatment on that particular function of that particular set of measurements---for example, a score constructed using AES, PCA, or ICW. Without further assumptions, it is not a causal effect on the underlying latent outcome itself. The next section formalizes this problem and develops our proposed solution.

\section{Comparability Challenges of ALTE}

The preceding illustration shows the comparability problem. The key reason is that a latent outcome has no intrinsic metric, so any causal effect defined on it necessarily depends on the measurement scale used to represent it. Consequently, researchers may obtain substantially different and mutually incomparable causal estimands simply by adding, removing, or replacing even a single measurement, or by changing the method used to combine the measurements. More broadly, when different studies seek to estimate the average treatment effect on the same latent outcome but use different measurement instruments, their resulting estimates are generally not directly comparable. This lack of comparability poses a fundamental obstacle to the accumulation of scientific knowledge and, in particular, to valid meta-analysis.


This section formally analyzes the comparability problem. Consider a researcher who seeks to study the causal effect of $Z_i$ on the latent outcome $\eta_i$. Recall that the observed measurements are functions of the latent outcome. Suppose the researcher chooses a set of measurements $Y_{i\mathcal M_A}=(Y_{i1},Y_{i2},\ldots,Y_{iJ})$ to measure $\eta_i$. As discussed above, however, there are many reasons why the researcher might instead choose a different set of measurements to represent the same latent outcome. For example, $Y_{i\mathcal M_B}=(Y_{i1},Y_{i2},\ldots,Y_{i,J-1})$ omits the last measurement, whereas $Y_{i\mathcal M_C}=(Y_{i1},Y_{i2},\ldots,Y_{iJ},Y_{i,J+1})$ adds a new measurement.


Given multiple measurements, researchers typically apply dimension-reduction techniques, including simple averaging, PCA, ICW, IRT. These scoring procedures specify how data are used to construct individual scores. We study its population counterpart. For a population law $P$, write
\[
 s_{\mathcal M}(\,\cdot\,;P):
 \prod_{j\in\mathcal M}\mathcal Y_j\longrightarrow\mathbb R
\]
for the measurable scoring function produced by applying the procedure to
the observed-data distribution induced by $P$. The procedure is fixed, but
its population output may change with both $P$ and $\mathcal M$. Therefore, we obtain a single score for different sets of measurements: $\widetilde Y_{i\mathcal M}(z;P)=s_{\mathcal M}\{Y_{i\mathcal M}(z);P\}$. The resulting low-dimensional variables are then treated as outcomes in downstream causal inference. It is clear that the associated causal estimand is $\tau_{\mathcal M}=\E[\tilde{Y}_{i\mathcal{M}}(1)-\tilde{Y}_{i\mathcal{M}}(0)]$ and can be identified by random assignment. Let $h_{\mathcal M,P}(a)=\int s_{\mathcal M}(y;P)\,Q_{\mathcal M}(dy\mid a)$, where $a \in \mathcal A$. Under the treatment-invariant measurement system in \eqref{eq:measurement-kernel}, by iterated expectations,
$
\tau_{\mathcal M}(P)
=\E_P\left[
 h_{\mathcal M,P}\{\eta_i(1)\}
    -h_{\mathcal M,P}\{\eta_i(0)\}
\right].
$ Thus every scoring procedure identifies a causal effect on the transformed latent outcome $h_{\mathcal M,P}$. The question is whether this identified quantity is invariant to the choice of $\mathcal M$ and whether it equals the ALTE. 

Because latent outcome has no intrinsic scale, the second question is just an issue of normalization. However, for the first question, it is more important and challenging.

\begin{definition}[Measurement-set comparability]
\label{def:measurement-set-comparability}
Fix the scoring procedure and measurement kernel $Q_{\calM}$. The scoring family is \emph{measurement-set comparable} over
the admissible population laws if
\(
 \tau_{\mathcal M}(P)=\tau_{\mathcal M'}(P)\) for every admissible $P$ and every nonempty $\mathcal M,\mathcal M'\in \mathcal J$.
\end{definition}

Measurement-set comparability requires that adding, removing, or replacing measurements when moving from $\mathcal{M}$ to $\mathcal{M}'$ does not change the target estimand. Otherwise, even a slight modification to the measurement set may produce a completely different causal estimand and, consequently, different causal effect estimates. Although researchers intend to study the same latent outcome $\eta$ and estimate the same average latent treatment effect $\tau$, the actual estimand may differ across studies simply because different sets of measurements are used. As a result, observed differences in estimated treatment effects may reflect changes in measurement rather than genuine differences in the underlying causal effect.

To characterize comparability for procedures fitted using pooled measurements, let $F_{z,P}$ denote the distribution of $\eta_i(z)$ under
$P$, and define the pooled latent distribution
\(
 F_P=(1-\pi)F_{0,P}+\pi F_{1,P}.
\) where $\pi=P(Z_i=1)$.

\begin{proposition}\label{eq:set-comparability-characterization} Suppose that every fitted component of $s_{\mathcal M}(\,\cdot\,;P)$
depends on $P$ only through the pooled distribution of $Y_{i\mathcal M}$, for each nonempty $\mathcal M\in\mathcal J$.
Then the scoring family is measurement-set comparable if and only if, for every admissible $P$, there exist a measurable function $h_P$ and constants $c_{\mathcal M}(P)$ such that $F_P$-almost every $a$ and every nonempty $\mathcal M\in\mathcal J$
\begin{equation}\label{eq:comparability-common-mean}
h_{\mathcal M,P}(a)=h_P(a)+c_{\mathcal M}(P)
\end{equation} 
The condition in \eqref{eq:comparability-common-mean} is also sufficient
for comparability under other fitting schemes that apply the same fitted
scoring function to both treatment states.
\end{proposition}

Here, $h_P$ is common across measurement sets, and $c_{\mathcal M}(P)$ is constant with respect to $a$. An important observation is that commonly used scoring rules $s_{\mathcal M}$ does not necessarily satisfy the above properties. A large class of methods produces a linear score from centered or standardized measurements. Write
$S_{i\mathcal M}(P)
 =s_{\mathcal M}(Y_{i\mathcal M};P)
 =\sum_{j\in\mathcal M}a_{j,\mathcal M}(P)
   \frac{Y_{ij}-\mu_j(P)}{\sigma_j(P)}$, where $\mu_j(P)$
and $\sigma_j(P)$ are the centering and scaling constants selected by
the procedure. The coefficients $a_{j,\mathcal M}$ may be equal weights, PCA loadings, inverse-covariance weights, or weights learned by another procedure.\footnote{Raw-outcome linear scores are included by absorbing $\sigma_{j}^{-1}$ into the coefficients.}

The following corollary shows that, even in the special case in which the measurements are assumed to be linear functions of the latent outcome, these commonly used linear scoring rules do not satisfy measurement-set comparability.

\begin{corollary}[Measurement-set dependence of linear scores]
\label{cor:linear-score}
Suppose the conditional means of the measurements are linear on the fixed latent scale,
$
\E(Y_{ij}\mid\eta_i=a)=\alpha_j+\lambda_j a.
$
Then, 
\[
 h_{\mathcal M,P}(a)
 =c_{\mathcal M}(P)+\kappa_{\mathcal M}(P)a,
\quad  \text{where} \quad 
\kappa_{\mathcal M}(P)
 =\sum_{j\in\mathcal M}a_{j,\mathcal M}(P)
      \frac{\lambda_j}{\sigma_j(P)},
\] Consequently, linear scores are measurement-set comparable only when $\kappa_{\mathcal M}$ is the same across admissible measurement sets (if $\tau(P)\neq0)$.
\end{corollary}

This corollary highlights a fundamental limitation of
commonly used linear scoring procedures. The coefficient
$\kappa_{\mathcal M}$ generally depends on the selected
measurements and their aggregation weights, and these
procedures do not necessarily preserve a common value
of $\kappa_{\mathcal M}$ across measurement sets.
Consequently, the resulting causal estimands may differ
even when the underlying treatment effect on the latent
outcome remains unchanged. For an equal-weight standardized index, adding a measurement changes the average of the standardized sensitivities $\lambda_j/\sigma_{j}$ unless the new sensitivity happens to equal the existing average. For PCA, the weights are the leading eigenvector of a correlation or covariance matrix; changing the measurement set generally changes both the eigenvector and its normalization. For ICW, the weights are functions of an inverse covariance matrix and therefore change with the included measurements and the population covariance structure. The same conclusion applies to inverse-regression scores and learned linear representations whenever their population output is a data-dependent linear combination of the measurements. 


\textbf{Parametric Methods.} Other approaches, such as structural equation models (SEM) and item response theory (IRT) models, instead impose parametric assumptions on the relationship between observed measurements and the underlying latent outcome. For instance, a standard SEM specifies a linear measurement model $Y_j = \lambda_j \eta + \epsilon_j$, where the factor loadings $\lambda_j$ may differ across measurements. Similarly, a two-parameter logistic IRT model specifies $P(Y_{ij}=1\mid\eta_i)=\frac{1}{1+\exp[-a_j(\eta_i-b_j)]}$, where $a_j$ and $b_j$ denote item-specific discrimination and difficulty parameters, respectively. The goal is to recover these finite-dimensional parameters and thereby identify the latent outcome. A key concern, however, is that these models rely on strong functional-form assumptions and may therefore be sensitive to model misspecification. More importantly, in practice, these methods routinely normalize the measurements or latent scale for identification. As we show above, such normalizations make the causal estimand depend on the particular set of measurements used.

\subsection{Solution: Nonparametric Scaled Index (NSI)}

The root cause of this problem is that, without imposing strong modeling assumptions, latent variables have no intrinsic metric—their meaning depends on the measurements used to define them. Existing methods typically normalize and combine these measurements in an ad hoc manner, which in turn alters the interpretation of the latent outcome depending on the scaling and composition of the observed variables. To address the comparability challenge, the scoring rule $s_{\mathcal M}(\,\cdot\,;P)$ must preserve the common information shared across different sets of measurements, as required by Proposition \ref{eq:set-comparability-characterization}. This is a strong requirement, as it must hold for all possible measurement sets, while the latent outcome is, in principle, defined through multiple measurements. 

Rather than imposing a strong measurement model, our approach addresses this difficulty by anchoring the latent scale to a common benchmark measurement and using nonparametric measurement bridges to align the remaining measurements with that benchmark in conditional expectation. To be specific, we select measurement $Y_1$ as the benchmark. This variable serves as a bridge linking different sets of measurements. Instead of arbitrarily combining measurements, we require a bridge function that transforms different measurements so that they have the same metric as the benchmark measurement $Y_1$. Formally, we require a measurement bridge function $\phi_j$ for measurement $Y_j$, so that $\mathbb{E}[\phi_j(Y_j) \mid \eta] = \mathbb{E}[Y_1 \mid \eta]$. Thus, conditional on $\eta$, all transformed measurements have the same expectation as $Y_1$. 

Because the latent outcome is unobserved, its scale is not intrinsically identified. We therefore define the scale of the latent outcome using the benchmark measurement $Y_1$. Formally, we adopt the following normalization. 

\begin{assumption}[Benchmark-scale normalization]
\label{ass:normal}
For $z\in\{0,1\}$, the potential variables are integrable and
\(
\E\{Y_{i1}(z)\mid\eta_i(z)\}=\eta_i(z)
\text{ almost surely}.
\)
\end{assumption}


Although this assumption has the same form as the centering assumption commonly used in the measurement error literature \citep{vishnubhatla2026proxy}, it should be interpreted here as a normalization rather than a substantive restriction on the measurement process. It does not require researchers to identify a benchmark measurement that is inherently centered on the latent outcome. Nor does it imply that $Y_1$ is measured without error. Instead, because the latent outcome has no intrinsic location or scale, the latent outcome is interpreted on the same scale as $Y_1$, while $Y_1$ itself remains a noisy measurement of $\eta$. To motivate this normalization, let $\xi(z)$ denote an initial
parameterization of the same underlying construct. Suppose that the benchmark has a treatment-invariant conditional response
distribution and that its conditional mean function, $r_1(x)=\E\{Y_1(z)\mid\xi(z)=x\}$, is one-to-one. Define the transformation
$\eta(z)=r_1\{\xi(z)\}=\E\{Y_1(z)\mid\xi(z)\}$. Iterated expectations then imply $\E\{Y_1(z)\mid\eta(z)\}=\eta(z)$.\footnote{
This normalization does not imply that
$\eta_i(z)=\xi_i(z)$. Under the one-to-one condition, the transformation preserves the information in the original latent construct. Thus, the normalization is without loss of generality for the parameterization of the latent construct and is best understood as a
choice of scale rather than an additional restriction on its conditional mean. The resulting scale determines the numerical causal estimand, which need not be invariant to nonlinear transformations.}

The role of normalization is particularly transparent in the standard SEM framework, $
Y_j=\lambda_j\eta+\epsilon_j.$ Because the latent outcome \(\eta\) has no intrinsic scale, the model requires a normalization for identification. A common approach is to set \(\lambda_1=1\), thereby assigning the scale of the first measurement to the latent outcome. Within a given measurement model, this normalization is without loss of generality because the scale of \(\eta\) is otherwise arbitrary. Under this normalization, \(\eta\) is expressed on the scale determined by the first measurement.\footnote{An alternative normalization fixes the location and scale of the latent outcome, for example by imposing \(\mathbb{E}[\eta]=0\) and \(\operatorname{Var}(\eta)=1\). However, anchoring the latent scale to a particular measurement may create comparability problems when the set of measurements differs across studies.} Our assumption generalizes this familiar normalization principle beyond the parametric SEM framework.

The following flow chart \ref{fig:flow} summarizes our proposed method. First, we select a benchmark variable, denoted by $Y_{1}$. This measurement serves to link different sets of measurements. Next, we identify a measurement bridge function $\phi_j$ that transforms other measurements $Y_j$, $j \neq 1$, such that \begin{equation}\label{equ:meabridge}
    \mathbb{E}[Y_1 \mid \eta] = \mathbb{E}[\phi_j(Y_j) \mid \eta].
\end{equation}
This bridge function maps each measurement onto the reference measurement $Y_{1}$ so that all measurements share a common conditional mean on the benchmark scale. Because $Y_j$ and $Y_1$ are functions of the same latent variable $\eta$, it is natural to expect that such a bridge function exists; we will formalize this in the next section. If such a bridge function does not exist, it is generally impossible to extract the same information from different measurements.

\begin{figure}[!h]
    \centering
    \begin{tikzpicture}[
    node distance=4cm,
    every node/.style={font=\small},
    startstop/.style={
        rectangle, rounded corners,
        minimum width=3cm, minimum height=1cm,
        draw=black, fill=gray!15
    },
    process/.style={
  rectangle,
  fill=gray!15,
  draw,
  text width=3cm,
  align=center,
  minimum height=1cm
},
    arrow/.style={thick,->,>=stealth}
]

\node (start) [startstop] {Choose Benchmark $Y_{1}$};
\node (step1) [process, right of=start] {Estimate Bridge Function $\phi_j(Y_{j})$};
\node (step3) [process, right of=step1] {Downstream Causal Analysis};

\draw [arrow] (start) -- (step1);
\draw [arrow] (step1) -- (step3);
\end{tikzpicture}
    \caption{Estimation Procedure}
    \label{fig:flow}
\end{figure}

Because the latent outcome has no intrinsic scale or unit, we interpret it on the same scale as $Y_1$ under the normalization assumption. We then use the transformed measurements for subsequent causal analysis. Since the measurements are (N)onparametrically (S)caled to construct a comparable (I)ndex, we refer to this approach as NSI. Under this approach, all transformed measurements within a study share the same conditional mean on the benchmark scale and can therefore be jointly used for causal inference. The method satisfies measurement-set comparability because $h_{\mathcal M,P}(a)=a$ and $c_{\mathcal M}(P)=0$. Moreover, consider two studies, A and B, that target the same latent outcome $\eta$ but use different sets of measurements. As long as they share a common benchmark measurement, say $Y_1$, and satisfy the same normalization and bridge restrictions, the resulting transformed measurements $\phi_j(Y_j)$ are mean-comparable across studies.\footnote{Here, we implicitly assume measurement invariance. This assumption is plausible when the two studies target the same population, are conducted during similar time periods, and use identical wording for the measurement items.}

\section{Nonparametric Identification}

This section gives conditions for the existence and identification of measurement bridge functions.

\subsection{Existence of the Measurement Bridge Function}

Different measurements may have different relationships with the same latent outcome. A measurement bridge function is therefore required to recover a common representation of the latent outcome. Under what conditions does such a bridge function exist? 

It is useful to consider two simple numerical examples. Suppose the two measurements are linear functions of the latent variable: $Y_{i1}=\eta_i+\epsilon_{i1}$ and $Y_{i2}=\lambda_2 \eta_i+\epsilon_{i2}$, where $\lambda_2 \neq 0$, and the measurement errors satisfy $\E[\epsilon_{ij}\mid \eta_i]=0$ for $j=1,2$. In this case, the bridge function takes the simple form $\phi_2(Y_{i2})=\frac{1}{\lambda_2}Y_{i2}$. We can verify that $\E[Y_{i1}\mid \eta_i]=\eta_i=\E[\phi_2(Y_{i2})\mid \eta_i]$.

Now suppose instead that $Y_{i2}=\eta_i^2+e_{i2}$. This measurement loses information about the latent variable because it does not distinguish between $\eta_i=a$ and $\eta_i=-a$; it collapses positive and negative values into the same observed value. Since $\E[Y_{i1}\mid \eta_i]=\eta_i$, it is not possible to recover $\eta_i$ (and thus $\E[Y_{i1}\mid \eta_i]$) from $\eta_i^2$ through a bridge function without additional information about the sign.

It is convenient to approach this problem from a functional analysis perspective. Define the Hilbert spaces $H_1=L^2(F_{Y_{ij}})$ and $H_2=L^2(F_{\eta_i})$, with inner product $\langle h_1, h_2 \rangle = \mathbb{E}[h_1 h_2]$, which denote the spaces of all square-integrable functions with respect to the cumulative distribution function $F$, respectively. Let $K$ be a conditional expectation operator mapping $H_1$ to $H_2$, defined by $K\phi=\mathbb{E}[\phi(Y_{ij})|\eta_i]$, and define $r(\eta_i) = \mathbb{E}[Y_{i1} \mid \eta_i]$. The adjoint of the conditional expectation operator is $K^*$. Then the condition $\mathbb{E}[\phi_j(Y_{ij}) \mid \eta_i] = \mathbb{E}[Y_{i1} \mid \eta_i]$ can be written compactly as $K\phi = r$. This equation is known as a Fredholm integral equation of the first kind. 

Picard proves the necessary and sufficient conditions for the solution for the Fredholm integral equation of the first kind. We adapt the result in the following Proposition \ref{prop:picard}. Let $(\mu_n,\alpha_n,\beta_n)$ denote the singular system of $K$. 

\begin{proposition}[Picard's existence result]\label{prop:picard}
    Suppose regularity assumptions hold:
i. $\mathbb{E}[Y^2_{i1}] < \infty$. ii. $\iint f(y_{ij}|\eta_i)f(\eta_i|y_{ij})dy_{ij}d\eta_i<\infty$ iii. $\sum_{n=1}^\infty \frac{1}{\mu_n^2}|<r,\beta_n>|^2 < \infty$.

Measurement bridge function $\phi_j$ satisfying $\mathbb{E}[\phi_j(Y_{ij})|\eta_i]=\mathbb{E}[Y_{i1}|\eta_i]$ exists if and only if 
 $r$ belongs to the orthogonal complement $\mathcal{N}(K^*)^\perp$.
\end{proposition}

As shown in the proof, the null space
$\mathcal N(K_j^*)$ consists of functions of the latent outcome whose conditional expectations given $Y_{ij}$ are zero. Orthogonality requires only that the benchmark regression $r_1$ contain no component
in these latent directions that are invisible to $Y_{ij}$. This assumption is theoretically necessary. 

There exists a familiar condition that provides a stronger but convenient sufficient condition for this requirement. The key observation from the example at the beginning is that the measurements do not contain sufficient information about the latent outcome. Formally, this is captured by the completeness assumption, which also appears in the nonparametric instrumental variables \citep{newey2003instrumental} and proximal causal inference literature \citep{miao2018identifying}.

\begin{assumption}[Completeness]\label{ass:complete}
    For any square-integrable function $g$ and for any $Y_{ij}$, $j\neq1,$
    $\mathbb{E}[g(\eta_i)|Y_{ij}]=0$ almost surely if and only if $g(\eta_i)=0$ almost surely.
\end{assumption}

This assumption first implies that $\eta_i$ and $Y_{ij}$ are not independent. In other words, the observed measurements indeed capture some aspects of the latent outcome. This is fundamentally a research design problem. If the condition is not satisfied, researchers should use theory to develop valid measures of the construct of interest. How to design informative measurement devices is itself an important topic; however, in this study, we primarily focus on the downstream question of how to identify and estimate causal effects on latent outcomes. Second, it requires that the measurement $Y_{ij}$ captures the full information (variability) of the latent outcome $\eta_i$. It accommodates both categorical and continuous variables. For example, if $\eta_i$ is discrete with $k$ categories and $Y_{ij}$ has $m$ categories, then we must have $m \ge k$. To see this, 
$$
\underbrace{
\begin{bmatrix}
    P(\eta^1_1|y^1_{ij}) &...& P(\eta^k_1|y^1_{ij})\\
    \vdots & &\vdots\\\
    P(\eta^1_1|y^m_{ij}) &...& P(\eta^k_1|y^m_{ij})
\end{bmatrix}}_{A}
\underbrace{
\begin{bmatrix}
    g(\eta^1_1)\\
     \vdots \\
     g(\eta^k_1)
\end{bmatrix}}_{g}
=
\begin{bmatrix}
    \mathbb{E}[g(\eta_i)|y^1_{ij}]\\
     \vdots \\
    \mathbb{E}[g(\eta_i)|y^m_{ij}]
\end{bmatrix}
$$
Completeness ($Ag=0 \Rightarrow g=0$) requires that the matrix $A$ has full rank, which can only hold if $m \ge k$. Many models, such as those in the exponential family, satisfy completeness.

Completeness implies that $\mathcal{N}(K^*_j)=\{0\}$ and, consequently, that $\mathcal{N}(K^*_j)^\perp$ is the entire Hilbert space to which $r$ belongs. We therefore obtain the following existence result, adapted from \citet{miao2018identifying}.

\begin{corollary}[Existence]\label{prop:exist}
    Suppose regularity assumptions in Proposition \ref{prop:picard} hold.
     Under assumption \ref{ass:complete}, there exists a measurement bridge function  $\phi_j$ such that $\mathbb{E}[\phi_j(Y_{ij})|\eta_i]=\mathbb{E}[Y_{i1}|\eta_i]$.
\end{corollary}

Completeness, together with the regularity conditions in
Proposition~\ref{prop:picard}, is sufficient but not necessary for
the existence of a measurement bridge. A simple alternative
sufficient condition for bridge existence is a nondegenerate
linear conditional-mean relationship with the latent outcome.
Specifically, under the benchmark normalization in
Assumption~\ref{ass:normal}, if
$\E(Y_j\mid\eta)=a_j+b_j\eta$, with $b_j\neq0$ and
$\E(Y_j^2)<\infty$, then
$\phi_j(y)=(y-a_j)/b_j$ is a square-integrable measurement
bridge, since
\(\E\{\phi_j(Y_j)\mid\eta\}
=
\eta
=
\E(Y_1\mid\eta).
\) In practice, we recommend designing measurements that plausibly exhibit a linear relationship with the underlying latent variable. A categorical indicator may also admit a bridge function for a particular benchmark even when completeness fails. 

\subsection{Identification}

Given that the measurement bridge function exists, we must ensure that it is identified. As in the proximal causal inference literature \citep{miao2024confounding}, identifying $\phi_j$ requires auxiliary information. In general, we need additional variables that serve as instrumental variables; then $\phi_j$ can be identified through a standard nonparametric instrumental variables (NPIV) approach.

\begin{proposition}[Identification]\label{prop:identify}
Suppose measurement bridge function $\phi$ exists. For any variable $W_i$, if the following two conditions hold,

(1) (Mean Independence) $\mathbb{E}[Y_{i1}|\eta_i]=\mathbb{E}[Y_{i1}|W_i,\eta_i]$ and $\mathbb{E}[\phi(Y_{ij})|\eta_i]=\mathbb{E}[\phi(Y_{ij})|W_i,\eta_i]$;

(2) (Completeness) For any square-integrable function $g$ and $j\in \{2,...J\}$, $\mathbb{E}[g(Y_{ij})|W_i]=0$ for almost all $W_i$ implies $g(Y_{ij})=0$ almost surely 

Then, $\phi$ is uniquely identified through NPIV $\mathbb{E}[\phi_j(Y_{ij})|W_i]=\mathbb{E}[Y_{i1}|W_i].
$
\end{proposition}

In our context, candidate instrumental variables include treatment assignment $Z_i$, measurements other than $Y_{i1}$ and $Y_{ij}$, and pre-treatment covariates $X_i$. We emphasize that they are not automatically valid; rather, their validity requires that the two conditions stated in the Proposition \ref{prop:identify} be satisfied. The first condition requires mean independence conditional on the latent outcome. One sufficient (though not necessary) condition is that the instrumental variables are independent of the measurement errors in both $Y_{i1}$ and $Y_{ij}$ conditional on $\eta_i$. Given that treatment assignment $Z_i$ is the natural choice of instrument, this implies an exclusion restriction: the treatment should affect the indicators only through the latent outcome. For example, in survey experiments, the treatment should not alter respondents’ response styles, the salience of particular items, or item-specific attitudes after conditioning on the proposed latent outcome. 

The second requirement, completeness, is the nonparametric analogue of the rank condition in the parametric setting \citep{darolles2011nonparametric,horowitz2011applied,newey2003instrumental}. If $W_i$ and $Y_{ij}$ are discrete with finite support, then the support of $W_i$ must be at least as rich as that of $Y_{ij}$. Once $\phi_j$ is identified, we obtain transformed measurements $\tilde{Y}_j = \phi_j(Y_{ij})$. In this case, we have multiple outcome variables, so that the ALTE is generally over-identified. 

\section{Semiparametric efficiency and estimation}\label{sec:measurement-efficiency}

This section characterizes the semiparametric efficiency bound and develops an efficient estimator. In particular, we establish when and how additional measurements can improve estimation efficiency. The analysis builds on conditional-moment efficiency theory \citep{chamberlain1987asymptotic,ai2003efficient,ai2012semiparametric}, NPIV efficiency analysis \citep{severini2012efficiency}, and local overidentification with different conditioning sets
\citep{chen2018overidentification}. 

\subsection{Setup}\label{subsec:setup}

Let $O=(Z,Y_1,\ldots,Y_J)$ be the observed data for one individual, and suppose $O_1,\ldots,O_n$ are independent and identically distributed. Treatment $Z\in\{0,1\}$ is randomized with known probability $\pi=P(Z=1)\in(0,1)$. The true observed-data law is denoted by $P_0$. Define the known treatment weight
\begin{equation}\label{eq:weight}
a(Z)=\frac Z\pi-\frac{1-Z}{1-\pi}.
\end{equation}
For any integrable observed outcome $V$, $\E[a(Z)V]=\E[V\mid Z=1]-\E[V\mid Z=0]$. 
Under the benchmark normalization $\E[Y_1(z)\mid\eta(z)]=\eta(z)$ and random assignment, the target is
\begin{equation}\label{eq:target-benchmark}
\tau_0=\E[\eta(1)-\eta(0)]
=\E[Y_1\mid Z=1]-\E[Y_1\mid Z=0]
=\E[a(Z)Y_1].
\end{equation}

For each auxiliary measurement $j=2,\ldots,J$, choose a valid conditioning vector $W_j$ containing $Z$. The vectors are prespecified and may differ across bridges. For example, with three measurements one may use $W_2=(Z,Y_3)$ and $W_3=(Z,Y_2)$ when the identifying assumptions justify both. Let $u_j=Y_1-\phi_{0j}(Y_j)$. The bridge restrictions are for $j=2,\ldots,J$,
\(\E[u_j\mid W_j]=0\). Throughout this section the bridge functions are unrestricted elements of $L^2(P_{Y_j})$.\footnote{For a specified closed linear bridge space $\Phi_j\subseteq L^2(P_{Y_j})$, use the restricted operator
$T_j:\Phi_j\to L^2(P_{W_j})$. Its adjoint is $T_j^*v=\Pi_{\Phi_j}\E[v(W_j)\mid Y_j]$, and the admissible null restriction is $\Pi_{\Phi_j}\E[v(W_j)\mid Y_j]=0$.} Because $a(Z)$ is measurable with respect to every $W_j$,
\begin{equation}\label{eq:common-target}
 \E[a(Z)\phi_{0j}(Y_j)]
 =\E[a(Z)Y_1]-\E[a(Z)\E[u_j\mid W_j]]
 =\tau_0.
\end{equation}
Thus each calibrated measurement provides another representation of the same contrast. 

We summarize the conditional expectation operators used in the subsequent derivations. For square-integrable functions, write $\ip f g=\E[fg]$ and $\norm f_2^2=\E[f^2]$. For a linear operator $A$, write $\Ran(A)$ for its range and $\Null(A)$ for its null space. The conditional-expectation operator for bridge $j$ is $ T_j:L^2(P_{Y_j})\longrightarrow L^2(P_{W_j})$,$(T_jh)(w)=\E[h(Y_j)\mid W_j=w]$. The bridge equation can be written $T_j\phi_{0j}=\E[Y_1\mid W_j]$.\footnote{Conditional Jensen's inequality gives $\norm{T_jh}_2\leq\norm h_2$, so $T_j$ is continuous even when recovering $h$ from $T_jh$ is ill-posed.} Its adjoint maps in the opposite direction and is defined by the identity $\ip{T_jh}v=\ip h{T_j^*v}$. In the full square-integrable spaces used here, $T_j^*:L^2(P_{W_j})\longrightarrow L^2(P_{Y_j}), T_j^*v=\E[v(W_j)\mid Y_j]$. Indeed, iterated expectation gives
\(
 \ip{T_jh}v
 =\E[v(W_j)h(Y_j)]
 =\ip h{\E[v(W_j)\mid Y_j]}.
\)

\subsection{Scores and tangent spaces}\label{subsec:tangent-definition}
Let $\calM_0$ be the dominated observed-data model with fixed
$P(Z=1)=\pi$ and otherwise unrestricted conditional outcome laws. For a regular one-dimensional submodel $\{P_t:|t|<\epsilon\}\subseteq\calM_0$ through $P_0$, let
$s(O)=\left.\partial_t\log p_t(O)\right|_{t=0}$ denote its likelihood score, when this derivative exists.\footnote{A more general and formal treatment is provided in Appendix~\ref{app:proofs}.} Because the treatment
probabilities are fixed, scores satisfy $\E[s\mid Z]=0$. It is now well known that the ambient
tangent space, obtained by taking the closed linear span of these scores, is
\begin{equation}\label{eq:ambient}
 \Hpi=\{s\in L^2(P_0):\E[s\mid Z]=0\}.
\end{equation}
Write $\mu(z)=\E[Y_1\mid Z=z]$. The derivative of the target along a
score $s\in\Hpi$ is
\(
 \dot\tau(s)=\E[a(Z)Y_1s]
 =\E\!\left[a(Z)\{Y_1-\mu(Z)\}s\right].
\) Thus the benchmark canonical
gradient is $\psib=a(Z)\{Y_1-\mu(Z)\}$, and its variance is
\(\Vb=\E[\psib^2]
 =\frac{\Var(Y_1\mid Z=1)}{\pi}
  +\frac{\Var(Y_1\mid Z=0)}{1-\pi}.
\)

Let $\calM\subseteq\calM_0$ denote the model that additionally imposes all bridge equations for unknown square-integrable bridge functions. We now characterize the score and tangent sets of $\calM$. Along a path with an $L^2(P_{Y_j})$ bridge derivative, write $\phi_{jt}=\phi_{0j}+t \dot\phi_j+o_{L^2}(t)$. For any fixed bounded test function $f_j(W_j)$, the conditional moment implies
\(
 \E_t[f_j(W_j)\{Y_1-\phi_{jt}(Y_j)\}]=0.
\)
Differentiating its integrand and its distribution gives
\(
 0=\E[f_j(W_j)u_js]-\E[f_j(W_j)\dot\phi_j(Y_j)].
\)
Since the equality holds for every bounded test function $f_j(W_j)$,
\begin{equation}\label{eq:linearized}
\E[u_js\mid W_j]=T_j\dot\phi_j,\qquad j=2,\ldots,J.
\end{equation}
The differentiated bridge equations tell us how a small change in the observed distribution must be offset by changes in the bridges. Appendix~\ref{app:density} derives the same result using density. Define the linearized score set
\begin{equation}\label{eq:exact-score-set}
 \mathcal S_{\mathrm{lin}}=
 \left\{s\in\Hpi:
 \begin{array}{l}
 \text{there exist }h_j\in L^2(P_{Y_j})\text{ such that}\\[-2pt]
 \E[u_js\mid W_j]=T_jh_j\text{ for all }j=2,\ldots,J
 \end{array}\right\}.
\end{equation}

Satisfying these equations, however, does not by itself show that the original bridge restrictions can continue to hold along a path of nearby distributions. We therefore need a condition ensuring that the solutions to the differentiated equations describe the local changes actually permitted by the model. There is a further difficulty when taking limits of these directions.
A conditional-moment change may be approximated arbitrarily closely by changes in a bridge even though no single square-integrable bridge change matches it exactly. To include such directions in the tangent space, we need to show that these approximate matches can be made exact by an arbitrarily small adjustment to the score. With several bridges, the adjustment must work for all equations simultaneously. We impose the
following regularity conditions.

\begin{assumption}[Regularity of the local bridge model]\label{ass:tangent}
The following conditions hold at $P_0$.
\begin{enumerate}
\item \emph{Moment continuity.} $\E[Y_1^2]<\infty$ and $\E[u_j^2\mid W_j]\leq C_0<\infty$ for every $j$.
\item \emph{Local path regularity.} The tangent space is generated by regular paths admitting the $L^2$ bridge derivatives used above, with differentiation under the relevant expectations valid. The closed span of their scores equals $\overline{\mathcal S_{\mathrm{lin}}}$.
\item \emph{Simultaneous closure.} Either all $\Ran(T_j)$ are closed, or the joint nondegeneracy condition \eqref{eq:joint-coercivity} in Appendix~\ref{app:joint} holds. 
\end{enumerate}
\end{assumption}

Part (i) ensures that a small mean-square change in the score produces only a small change in each conditional moment:
\( \norm{\E[u_j(s_1-s_2)\mid W_j]}_2
 \leq \sqrt{C_0}\,\norm{s_1-s_2}_2. \) Part (ii) is a local path regularity condition linking the bridge equations to the tangent space of the observable model. It requires regular submodels along which the bridges are differentiable in \(L^2\) to generate the tangent space, with their closed score span characterized by the linearized bridge equations. This formulation parallels the smooth-path framework of \citet{chen2018overidentification} for conditional moment models. Because the condition concerns closed spans, it suffices that compatible directions can be approximated by linear combinations of scores from such smooth submodels. In the finite-support setting, Appendix~\ref{app:finite} verifies the local-path condition for categorical measurements under positive-cell and rank conditions.
 Part (iii) allows approximate conditional-moment matches to be replaced by exact ones without materially changing the score. For a single bridge, \citet{severini2012efficiency} achieve this by adding a residual-weighted adjustment to the candidate score; a positive lower bound on the conditional residual variance ensures that a small moment discrepancy requires only a small score adjustment. Our joint residual condition permits the same construction for all bridge equations simultaneously, since correcting one equation separately can otherwise disturb another.
The closed-range alternative in part (iii) holds automatically for
finite-support measurements.


\begin{proposition}[Joint tangent space]\label{prop:tangent}
Under Assumption~\ref{ass:tangent}, the tangent space is
\begin{equation}\label{eq:tangent-characterization}
\calT=\{s\in\Hpi:\E[u_js\mid W_j]\in\overline{\Ran(T_j)}
 \text{ for every }j\}.
\end{equation}
\end{proposition}

The proof first uses part (ii) to identify $\calT$ with $\overline{\mathcal S_{\mathrm{lin}}}$, then applies the joint closure lemma using parts (i) and (iii). Next, we characterize the orthogonal complement of the tangent space, which plays a key role in deriving the efficiency bound. Define \(\calN_j=\Null(T_j^*)
 =\{v_j\in L^2(P_{W_j}):\E[v_j(W_j)\mid Y_j]=0\}.\) Its elements are weighting functions orthogonal to every square-integrable change in the $j$th bridge:
\(
 v_j\in\calN_j
 \text{ iff }
 \E[v_j(W_j)h(Y_j)]=0 \text{ for every }h\in L^2(P_{Y_j}).
\)
The identity between the range and adjoint null space gives $\overline{\Ran(T_j)}=\calN_j^\perp$. Consequently, for $s\in\Hpi$,
\begin{align*}
 \E[u_js\mid W_j]\in\overline{\Ran(T_j)}
 &\Longleftrightarrow
 \E[v_j(W_j)\E[u_js\mid W_j]]=0\quad\forall v_j\in\calN_j\\
 &\Longleftrightarrow
 \E[s\,u_jv_j(W_j)]=0\quad\forall v_j\in\calN_j.
\end{align*}
Thus each admissible residual direction $u_jv_j(W_j)$ is orthogonal to every permitted likelihood score.

Collect their sums and mean-square limits in the closed space
\begin{equation}\label{eq:residual-space}
 \calV=\overline{\left\{\sum_{j=2}^J u_jv_j(W_j):
 v_j\in\calN_j\text{ for all }j\right\}}^{\,L^2(P_0)}.
\end{equation}
The moment bound ensures the residual products are square-integrable. Also,
\(
 \E[u_jv_j(W_j)\mid Z]
 =\E[v_j(W_j)\E[u_j\mid W_j]\mid Z]=0,
\) so $\calV\subseteq\Hpi$.

\begin{proposition}\label{prop:corrections}
Under Assumption~\ref{ass:tangent},
\(\calT=\Hpi\cap\calV^\perp,\text{ and }
 \calT^{\perp_{\Hpi}}=\calV.\)
All derivative representers that lie in $\Hpi$ form the affine space $\psib+\calV$. In particular,
\begin{equation}\label{eq:single-gradient}
 \psib+u_jv_j(W_j),\text{ where } v_j\in\calN_j,
\end{equation}
is an influence function.
\end{proposition}

Here $\perp_{\Hpi}$ denotes the orthogonal complement inside the known-randomization score space. In the larger centered space $L_0^2(P_0)=\{f:\E[f]=0,\E[f^2]<\infty\}$, the complement also contains centered functions of treatment: $\calT^\perp=L_0^2(P_Z)\oplus\calV$, an orthogonal direct sum. These treatment-only directions have already been removed from $\psib$. The proposition follows from the displayed range-null equivalence and closure. 


\subsection{Efficiency bound}\label{subsec:efficiency}

For a closed linear subspace $\mathcal U$, let $\Pi_{\mathcal U}f$ denote the orthogonal projection of $f$ onto $\mathcal U$.
Define
\(
 g_*=\Pi_{\calV}\psib.
\)
The random variable $g_*$ is the part of benchmark sampling error that can be represented by admissible bridge-residual corrections.

\begin{theorem}[Efficient influence function and exact variance reduction]\label{thm:efficiency}
Under Assumption~\ref{ass:tangent}, the canonical gradient (efficient influence function) in $\calM$ is
\(
\psieff=\Pi_{\calT}\psib=\psib-g_*.\) Its variance is
\begin{equation}\label{eq:bound}
 \begin{split}
 \Ve&=\Vb-\norm{g_*}_2^2\\
 &=\inf_{v_j\in\calN_j,\ j=2,\ldots,J}
 \E\!\left[\left\{\psib+\sum_{j=2}^J u_jv_j(W_j)\right\}^2\right].
 \end{split}
\end{equation}
For every admissible tuple $(v_2,\ldots,v_J)$,
\begin{equation}\label{eq:excess-risk}
 \E\!\left[\left\{\psib+\sum_{j=2}^J u_jv_j(W_j)\right\}^2\right]-\Ve
 =\left\|g_*+\sum_{j=2}^J u_jv_j(W_j)\right\|_2^2.
\end{equation}
The infimum is attained by square-integrable coefficient functions $v_{*j}$ exactly when their residual sum equals $-g_*$. The optimal correction is unique even when its coefficient representation is not.
\end{theorem}

Proposition~\ref{prop:corrections} gives the orthogonal decomposition $\Hpi=\calT\oplus\calV$. Hence $\psib-g_*\in\calT$. Moreover, for every $s\in\calT$,
\(
 \E[(\psib-g_*)s]=\E[\psib s]=\dot\tau(s),
\) because $g_*\perp\calT$. The proposed function therefore both represents the derivative and belongs to the tangent space. These two properties establish that it is the canonical gradient. Theorem \ref{thm:efficiency} characterizes the efficiency bound. The term $\sum_{j=2}^J u_jv_j(W_j)$ can be interpreted as a residual correction. The goal is to find the optimal correction that minimizes the variance of the influence function.



Suppose an attaining tuple $v_*=(v_{*2},\ldots,v_{*J})$ exists, and write
\[
 \psi_{v_*}=\psib+\sum_{k=2}^J u_kv_{*k}(W_k).
\] We now characterize the functions $v_{*j}(W_j)$ that attain the minimum. A feasible perturbation of the $j$th function is $v_{*j}+t w_j$, where $w_j\in\calN_j$. Differentiating the variance gives
\[
 \left.\partial_t\E[(\psi_{v_*}+t u_jw_j(W_j))^2]\right|_0
 =2\E[\psi_{v_*}u_jw_j(W_j)].
\]
It is zero for every feasible direction exactly when
\begin{equation}\label{eq:normal-equations}
 \Pi_{\calN_j}\E\!\left[
 u_j\left\{\psib+\sum_{k=2}^J u_kv_{*k}(W_k)\right\}
 \,\middle|\,W_j\right]=0,
 \qquad j=2,\ldots,J.
\end{equation}
Here $\Pi_{\calN_j}$ is ordinary $L^2(P_{W_j})$ projection. The equations also verify efficiency directly. Since $\calN_j^\perp=\overline{\Ran(T_j)}$, these equations say that $\E[u_j\psi_{v_*}\mid W_j]$ belongs to the required range closure for every $j$. Together with $\psi_{v_*}\in\Hpi$ and Proposition~\ref{prop:tangent}, solving the variance equations also verifies that the candidate is the EIF. 

To interpret this condition, define the conditional residual variance $\omega_j(W_j)=\E[u_j^2\mid W_j]$ and the conditional covariance $c_j(W_j)=\E[u_j\psib\mid W_j]$, which measures the conditional association between the bridge discrepancy and the benchmark sampling error. Expanding \eqref{eq:normal-equations} gives
\begin{equation}\label{eq:cross-normal}
 \Pi_{\calN_j}\left[
 c_j+\omega_jv_{*j}
 +\sum_{k\ne j}\E[u_ju_kv_{*k}(W_k)\mid W_j]
 \right]=0.
\end{equation}
The cross-bridge terms generally do not vanish; therefore optimizing each bridge separately generally fails to give the joint optimum. A useful population criterion is
\begin{equation}\label{eq:population-ridge}
 v_\lambda=\arg\min_{v_j\in\calN_j}
 \left\{\E\!\left[\left(\psib+\sum_{j=2}^J u_jv_j(W_j)\right)^2\right]
 +\lambda\sum_{j=2}^J\norm{v_j}_2^2\right\},\qquad\lambda>0.
\end{equation} where the penalty ensures a unique coefficient vector.





It is clear that the exact gain from additional measurements is the part of benchmark variance explained by the admissible correction space. An auxiliary measurement can be strongly correlated with the benchmark yet supply no such direction after calibration uncertainty is taken into account.

\begin{corollary}\label{cor:gain}
Under Assumption~\ref{ass:tangent},
\begin{equation}\label{eq:no-gain}
 \Ve=\Vb
 \ \Longleftrightarrow\ \psib\perp\calV
 \ \Longleftrightarrow\
 \E[\psib u_jv_j(W_j)]=0\quad\forall j,\ v_j\in\calN_j.
\end{equation}
Equivalently, there is no gain exactly when $c_j\in\overline{\Ran(T_j)}$ for every $j$. A strict gain occurs when at least one admissible residual direction has nonzero covariance with $\psib$.

\end{corollary}

\subsection{Estimator}\label{subsec:estimator}

For the estimator, the efficient influence function suggests starting from the benchmark and adding learned bridge-residual corrections. To preserve exact benchmark unbiasedness, let $\widetilde\mu$ be any function on $\{0,1\}$ and define
\(
 m_{\mathrm b}(O;\widetilde\mu)
 =\widetilde\mu(1)-\widetilde\mu(0)
  +a(Z)\{Y_1-\widetilde\mu(Z)\}.
\)
Its expectation is $\tau_0$ for every $\widetilde\mu$, because
\(
 \E[a(Z)\{Y_1-\widetilde\mu(Z)\}]
 =\mu(1)-\widetilde\mu(1)-\{\mu(0)-\widetilde\mu(0)\}.
\)
At $\widetilde\mu=\mu$, it equals $\tau_0+\psib$.

For candidate bridges $\phi_j\in L^2(P_{Y_j})$ and weights $v_j\in L^2(P_{W_j})$, write $\phi=(\phi_2,\ldots,\phi_J)$ and $v=(v_2,\ldots,v_J)$, and set
\begin{equation}\label{eq:contribution}
 m(O;\widetilde\mu,\phi,v)
 =m_{\mathrm b}(O;\widetilde\mu)
  +\sum_{j=2}^J v_j(W_j)\{Y_1-\phi_j(Y_j)\}.
\end{equation}
The centered estimating function is $\psi(O;\tau,\widetilde\mu,\phi,v)=m(O;\widetilde\mu,\phi,v)-\tau$. At the true nuisance functions and an optimal admissible
correction, this estimating function coincides with the
efficient influence function. For fixed candidate functions with the required integrability, the estimating equation has the exact bias identity
\begin{equation}\label{eq:bias-exact}
 \begin{split}
 \E[m(O;\widetilde\mu,\phi,v)]-\tau_0
 &=-\sum_{j=2}^J\ip{T_j^*v_j}{\phi_j-\phi_{0j}}\\
 &=-\sum_{j=2}^J\ip{v_j}{T_j(\phi_j-\phi_{0j})}.
 \end{split}
\end{equation}
To see this, write $Y_1-\phi_j(Y_j)=u_j-\{\phi_j(Y_j)-\phi_{0j}(Y_j)\}$. The expectation of $v_ju_j$ is zero by the true bridge equation, and iterated expectation gives the displayed inner products. The benchmark contribution has exactly zero bias for every $\widetilde\mu$. This establishes a double robustness property: the bias vanishes when each bridge is correct or each corresponding correction satisfies its population null restriction. Moreover, the bias is second-order around \((\phi_0,v^\circ)\), where \(v_j^\circ\in\mathcal N_j\), implying Neyman orthogonality with respect to both nuisance functions \citep{chernozhukov2018double}.


We estimate by solving the empirical estimating
equation with cross-fitted nuisance estimates. Partition the sample independently of the observations into a fixed number $K$ of folds $I_1,\ldots,I_K$. For each fold $k$, estimate $\mu$, the bridges, and the correction functions using observations outside $I_k$. For $i\in I_k$, evaluate
\(
 \widehat m_i=m(O_i;\widehat\mu^{(-k)},\widehat\phi^{(-k)},\widehat v^{(-k)}),
\)
and form
\(\widehat\tau=\frac1n\sum_{k=1}^K\sum_{i\in I_k}\widehat m_i.\)
This solves the empirical equation based on $\psi=m-\tau$. Since that equation is affine in $\tau$ with derivative $-1$, it also equals a one-step estimator with the fitted nuisance functions held fixed. 

To estimate the bridge, note that for any candidate $\phi_j\in L^2(P_{Y_j})$,
\begin{equation}\label{eq:minimax}
 \sup_{f\in L^2(P_{W_j})}
 \E\!\left[f(W_j)\{\phi_j(Y_j)-Y_1\}-\tfrac12 f(W_j)^2\right]
 =\tfrac12\norm{T_j(\phi_j-\phi_{0j})}_2^2.
\end{equation}
The maximizing test function $f$ is the conditional mean of the candidate residual. Thus minimizing the left side fits a bridge in the weak norm. Regularized empirical versions can use basis expansions or other controlled function classes \citep{dikkala2020minimax,bennett2026inference}. This approach avoids directly minimizing an inaccessible strong bridge error, but it does not eliminate the inverse problem.

To estimate the correction, choose a bridge basis $p_j(Y_j)\in\R^{d_j}$ and an instrument basis $b_j(W_j)\in\R^{m_j}$. Write
\(\phi_j(Y_j)=p_j(Y_j)^\top\theta_j\), and \(v_j(W_j)=b_j(W_j)^\top\beta_j.\)
The population null restrictions $v_j \in \calN_j$, become $ M_j\beta_j=0$, where $M_j=\E[p_j(Y_j)b_j(W_j)^\top]\in\R^{d_j\times m_j}$. If $p_j$ contains all indicators on a finite support of $Y_j$, these equations are equivalent to $\E[v_j(W_j)\mid Y_j]=0$. 

On a training fold, let $\mathbb P_{\mathrm{tr}}$ denote the empirical average and define \(\widehat M_j=\mathbb P_{\mathrm{tr}}[p_jb_j^\top]\), \(\widehat u_j=Y_1-\widehat\phi_j(Y_j)\), and \(\widehat\psi_{\mathrm b}=a(Z)\{Y_1-\widehat\mu(Z)\}\). Stack $\beta=(\beta_2^\top,\ldots,\beta_J^\top)^\top$ and solve
\begin{equation}\label{eq:sample-optimization}
 \begin{split}
 \widehat\beta\in\arg\min_{\widehat M_j\beta_j=0\ \forall j}
 \bigg[&\mathbb P_{\mathrm{tr}}\!\left[
 \left\{\widehat\psi_{\mathrm b}
 +\sum_{j=2}^J\widehat u_jb_j(W_j)^\top\beta_j\right\}^{2}\right]\\
 &+\lambda_n\sum_{j=2}^J
 \mathbb P_{\mathrm{tr}}[(b_j(W_j)^\top\beta_j)^2]\bigg].
 \end{split}
\end{equation}
The fitted functions are $\widehat v_j(w)=b_j(w)^\top\widehat\beta_j$. The first term estimates the joint population variance criterion. The second penalty term controls the size of the correction functions. Appendix~\ref{app:implementation} gives the matrix solution and describes the null-space basis used to impose the constraints.

\begin{theorem}[Asymptotic linearity and inference]\label{thm:estimator}
Suppose Assumption~\ref{ass:tangent} holds and $0<\Ve<\infty$. Let $K$ be fixed and $|I_k|/n\to\kappa_k\in(0,1)$. For each fold, evaluate norms and expectations on an independent observation conditional on its training sample. Suppose all fitted contributions are square-integrable, the estimated bridge errors belong to their respective $L^2(P_{Y_j})$ spaces, the fitted weights belong to $L^2(P_{W_j})$, and
\begin{align}
 \norm{m(\,\cdot\,;\widehat\mu^{(-k)},\widehat\phi^{(-k)},\widehat v^{(-k)})
       -(\tau_0+\psieff)}_2&=o_p(1),
 \label{eq:score-consistency}\\
 \left|\sum_{j=2}^J\ip{T_j^*\widehat v_j^{(-k)}}{
       \widehat\phi_j^{(-k)}-\phi_{0j}}\right|&=o_p(n^{-1/2}).
 \label{eq:bias-condition}
\end{align}
Then
\begin{equation}\label{eq:CLT}
 \sqrt n(\widehat\tau-\tau_0)
 =\frac1{\sqrt n}\sum_{i=1}^n\psieff(O_i)+o_p(1)
 \ \overset d\longrightarrow\ N(0,\Ve).
\end{equation}
Moreover,
\begin{equation}\label{eq:SE}
 \widehat V=\frac1n\sum_{i=1}^n(\widehat m_i-\widehat\tau)^2\to_p\Ve,
 \qquad
 \widehat\tau\ \pm\ z_{1-\alpha/2}\sqrt{\widehat V/n}
\end{equation}
has asymptotic pointwise coverage $1-\alpha$.
\end{theorem}

Theorem \ref{thm:estimator} establishes asymptotic normality. Condition~\eqref{eq:score-consistency} requires the whole fitted contribution to approach the efficient one in mean square. It includes approximation of the optimal combined correction and stability of products involving estimated bridges and weights. Condition~\eqref{eq:bias-condition} is the second-order remainder condition arising from Neyman orthogonality. It is useful to see that for any reference $v^\circ$ satisfying the population null constraints, by \eqref{eq:bias-exact}, 
\begin{align}
 |\E[m]-\tau_0|
 &\leq\sum_{j=2}^J\norm{v_j-v_j^\circ}_2
             \norm{T_j(\phi_j-\phi_{0j})}_2,
 \label{eq:bias-weak}\\
 |\E[m]-\tau_0|
 &\leq\sum_{j=2}^J\norm{T_j^*v_j}_2
             \norm{\phi_j-\phi_{0j}}_2.
 \label{eq:bias-balance}
\end{align} The first bound pairs a strong correction error with a weak bridge error. The second pairs population imbalance, $T_j^*v_j=\E[v_j(W_j)\mid Y_j]$, with a strong bridge error. The sufficiently small product in either one implies condition~\eqref{eq:bias-condition}.



Our derivation begins with the benchmark derivative and exploits the bridge restrictions to identify admissible variance reductions. In Appendix~\ref{app:riesz}, we show that the same residual correction can also be understood through Riesz representation and NPIV debiasing. Appendix~\ref{app:source} discusses the source condition. 

\subsection{Pre-treatment covariates}
Observed pre-treatment covariates $X$ can improve precision without changing the benchmark-scale target. Let $\mu(z,x)=\E[Y_1\mid Z=z,X=x]$. Ordinary covariate adjustment replaces the benchmark influence function by
\[
 \psi_{\mathrm b,X}
 =\mu(1,X)-\mu(0,X)-\tau_0
  +a(Z)\{Y_1-\mu(Z,X)\},
\]
as in standard randomized-trial efficiency analysis \citep{hahn1998role,tsiatis2008covariate}. Covariates can also enter the bridge conditioning vectors, provided equations $\E[Y_1-\phi_{0j}(Y_j)\mid W_j]=0$ remain valid with $(Z,X)\subseteq W_j$. 

The efficiency result applies with this gradient and the enlarged residual space $\calV_X$, giving $\psi_{\mathrm{eff},X}=\psi_{\mathrm b,X}-\Pi_{\calV_X}\psi_{\mathrm b,X}$. If calibration remains $\phi_j(Y_j)$, the null constraint remains $\E[v_j(W_j)\mid Y_j]=0$. Appendix~\ref{app:covariates} derives the covariate-specific score and precision formulas, gives the required changes to the cross-fitted estimator.

\section{Simulation}

We conduct two simulation exercises to evaluate the performance of the NSI estimator. The first examines finite-sample bias, variance estimation, confidence interval coverage, test size, and statistical power with categorical measurements. Corresponding results for continuous measurements are reported in SI~\ref{app:sim2-continuous-experiment}. The second assesses the stability of treatment-effect estimates across nonlinear measurement systems. 


\subsection{Finite-sample performance with five-category measurements}\label{app:sim2-discrete-experiment}

Because categorical indicators are widely used in practice, our first experiment focuses on measurements with five categories and evaluates the finite-sample bias of NSI and the accuracy of its standard errors for statistical inference. We consider two treatment-effect values, \(\tau=0\) and \(\tau=0.06\), and sample sizes \(n=400,800,1200,1600\).

We generate randomized treatment and a predetermined covariate independently as \(Z,X\sim\operatorname{Bernoulli}(1/2)\). The latent outcome takes the five values \(\eta_k\in\{0.1,0.3,0.5,0.7,0.9\}\). Treatment and the covariate shift probability toward higher or lower latent values according to
\[
\begin{gathered}
\Pr(\eta=\eta_k\mid Z=z,X=x)
=\frac15+\left\{0.045(2x-1)+\frac{\theta}{2}(z-\tfrac12)\right\}d_k
\end{gathered}
\] where $d=(-2,-1,0,1,2)^{\mathsf T}.$ This gives the conditional mean
\(
\mathbb E(\eta\mid Z,X)=\frac12+0.09(2X-1)+\theta(Z-\tfrac12).
\) All five latent probabilities are positive at the two effect sizes considered. Under \(\theta=0\), the full latent distribution is the same across treatment arms conditional on \(X\).

Conditional on the latent outcome, we draw three mutually independent measurements. The benchmark is \(Y_1=K/4\), where \(K\mid\eta\sim\operatorname{Binomial}(4,\eta)\), so it takes values \(0,0.25,0.5,0.75,1\). Each auxiliary, \(Y_2\) or \(Y_3\), takes values \(0,1,2,3,4\): when \(\eta=\eta_k\), it reports the corresponding category \(k-1\) with probability \(0.90\) and each other category with probability \(0.025\). Because \(\mathbb E(Y_1\mid\eta)=\eta\), the benchmark-scale treatment effect is\(\tau=\mathbb E(Y_1\mid Z=1)-\mathbb E(Y_1\mid Z=0)=\theta.\) The auxiliary measurements provide additional information about the same latent outcome. Each has an identified common bridge to the benchmark scale. The measurement probability matrices, population bridge values, and identification argument are given in SI~\ref{app:sim2-discrete-identification}.

We partition each sample into five folds. For each fold, we estimate the nuisance parameters using the other four folds and apply the resulting estimator to observations in the held-out fold. Within each training sample, we use three-fold cross-validation to select tuning parameters, favoring choices that yield lower estimated variability. SI~\ref{app:sim2-discrete-identification} provides the estimating equations and details of the tuning procedure. We compare NSI with a benchmark estimator: it estimates the treatment effect using the benchmark measurement \(Y_1\) and covariate \(X\).

\begin{table}[htbp]
\centering\small
\setlength{\tabcolsep}{4pt}
\caption{Finite-sample performance of NSI with five-category measurements}
\label{tab:sim2-discrete5-bridge}
\begin{tabular}{rrrrrrrr}
\hline
$n$ & $\theta$ & Bias & Emp. SD & RMS SE & Var. ratio &  Rej. (\%) \\
\hline
400 & 0.00 & -0.00053 & 0.0291 & 0.0284 & 0.953 &  6.20 \\
800 & 0.00 & -0.00049 & 0.0204 & 0.0201 & 0.963 &  5.50 \\
1200 & 0.00 & -0.00056 & 0.0162 & 0.0164 & 1.018 &  4.70 \\
1600 & 0.00 & -0.00001 & 0.0143 & 0.0142 & 0.979 & 5.00 \\
\hline
400 & 0.06 & 0.00072 & 0.0297 & 0.0283 & 0.909 &  58.20 \\
800 & 0.06 & 0.00038 & 0.0202 & 0.0200 & 0.979 &  86.40 \\
1200 & 0.06 & 0.00071 & 0.0168 & 0.0163 & 0.945 &  95.20 \\
1600 & 0.06 & -0.00031 & 0.0147 & 0.0141 & 0.930 &  98.00 \\
\hline
\end{tabular}
\par\medskip
\begin{minipage}{.98\linewidth}\footnotesize
\textit{Notes:} All estimators use covariate adjustment. Null cells use 1,000 replications and alternative cells use 500 Emp. SD is the standard deviation across replications; RMS SE is the square root of the mean squared reported standard error. Var. ratio is $(\mathrm{RMS\ SE}/\mathrm{Emp.\ SD})^2$. Rejection measures size at $\theta=0$ (target five percent) and power at $\theta=0.06$.
\end{minipage}
\end{table}

Table~\ref{tab:sim2-discrete5-bridge} shows that NSI exhibits little finite-sample bias, with a maximum absolute bias of only 0.00072. Under the null hypothesis ($\theta=0$), the empirical rejection rate reaches the nominal 5\% level at $n=1600$. Under the alternative ($\theta=0.06$), statistical power increases with sample size, reaching 86.4\% at a sample size of 800. Figure~\ref{fig:sim2-discrete-power} further shows that NSI achieves higher empirical power than the benchmark estimator at every simulated sample size.

Figure~\ref{fig:sim2-discrete-variance} compares the empirical standard deviations with the root mean squared estimated standard errors. The ratio of the mean estimated variance to the empirical variance ranges from 0.909 to 1.018. 
 A corresponding simulation with continuous measurements and nonlinear bridge functions is presented in SI~\ref{app:sim2-continuous-experiment}.

\begin{figure}[tbp]
\centering
\includegraphics[width=\linewidth]{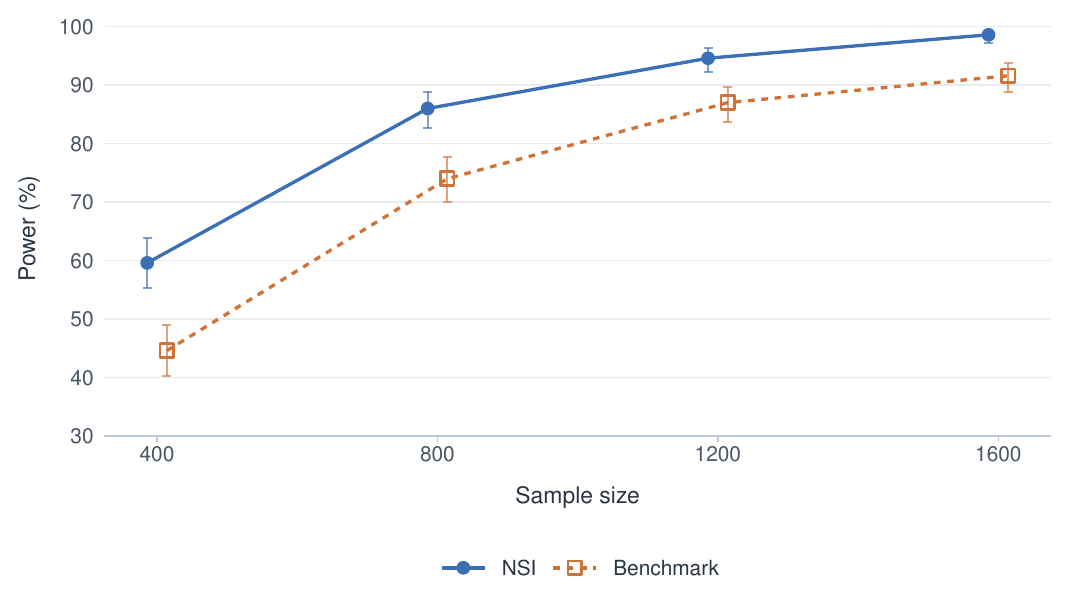}
\caption{Power with five-category measurements at $\theta=0.06$. Points report rejection by two-sided five-percent tests over 500 replications per sample size. Bars are pointwise 95-percent Wilson Monte Carlo intervals. Blue solid lines and filled symbols denote NSI; orange dashed lines and open symbols denote the benchmark.}
\label{fig:sim2-discrete-power}
\end{figure}

\begin{figure}[tbp]
\centering
\includegraphics[width=\linewidth]{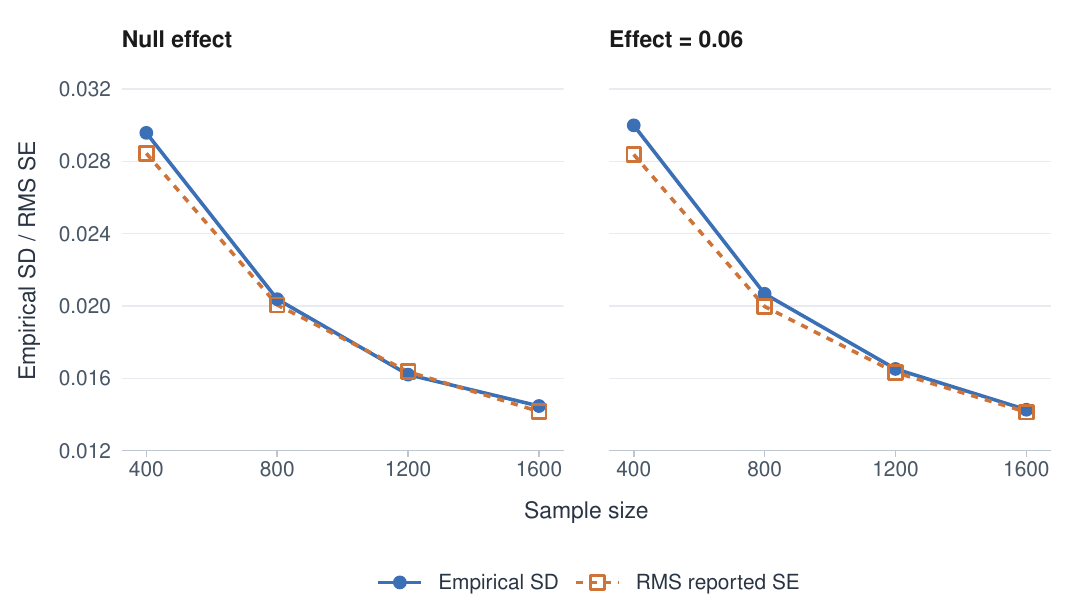}
\caption{NSI variance calibration with five-category measurements. Blue solid lines and filled symbols show empirical standard deviations; orange dashed lines and open symbols show root mean squared reported standard errors.}
\label{fig:sim2-discrete-variance}
\end{figure}

\subsection{Comparability across nonlinear measurement systems} \label{sec:sim-comparability}

This simulation examines whether the estimators are robust to changes in the set of measurements. To isolate the effect of changing the measurement system, we construct two paired studies that share the same treatment assignments, latent outcomes, benchmark measurement, and treatment effects. The two studies differ only in their auxiliary measurements. Each of the 500 replications contains \(n=800\) observations, with treatment assigned independently with probability one half. The untreated latent outcome is \( \eta_{0i}=0.9x_i+0.3x_i^2+u_i,\) where $x_i\sim N(0,1)$, and $u_i\sim N(0,0.7^2)$, and the individual treatment effect is \(\tau_i=0.8+0.3x_i\), giving an average latent treatment effect of \(0.8\). The realized latent outcome is \(\eta_i=\eta_{0i}+Z_i\tau_i\), and the common benchmark is
\(Y_{i1}=\eta_i+\epsilon_{i1}\), where
\(\epsilon_{i1}\sim N(0,0.8^2)\). Both studies use the same realized benchmark, including its measurement error.

Each study observes the benchmark and two auxiliary measurements. To construct the latter, let \(V_{ij}^{(s)}=\eta_i+\epsilon_{ij}^{(s)}\) for
study \(s\in\{1,2\}\) and measurement \(j\in\{2,3\}\). The auxiliary
errors are mean-zero Gaussian draws. Their standard
deviations are \((0.45,0.55)\) in Study~1 and \((0.50,0.70)\) in Study~2, respectively, so the second study has noisier auxiliary measurements before transformation. Study~1 applies increasing cubic and hyperbolic-sine transformations:
\[
\begin{aligned}
Y_{i2}^{(1)}&=V_{i2}^{(1)}+0.30\{V_{i2}^{(1)}\}^3,\\
Y_{i3}^{(1)}&=\sinh\{0.70V_{i3}^{(1)}\}/0.70.
\end{aligned}
\]
Both transformations magnify differences at large absolute values of the noisy latent measurements. Study~2 instead uses
\[
\begin{aligned}
Y_{i2}^{(2)}&=-\operatorname{arsinh}\{1.25V_{i2}^{(2)}\}/1.25,\\
Y_{i3}^{(2)}&=\tanh\{0.80V_{i3}^{(2)}\}/0.80.
\end{aligned}
\]
The first transformation reverses the direction of the measurement and compresses its tails; the second preserves direction but compresses the measurement into a bounded range. Thus, the paired comparison varies measurement precision, curvature, orientation,
and compression while holding the underlying causal effect, benchmark scale, and number of measurements fixed.

The NSI estimator uses \(W_2=(Z,Y_3)\) and \(W_3=(Z,Y_2)\). We use five outer folds and three inner folds for tuning. For each bridge, we jointly select its basis configuration and bridge penalty from
five prespecified pairs: degree-one polynomial bridge and instrument bases (linear); bounded natural-spline bridge/instrument dimensions
\((3,4)\), \((4,5)\), and \((6,8)\); and an ordinary natural-spline bridge of dimension four with a bounded natural-spline instrument basis of dimension five. The instrument basis includes treatment interactions. 

We compare PCA, ICW, and NSI. PCA and ICW retain their conventional standardizations within each measurement panel, whereas NSI is expressed in benchmark units. For each method, we report the Monte Carlo mean and standard deviation in each study and the mean paired difference between studies. To summarize sensitivity to the measurement system without designating either study as the reference, we also report the relative difference
\[
100\frac{|\overline{\widehat\tau}_{m1}-\overline{\widehat\tau}_{m2}|}
{\{|\overline{\widehat\tau}_{m1}|+|\overline{\widehat\tau}_{m2}|\}/2}.
\]
This measure expresses the absolute difference between the two means as a percentage of their average absolute magnitude and is invariant to a common rescaling.

\input{tables/table_comparability.tex}

Table~\ref{tab:cbs-comparability} reports the results. NSI produces similar estimates across the two studies: the Monte Carlo means are 0.777 and 0.764, with a mean paired difference of 0.0123 and a relative difference of 1.6 percent. In contrast, the corresponding relative differences are substantially larger for PCA (31.4 percent) and ICW (25.5 percent). NSI substantially reduces sensitivity to the auxiliary measurement system.

\section{Application: Winning Access to the Hajj}
\label{sec:hajj-application}

We return to the lottery design and measurement problem introduced in Section~\ref{sec:hajj-running-example}. We study the causal effect of winning access to the Hajj---that is, receiving a visa---on a range of latent social attitudes. The Hajj is one of the most important religious events in the lives of Muslims. Because demand for participation greatly exceeds the available quota, some countries use lotteries to determine who receives a Hajj visa. Of the 134,948 applicants to the government lottery, 59 percent won. Applicants entered the lottery in parties of up to twenty people and were randomized within cells defined by sect, departure location, accommodation category, and party-size considerations. The follow-up survey interviewed 1,605 applicants eight to eleven months after the lottery.

We first examine six latent outcomes on tolerance from the original study by \citet{clingingsmith2009estimating}, then return to beliefs in localized Muslim practices, the outcome used in the running example. Throughout, treatment is
 \emph{winning the visa lottery}. 

\subsection{Tolerance, intergroup relations, and political attitudes}
\label{sec:hajj-tolerance}

The six latent tolerance outcomes can be classified into two broad categories. The first, \emph{intergroup regard and coexistence}, includes views of people from other countries, perceived equality across ethnic and religious groups, and harmonious coexistence. The second, \emph{peace and political boundaries}, includes peaceful inclination, the political role of Islam, and views of the West. Higher scores indicate more favorable intergroup attitudes, greater perceived equality and harmony, and a stronger inclination toward peace. For political Islam, higher scores indicate greater support for a political role for Islam, whereas for views of the West, higher scores indicate less favorable views.

For each outcome, we choose as the benchmark the item with the largest number of distinct response values.\footnote{Ties are broken by the number of observed responses and then by the listed item order.} The Appendix~\ref{si:app} lists all measurements, benchmark items, and response codes. For example, the benchmark for views of people from other countries is the rating of Saudi people, coded from $-2$ to $2$, whereas the benchmark for harmony is the binary indicator of frequent prayer in another sect's mosque.

We estimate NSI for each latent outcome using two choices of instruments for estimating the measurement bridges. Let $Z_i$ indicate lottery success, and let $X_i$ contain lottery-cell indicators. Following the original study, we construct lottery-cell indicators from
departure location, accommodation category, and party-size category. The first, basic instrument set uses $(Z_i,X_i)$. The second, extended instrument set additionally includes the other measurements, $(Z_i,X_i,Y_{i,-\{1,j\},r})$, where $Y_{i,-\{1,j\},r}$ contains the other measurements retained for outcome $r$, excluding the benchmark measurement and measurement $j$ itself. Both specifications use the same respondents, benchmark measurements, nuisance covariates, and five cross-fitting folds, with all respondents from the same application party assigned to the same fold. Implementation details and additional results are provided in Appendix~\ref{si:app}.

Using treatment and lottery cells as instruments, the estimate for views of other countries is $-0.0087$, with a 95\% interval of $[-0.0421,0.0247]$. Adding the other measurements as instruments yields an estimate of $0.0338$, with a 95\% interval of $[0.0058,0.0619]$ and $p=0.018$. The corresponding estimates for views of other groups are $0.0411$ and $0.0721$, with both intervals excluding zero ($p=0.010$ and $p=0.002$). For harmony, the basic-instrument estimate is $-0.0002$, whereas the extended-instrument estimate is $0.0018$, with a 95\% interval of $[0.0003,0.0033]$ and $p=0.022$. Peaceful inclination, political Islam, and views of the West have intervals that include zero under both specifications.

\begin{figure}[!htb]
\centering
\includegraphics[width=\linewidth]{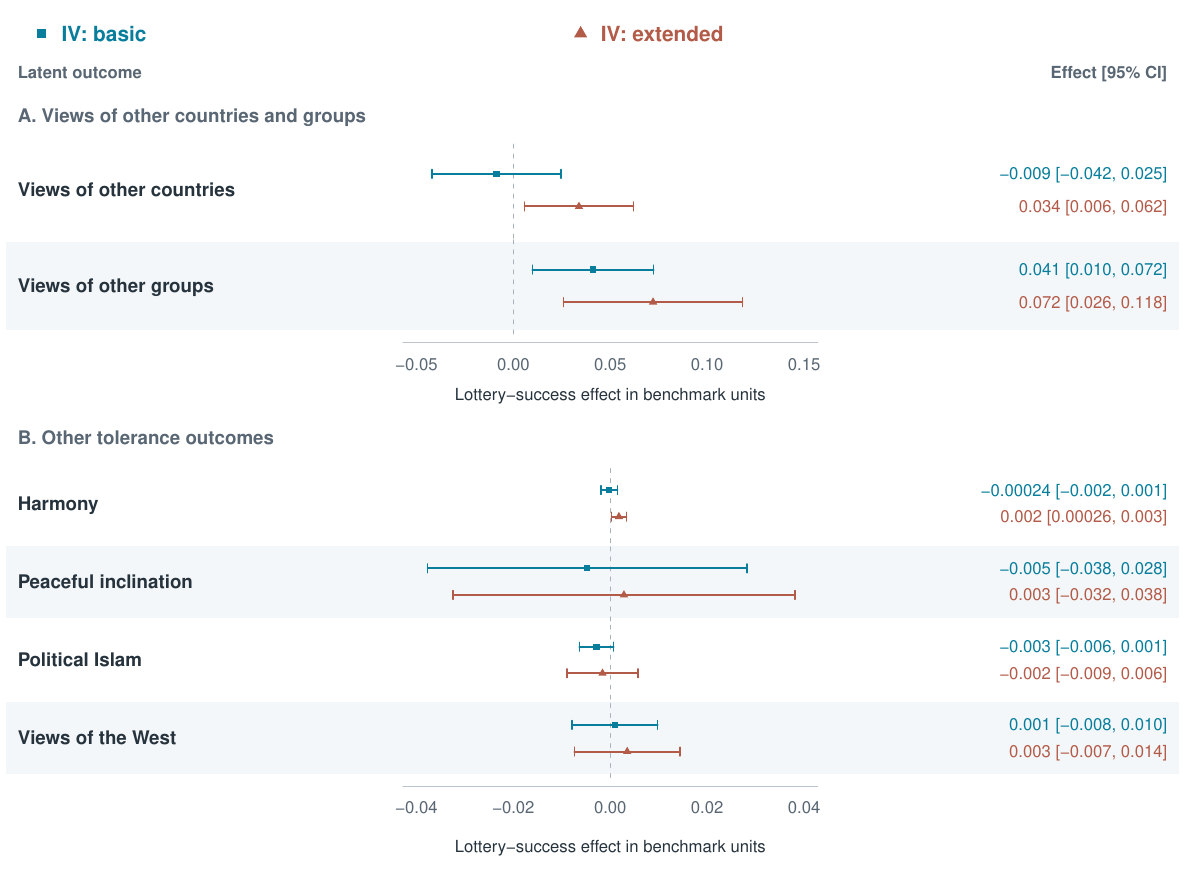}
\caption{Causal effects on the six tolerance latent outcomes}
\label{fig:hajj-main-tolerance}
\begin{minipage}{\linewidth}
\renewcommand{\baselinestretch}{1}\footnotesize
\textit{Notes:}  \textit{IV: basic} (squares) uses treatment and lottery-cell indicators as instruments. \textit{IV: extended} (triangles) also uses the other measurements in the same outcome family, excluding the benchmark and the item being transformed.
\end{minipage}
\end{figure}

Substantively, the estimates under extended instruments suggest more favorable views of people from other countries, greater acceptance of other ethnic and religious groups, and more harmonious coexistence. These patterns are broadly consistent with \citet[Table~V]{clingingsmith2009estimating}, who report improvements from Hajj attendance in these three outcomes as well as in peaceful inclination. Like the original study, we find no clear evidence of changes in political Islam or views of the West; unlike the original study, however, we find no clear increase in peaceful inclination. This comparison is necessarily qualitative because the original study estimates the effects of Hajj attendance using lottery assignment as an instrument and reports averages of standardized item effects, whereas we estimate the effects of lottery success using NSI in benchmark units and complete-case samples. 

\subsection{Returning to beliefs in localized Muslim practices}
\label{sec:hajj-localized-beliefs}

We now use NSI to revisit the measurement choices in the running example. The latent outcome is endorsement of localized religious beliefs and customs, measured by the ten questions in Table~\ref{tab:hajj-running-measurements}. There, we applied methods that do not satisfy measurement-set comparability. As discussed above, the resulting causal effects are not comparable across different sets of measurements because each set defines a different estimand.

We choose the first measurement as the benchmark. NSI transforms all other measurements onto the same scale as the benchmark, allowing them to be combined while preserving measurement-set comparability. In batteries C and D, the benchmark remains observed and is used to estimate the transformations. 

\input{hajj_application/tables/table_localized_beliefs_nsi.tex}

Table~\ref{tab:hajj-localized-nsi-batteries} reports the NSI estimates. All six batteries yield negative estimates under both specifications. The 95\% intervals exclude zero for batteries A--E, whereas those for battery F include zero. The direction of the estimates is therefore consistent across the selected batteries. In the running example, the absolute magnitude of the AES estimate increases by 59\%, with a statistically significant paired difference ($p=0.047$). In contrast, NSI yields a negligible difference of $0.0108$. For batteries C and D, PCA produces estimates with opposite signs, whereas the corresponding NSI estimates remain negative. For batteries E and F, the difference between the NSI estimates is only $0.0045$ ($p=0.599$), compared with a difference of $0.15$ between the ICW estimates. Taken together, under the bridge assumptions, all NSI specifications target the treatment effect on the same benchmark scale.

\section{Conclusion}
Causal inference in experiments is often framed as a problem of treatment assignment, identification, and estimation, taking the outcome as given. This paper argues that when the outcome of interest is latent, measurement is itself part of the causal inference problem. We identify a fundamental comparability problem in current practice: different sets of measurements may lead standard methods to target different empirical quantities even when the underlying latent treatment effect is the same. 

We develop a general nonparametric framework to address this problem. Our approach uses a benchmark measurement and nonparametric bridge functions to align different measurements with a common scale in expectation. We further derive the semiparametric efficiency bound and characterize the conditions under which additional measurements yield efficiency gains. The framework allows the relationship between measurements and the latent outcome to remain unknown and nonlinear. More broadly, the paper suggests that measurement design should be treated as part of experimental design. Researchers who collect multiple measurements within a study should consider whether those measurements preserve enough information about the latent outcome to support bridge-function identification. 

Several limitations and extensions remain for future work. The nonparametric approach  can be less stable and efficient in finite samples. Future research could investigate weaker identification strategies, improved regularization and efficiency in finite samples, and extensions to settings with richer treatment regimes, longitudinal measurements, or interference. Even with these limitations, the central message is clear: when outcomes are latent, causal inference cannot proceed as if measurement were secondary. Measurement choices shape the empirical meaning of the estimand, the comparability of results, and the credibility of substantive conclusions. By placing measurement at the center of causal inference, this paper aims to provide a foundation for more interpretable, comparable, and robust experimental research on latent outcomes.

\section*{Statement of AI Usage} In preparing this paper, the authors used ChatGPT 5.6 Sol and 6 Astra for the following tasks: (1) reviewing and refining author-written R code for the accompanying R package; (2) proofreading the final draft; (3) assisting with code for simulations and the formatting of tables and figures; and (4) checking mathematical proofs and identifying and correcting errors. All AI-assisted outputs were carefully reviewed, and where necessary, revised by the authors. The authors take full responsibility for the content of the paper.

\begin{spacing}{0.0}
	\setlength{\bibsep}{10pt}
	\setstretch{0.1}
	\bibliographystyle{apalike2}
	\bibliography{literature}
\end{spacing}

\newpage


\clearpage

\appendix
\addcontentsline{toc}{section}{Appendix} 
\part{Supplementary Information} 
\parttoc 

\setcounter{figure}{0}
\setcounter{table}{0}
\setcounter{proposition}{0} 
\renewcommand\thefigure{A.\arabic{figure}}
\renewcommand\thetable{A.\arabic{table}}
\renewcommand\theproposition{\thesection.\arabic{proposition}}

\onehalfspacing
\setcounter{page}{1}

\singlespacing

\input{appendix_theory}

\end{document}

%% file: tables/table_running_beliefs.tex
\begin{table}[!htb]
\centering
\caption{Beliefs in localized Muslim practices: changes in lottery-success estimates}
\label{tab:hajj-running-beliefs}
\begin{threeparttable}
\renewcommand{\baselinestretch}{1}\small
\setlength{\tabcolsep}{3pt}
\renewcommand{\arraystretch}{1.1}
\begin{tabularx}{\linewidth}{@{}>{\raggedright\arraybackslash}p{0.08\linewidth}>{\raggedright\arraybackslash}X*{3}{>{\centering\arraybackslash}p{0.18\linewidth}}>{\centering\arraybackslash}p{0.06\linewidth}@{}}
\toprule
Method & Batteries (items) & First estimate & Second estimate & \shortstack{Change\\(second $-$ first)} & $p$ \\
\midrule
AES & A: 1, 5, 6, 8\newline B: 1, 4, 7, 8, 9 & \begin{tabular}[t]{@{}c@{}}$-0.090$\\[-1pt]{\footnotesize $[-0.147,-0.032]$}\end{tabular} & \begin{tabular}[t]{@{}c@{}}$-0.142$\\[-1pt]{\footnotesize $[-0.195,-0.089]$}\end{tabular} & \begin{tabular}[t]{@{}c@{}}$-0.053$\\[-1pt]{\footnotesize $[-0.096,-0.009]$}\end{tabular} & $0.047$ \\[7pt]
PCA & C: 3, 7, 8, 9\newline D: 7, 8, 9 & \begin{tabular}[t]{@{}c@{}}$-0.144$\\[-1pt]{\footnotesize $[-0.253,-0.036]$}\end{tabular} & \begin{tabular}[t]{@{}c@{}}$0.112$\\[-1pt]{\footnotesize $[0.014,0.211]$}\end{tabular} & \begin{tabular}[t]{@{}c@{}}$0.257$\\[-1pt]{\footnotesize $[0.120,0.393]$}\end{tabular} & $0.002$ \\[7pt]
ICW & E: 1, 4, 7, 9\newline F: 1, 6 & \begin{tabular}[t]{@{}c@{}}$-0.286$\\[-1pt]{\footnotesize $[-0.392,-0.180]$}\end{tabular} & \begin{tabular}[t]{@{}c@{}}$-0.136$\\[-1pt]{\footnotesize $[-0.234,-0.038]$}\end{tabular} & \begin{tabular}[t]{@{}c@{}}$0.151$\\[-1pt]{\footnotesize $[0.055,0.247]$}\end{tabular} & $0.010$ \\[7pt]
\bottomrule
\end{tabularx}
\begin{tablenotes}[flushleft]
\item \textit{Notes:} Each row compares only the displayed method on its two selected batteries. Brackets contain 90\% confidence intervals. The final column tests the paired change. Every paired change and both estimates in each row are significant at the 10\% level. Letters identify the selected batteries; item numbers refer to Table~\ref{tab:hajj-running-measurements}.
\end{tablenotes}
\end{threeparttable}
\end{table}

%% file: tables/table_comparability.tex
\begin{table}[htbp]
\centering
\caption{Paired nonlinear measurement-system comparison}
\label{tab:cbs-comparability}
\begin{tabular}{lllll}
\hline
Method  & Study 1: mean (SD) & Study 2: mean (SD) & Gap (MCSE) & Rel. gap (\%) \\
\hline
PCA  & 0.641 (0.108) & 0.880 (0.097) & -0.2388 (0.0045) & 31.4 \\
ICW  & 0.388 (0.063) & 0.502 (0.056) & -0.1135 (0.0024) & 25.5 \\
NSI  & 0.777 (0.130) & 0.764 (0.097) & 0.0123 (0.0047) & 1.6 \\
\hline
\end{tabular}
\par\medskip
\begin{minipage}{0.98\linewidth}
\footnotesize \textit{Notes:} The gap is Study 1 minus Study 2. Relative gap is the absolute gap divided by the average absolute mean.
\end{minipage}
\end{table}

%% file: hajj_application/tables/table_localized_beliefs_nsi.tex
\begin{table}[!htbp]\renewcommand{\baselinestretch}{1}\centering\small\selectfont
\caption{Localized beliefs: fixed auxiliary-battery comparisons}
\label{tab:hajj-localized-nsi-batteries}
\begin{threeparttable}
\setlength{\tabcolsep}{3pt}
\begin{tabularx}{\linewidth}{@{}>{\raggedright\arraybackslash}X*{3}{>{\centering\arraybackslash}p{0.21\linewidth}}c@{}}
\toprule
Batteries (items) & First estimate & Second estimate & \shortstack{Change\\(second $-$ first)} & $p$ \\
\midrule
A: 1, 5, 6, 8\newline B: 1, 4, 7, 8, 9 & \begin{tabular}[t]{@{}c@{}}$-0.0301$\\[-1pt]{\footnotesize $[-0.0519,-0.0083]$}\end{tabular} & \begin{tabular}[t]{@{}c@{}}$-0.0193$\\[-1pt]{\footnotesize $[-0.0320,-0.0065]$}\end{tabular} & \begin{tabular}[t]{@{}c@{}}$0.0108$\\[-1pt]{\footnotesize $[-0.0054,0.0271]$}\end{tabular} & 0.192 \\[6pt]
C: 3, 7, 8, 9\newline D: 7, 8, 9 & \begin{tabular}[t]{@{}c@{}}$-0.0227$\\[-1pt]{\footnotesize $[-0.0347,-0.0107]$}\end{tabular} & \begin{tabular}[t]{@{}c@{}}$-0.0229$\\[-1pt]{\footnotesize $[-0.0364,-0.0094]$}\end{tabular} & \begin{tabular}[t]{@{}c@{}}$-0.0001$\\[-1pt]{\footnotesize $[-0.0052,0.0049]$}\end{tabular} & 0.955 \\[6pt]
E: 1, 4, 7, 9\newline F: 1, 6 & \begin{tabular}[t]{@{}c@{}}$-0.0205$\\[-1pt]{\footnotesize $[-0.0331,-0.0079]$}\end{tabular} & \begin{tabular}[t]{@{}c@{}}$-0.0160$\\[-1pt]{\footnotesize $[-0.0330,0.0010]$}\end{tabular} & \begin{tabular}[t]{@{}c@{}}$0.0045$\\[-1pt]{\footnotesize $[-0.0123,0.0213]$}\end{tabular} & 0.599 \\[6pt]
\bottomrule\end{tabularx}
\begin{tablenotes}[flushleft]\footnotesize
\item \textit{Notes:} Changes are second minus first. We use extended IVs here. Brackets give nominal 95\% intervals. All fits use the original units of item 1 and the same 1,165 respondents in 804 application parties, fixed benchmark and outer folds. The listed items identify the selected batteries in the running example.
\end{tablenotes}\end{threeparttable}\end{table}

%% file: appendix_theory.tex
\section{Proofs of Bridge Identification}\label{si:proof}

\subsection{Proof of Proposition \ref{eq:set-comparability-characterization}}

\begin{proof}
Let $\mathcal B$ be a nonempty
collection of nonempty subsets of the measurement index set $\mathcal J$. We first prove sufficiency, which does not require pooled fitting.
Fix an admissible population $P$ and write
$F=F_P=(1-\pi)F_{0,P}+\pi F_{1,P}$. Since $0<\pi<1$,
both $F_{0,P}$ and $F_{1,P}$ are absolutely continuous with respect to
$F$. Thus, if \eqref{eq:comparability-common-mean} holds $F$-almost
surely, it also holds almost surely under each treatment-specific latent
distribution. It follows that
\[
 \begin{aligned}
 \tau_{\mathcal M}(P)
 &=\int h_{\mathcal M,P}(a)\,F_{1,P}(da)
   -\int h_{\mathcal M,P}(a)\,F_{0,P}(da)\\
 &=\int h_P(a)\,F_{1,P}(da)
   -\int h_P(a)\,F_{0,P}(da),
 \end{aligned}
\]
because the additive constant $c_{\mathcal M}(P)$ cancels. The last
expression does not depend on $\mathcal M$. Since $P$ was arbitrary,
the scoring family is measurement-set comparable.

For necessity, suppose the family is measurement-set comparable and all
fitting uses only pooled measurements. Fix an admissible population $P$,
write $F=F_P$, and choose a reference measurement set
$\mathcal M_0\in\mathcal B$. For any other
$\mathcal M\in\mathcal B$, define
\[
 d(a)=h_{\mathcal M,P}(a)-h_{\mathcal M_0,P}(a).
\]
The score integrability conditions and conditional Jensen's inequality
imply $d\in L^1(F)$. Let $B\subseteq\mathcal A$ be an arbitrary measurable set and put
\(
 g_B(a)=\mathbf{1}\{a\in B\}-F(B).
\)
Then $|g_B|\leq1$ and $\int g_B\,dF=0$. Define two latent
distributions by their densities relative to $F$:
\[
 \begin{aligned}
 F_1^B(da)&=\{1+(1-\pi)g_B(a)\}\,F(da),\\
 F_0^B(da)&=\{1-\pi g_B(a)\}\,F(da).
 \end{aligned}
\]
Both densities are positive and bounded, and each integrates to one.
Moreover,
\[
 \begin{aligned}
 (1-\pi)F_0^B+\pi F_1^B&=F,\\
 F_1^B(da)-F_0^B(da)&=g_B(a)\,F(da).
 \end{aligned}
\]
Consider a population $P^B$ having these latent treatment and control
distributions, the same measurement kernels, and the same assignment
probability. Such a population is included among the admissible
decompositions.

The pooled latent distribution under $P^B$ remains $F$. Consequently,
the pooled measurement distributions also remain unchanged. By the
pooled-fitting condition, the same procedure produces the same scoring
functions under $P$ and $P^B$. Hence their conditional-mean difference
is still $d$, $F$-almost surely. Applying measurement-set comparability
at $P^B$ gives
\[
 \begin{aligned}
 0
 &=\tau_{\mathcal M}(P^B)-\tau_{\mathcal M_0}(P^B)\\
 &=\int d(a)\,F_1^B(da)-\int d(a)\,F_0^B(da)\\
 &=\int d(a)g_B(a)\,F(da).
 \end{aligned}
\]
Substituting the definition of $g_B$ yields for every measurable $B$,
\[
 \int_B d(a)\,F(da)
 =F(B)\int d(a)\,F(da)
\]
Let $c_{\mathcal M}(P)=\int d(a)\,F(da)$. The last identity says
that the integral of $d-c_{\mathcal M}(P)$ over every measurable set
is zero. Taking, in particular, its positive and negative sets shows
that $d=c_{\mathcal M}(P)$ $F$-almost surely.

Set $h_P=h_{\mathcal M_0,P}$ and $c_{\mathcal M_0}(P)=0$.
Since $\mathcal M$ was arbitrary,
\[
 h_{\mathcal M,P}(a)=h_P(a)+c_{\mathcal M}(P)
 \quad F_P\text{-almost surely},
 \qquad \mathcal M\in\mathcal B.
\]
This proves necessity. The construction applies to every admissible $P$,
completing the proof.
\end{proof}

\subsection{Proof of Proposition \ref{prop:picard}}

\begin{proof}
The proof follows \citet{miao2018identifying}. We mainly apply Picard's theorem \citep[][Theorem 15.18, Page 311]{kress1989linear}:
\begin{theorem}[Picard]
    Let $K: H_1 \rightarrow H_2$ be a compact linear operator with singular system ($\mu_n,\alpha_n,\beta_n$). The equation of the first kind $Kh=r$ is solvable if and only if 
    
    1. $r$ belongs to the orthogonal complement $\mathcal{N}(K^*)^\perp$; and
    
    2. $\sum_{n=1}^\infty \frac{1}{\mu_n^2}|<r,\beta_n>|^2 < \infty$.
\end{theorem}

    Define Hilbert spaces $H_1=L^2(F_{Y_{ij}})$ and $H_2=L^2(F_{\eta_i})$, which denote the space of all square integrable functions with respect to a cumulative distribution function $F$ respectively.

    Let $K$ be an integral operator with kernel $k(y_{ij},\eta_i)=\frac{f(y_{ij},\eta_i)}{f(y_{ij})f(\eta_{i})}$. It turns out that $K$ is the conditional expectation operator \citep{carrasco2007linear}. To see this, $K: H_1 \rightarrow H_2$, $K\phi=\int_{-\infty}^{\infty} k(y_{ij},\eta_i) \phi(y_{ij})dF(y_{ij})=\mathbb{E}[\phi(Y_{ij})|\eta_i]$. Also, under $\iint [\frac{f(y_{ij},\eta_i)}{f(y_{ij})f(\eta_{i})}]^2 f(y_{ij})f(\eta_{i})dy_{ij}d\eta_i<\infty$, which is equivalent to the 
  regularity assumption ii, it is compact \citep[][Example 2.3, page 5659]{carrasco2007linear}.

  We observe that $\mathbb{E}[\phi_j(Y_{ij})|\eta_i]=\mathbb{E}[Y_{i1}|\eta_i]$ can be represented by $K\phi=r$, where $r(\eta_i):= \mathbb{E}[Y_{i1}|\eta_i] \in H_2$ by regularity assumption i.

   The adjoint of the conditional expectation operator is $K^*:H_2 \rightarrow H_1$. For $g \in H_2$, $K^* g=\mathbb{E}[g(\eta_i)|Y_{ij}]$. It can be checked by $<K\phi,g>=<\phi,K^*g>$.

The results naturally follow Picard's theorem by noting that regularity assumption iii implies the second condition of the Picard's theorem.

\end{proof}

\subsection{Proof of Corollary \ref{prop:exist}}

\begin{proof}
For any $g \in \mathcal{N}(K^*)=\{g:K^*g=0\}$, we have $K^*g=\mathbb{E}[g(\eta_i)|Y_{ij}]=0$ almost surely, hence by completeness assumption \ref{ass:complete}, $g=0$ almost surely. Therefore, $\mathcal{N}(K^*)=\{g=0\}$, which implies $\mathcal{N}(K^*)^{\perp}=H_2$. Because $r(\eta_i)\in H_2$, we conclude $r \in \mathcal{N}(K^*)^\perp$. 
\end{proof}

\subsection{Proof of Proposition \ref{prop:identify}}

\begin{proof}
Recall that a measurement bridge function satisfies $\mathbb{E}[Y_{i1}|\eta_i]=\mathbb{E}[\phi_j(Y_{ij})|\eta_i].$ Taking expectation on both sides of the bridge equation with respect to $p(\eta_i|W_i)$. The LHS is
$$
\mathbb{E}_{\eta_i|W_i}\mathbb{E}[Y_{i1}|\eta_i]= \mathbb{E}_{\eta_i|W_{i}}\mathbb{E}[Y_{i1}|\eta_i,W_i]=\mathbb{E}[Y_{i1}|W_i]
$$ where the first equation holds if $\mathbb{E}[Y_{i1}|\eta_i]=\mathbb{E}[Y_{i1}|W_i,\eta_i]$. The second equality is the tower property.

The RHS follows the same reasoning but needs $\mathbb{E}[\varphi(Y_{ij})|\eta_i]=\mathbb{E}[\phi(Y_{ij})|W_i,\eta_i]$. Therefore, we obtain $\mathbb{E}[\phi_j(Y_{ij})|W_i]=\mathbb{E}[Y_{i1}|W_i].
$

Now, let $K : \phi \mapsto \mathbb{E}[\phi \mid W_i]$ and define $r(W_i) = \mathbb{E}[Y_{i1} \mid W_i]$. We obtain the well-known NPIV equation $K\phi=r$. The function $\phi$ uniquely identified if and only if completeness condition hold: for any square-integrable function $g$, $\mathbb{E}[g(Y_{ij})|W_i]=0$ for almost all $W_i$ implies $g(Y_{ij})=0$ almost surely \citep{newey2003instrumental}.
\end{proof}

\section{Technical details for Efficiency and Estimation}\label{app:proofs}

This section supplies the arguments used in Section~\ref{sec:measurement-efficiency}. Expectations and Hilbert-space operations are taken under the true law $P_0$ unless a path or training-sample average is explicitly indicated. 

\subsection{Notation and function spaces}\label{app:notation}

The following table summarizes the notation used in Section \ref{sec:measurement-efficiency} and this appendix.
{\small\renewcommand{\arraystretch}{1.14}
\begin{longtable}{@{}>{\raggedright\arraybackslash}p{.23\linewidth}>{\raggedright\arraybackslash}p{.70\linewidth}@{}}
\toprule
Symbol & Definition and interpretation\\
\midrule\endfirsthead
\toprule Symbol & Definition and interpretation\\\midrule\endhead
\bottomrule\endfoot
$P_0$, $P_{W_j}$ & True joint law and its marginal law of $W_j$. All $L^2$ spaces identify functions equal almost surely.\\
$a(Z)$, $\tau_0$ & Known treatment-contrast weight in \eqref{eq:weight}; benchmark-scale target in \eqref{eq:target-benchmark}.\\
$\mu(z)$ & The function $z\mapsto\E[Y_1\mid Z=z]$.\\
$\phi_{0j}$, $u_j$ & Population bridge and its observed discrepancy $u_j=Y_1-\phi_{0j}(Y_j)$.\\
$L^2(P_{Y_j})$, $T_j$ & Full square-integrable bridge space; conditional-expectation map $L^2(P_{Y_j})\to L^2(P_{W_j})$.\\
$T_j^*$ & $T_j^*v=\E[v(W_j)\mid Y_j]$, the adjoint conditional-expectation map.\\
$\calN_j$ & $\{v\in L^2(P_{W_j}):\E[v(W_j)\mid Y_j]=0\}$, the weights insensitive to every square-integrable bridge direction.\\
$s$, $h_j=\dot\phi_j$ & Likelihood score of a distribution path and derivative of its bridge in $L^2(P_{Y_j})$.\\
$\Hpi$, $\calT$, $\mathcal S_{\mathrm{lin}}$ & Ambient score space, closed span of actual model scores, and solutions to the exact linearized bridge equations.\\
$\calV$ & Closed space of admissible combined residual corrections, defined in \eqref{eq:residual-space}.\\
$\psib$, $\psieff$ & Benchmark and efficient influence functions. They are mean-zero random variables.\\
$g_*$, $\Ve$ & Projection $\Pi_{\calV}\psib$; squared $L^2$ norm of $\psieff$.\\
$\Pi_{\mathcal U}$ & Orthogonal projection onto a closed subspace $\mathcal U$. \\
$\mathbf T$, $\mathsf A$ & Block structural map and joint score-to-moment map, defined with domains and codomains in \eqref{eq:joint-operators}.\\
$\mathsf R$ & Residual-sum operator defined in \eqref{eq:residual-map}; it maps a tuple of weight functions to a single random correction. It is used only in the appendix.\\
$m$, $\psi$ & Uncentered estimating contribution in \eqref{eq:contribution} and its centered version $m-\tau$. At appropriate population nuisances the latter is an IF.\\
$\ell_{0j}$, $q_j$ & Riesz representer of bridge-expression sensitivity, $\ell_{0j}=\E[a(Z)\mid Y_j]$, and an adjoint solution $q_j=a+v_j$.\\
$\mathbb P_{\mathrm{tr}}$, $\Pn$ & Training-sample and full-sample empirical averages. For a fitted function, $P_0f$ means expectation on an independent observation conditional on training.\\
\end{longtable}}

All spaces $L^2(P_{Y_j})$ and Hilbert-space operations are evaluated at the true marginal laws. Along a regular submodel, bridges and their derivatives are compared in these fixed reference spaces, with the required integrability also maintained under the nearby distributions. The population bridges are otherwise unrestricted square-integrable functions.

\subsection{Density scores, product-rule differentiation, and the benchmark IF}\label{app:density}

\subsubsection{A precise definition of a regular density path}
Let $\nu$ be a dominating measure and $p_t$ a density of $P_t\in\calM$. A standard regularity requirement is differentiability in quadratic mean at zero: for some $s\in L_0^2(P_0)$,
\begin{equation}\label{eq:DQM}
\int\left\{\sqrt{p_t(o)}-\sqrt{p_0(o)}-\frac t2s(o)\sqrt{p_0(o)}\right\}^2\dd\nu(o)=o(t^2).
\end{equation}
Differentiating expectations of unbounded integrands also requires the integrability and interchange conditions maintained in the local-path assumption. For smooth dominated paths the function $s$ is the log-density derivative $s(O)=\left.\frac{\partial}{\partial t}\log p_t(O)\right|_{t=0}$. The tangent space is the closure of the span of the actual scores.

For a bridge-specific conditioning value $W_j=w$, write the remaining integration variables as $r$. The restriction is
\[
0=\int\{y_1-\phi_{jt}(y_j)\}\frac{p_t(r,w)}{p_{t,W_j}(w)}\dd r.
\]
Assume differentiation can pass through the integral and $\dot\phi_j\in L^2(P_{Y_j})$ exists. At $t=0$, the product and quotient rules give
\begin{align}
0={}&-\int\dot\phi_j(y_j)\frac{p_0(r,w)}{p_{0,W_j}(w)}\dd r\notag\\
&+\int u_j(r,w)\frac{\dot p(r,w)}{p_{0,W_j}(w)}\dd r\notag\\
&-\frac{\dot p_{W_j}(w)}{p_{0,W_j}(w)}
\int u_j(r,w)\frac{p_0(r,w)}{p_{0,W_j}(w)}\dd r.\label{eq:moment-product}
\end{align}
Here $\dot p=p_0s$. Integrating the joint density derivative gives
\[
\frac{\dot p_{W_j}(w)}{p_{0,W_j}(w)}
=\int s(r,w)\frac{p_0(r,w)}{p_{0,W_j}(w)}\dd r
=\E[s\mid W_j=w].
\]
The three terms in \eqref{eq:moment-product} are therefore
\(-T_j\dot\phi_j(w)\),\(\E[u_js\mid W_j=w]\), and \(-\E[s\mid W_j=w]\E[u_j\mid W_j=w]\). The third term is zero by the population bridge equation. Consequently $T_j\dot\phi_j=\E[u_js\mid W_j]$ for every $j$, proving \eqref{eq:linearized}. The same calculation yields, for any differentiable integrand $f_t(O)$,
\begin{equation}\label{eq:conditional-derivative}
\left.\partial_t\E_t[f_t(O)\mid W_j]\right|_0
=\E[\dot f\mid W_j]+\E[f_0\{s-\E[s\mid W_j]\}\mid W_j].
\end{equation}

\subsection{Joint tangent-space closure: explicit operators and both inclusions}\label{app:joint}

To understand assumptions \ref{ass:tangent} and proposition \ref{prop:tangent}, it is useful to begin with one bridge, based on \cite{severini2012efficiency}. Suppress its index, write $\omega(W)=\E[u^2\mid W]$, and suppose $0<c_0\leq\omega\leq C_0$. For $b\in L^2(P_W)$ define the random score adjustment
\[
(\mathsf Lb)(O)=\frac{u(O)}{\omega(W)}b(W).
\]
The conditional moment restriction implies $\E[\mathsf Lb\mid W]=0$. Since $Z$ is part of $W$, this adjustment belongs to $\Hpi$. Its squared norm and its induced conditional-moment change are
\[
\norm{\mathsf Lb}_2^2=\E[b(W)^2/\omega(W)]\leq c_0^{-1}\norm b_2^2,
\text{ and }
\E[u\mathsf Lb\mid W]=b(W).
\]
Thus the score-to-moment map has a bounded right inverse, although
$T$ itself need not have a bounded inverse or a closed range.

Let $s\in\Hpi$ satisfy $\E[us\mid W]\in\overline{\Ran(T)}$, and
choose $h_m$ with $Th_m\to\E[us\mid W]$.
The corrected scores are
\(
s_m=s+\mathsf L\{Th_m-\E[us\mid W]\}.
\)
They belong to $\mathcal S_{\mathrm{lin}}$ and converge to $s$ in
$L^2(P_0)$. For the reverse inclusion, if $s_m\to s$ and
$\E[us_m\mid W]\in\Ran(T)$, the upper bound on $\omega$ gives
\(
\norm{\E[u(s_m-s)\mid W]}_2^2
\leq C_0\norm{s_m-s}_2^2\longrightarrow0.
\)
Therefore $\E[us\mid W]\in\overline{\Ran(T)}$, proving both
inclusions in the scalar analytic closure identity.

\paragraph{The joint operators.}

Define $\calG_j=L^2(P_{W_j})$ and the product spaces
\[
\mathcal H_Y=\bigoplus_{j=2}^JL^2(P_{Y_j}),\qquad
\calG=\bigoplus_{j=2}^J\calG_j.
\]
A product-space element is a tuple; its squared norm is the sum of its component squared norms. For example, $q=(q_2,\ldots,q_J)\in\calG$ has $\norm q_{\calG}^2=\sum_j\E[q_j(W_j)^2]$. The block structural operator and the joint score-to-moment operator are
\begin{equation}\label{eq:joint-operators}
\begin{aligned}
\mathbf T:&\ \mathcal H_Y\longrightarrow\calG,
&\mathbf T h&=(T_jh_j)_{j=2}^J,\\
\mathsf A:&\ \Hpi\longrightarrow\calG,
&(\mathsf A s)_j&=\E[u_js\mid W_j].
\end{aligned}
\end{equation}
The first takes a tuple of bridge changes and returns their conditional means. The second takes \emph{one} likelihood score and returns all its conditional-moment changes. The linearized restriction is $\mathsf A s=\mathbf T\dot\phi$. 

Conditional Cauchy--Schwarz gives
\[
|\E[u_js\mid W_j]|^2\leq\E[u_j^2\mid W_j]\E[s^2\mid W_j]
\leq C_0\E[s^2\mid W_j].
\]
Thus $\norm{\mathsf A s}_{\calG}^2\leq(J-1)C_0\norm s_2^2$. The operator is bounded. Since $J$ is finite,
\(
\overline{\Ran(\mathbf T)}=\bigoplus_{j=2}^J\overline{\Ran(T_j)}.
\)

One sufficient condition for simultaneous closure, used as an alternative to closed component ranges in Assumption~\ref{ass:tangent}, is
\begin{equation}\label{eq:joint-coercivity}
\E\!\left[\left\{\sum_{j=2}^J u_jq_j(W_j)\right\}^2\right]
\geq c_0\sum_{j=2}^J\norm{q_j}_2^2
\quad\text{for every }q\in\calG,\quad c_0>0.
\end{equation}
This condition rules out near-cancellation in the \emph{joint} residual map. Separate lower bounds on $\E[u_j^2\mid W_j]$ do not imply it. It is sufficient, not necessary. In finite-support models the closed-range alternative can be used even when the joint condition fails.

\begin{lemma}[Joint analytic closure]\label{lem:joint-closure}
Under the moment bound in Assumption~\ref{ass:tangent}, if all component ranges are closed or \eqref{eq:joint-coercivity} holds, then
\begin{equation}\label{eq:joint-closure}
\overline{\mathcal S_{\mathrm{lin}}}
=\{s\in\Hpi:\ (\mathsf A s)_j\in\overline{\Ran(T_j)}\text{ for every }j\}.
\end{equation}
\end{lemma}
\begin{proof}
For the first inclusion, let $s_m\in\mathcal S_{\mathrm{lin}}$ with $s_m\to s$ in $L^2(P_0)$. The space $\Hpi$ is closed, so $s\in\Hpi$. Continuity gives $\mathsf A s_m\to\mathsf A s$. Each image lies in $\Ran(\mathbf T)$, hence $\mathsf A s$ lies in its closure. This proves the inclusion into the right side of \eqref{eq:joint-closure}.

If every component range is closed, their finite product is closed. Its inverse image under the continuous map $\mathsf A$ is closed, so $\mathcal S_{\mathrm{lin}}$ already equals the right side. This proves the reverse inclusion for the first alternative.

For the second alternative, first compute the adjoint of $\mathsf A$ on the ambient space $\Hpi$. For $s\in\Hpi$ and $q\in\calG$,
\begin{align*}
\ip{\mathsf A s}{q}_{\calG}
&=\sum_j\E[q_j(W_j)\E[u_js\mid W_j]]\\
&=\E\!\left[s\sum_j u_jq_j(W_j)\right].
\end{align*}
The random sum is in $\Hpi$, since $Z$ is in every $W_j$ and $\E[u_j\mid W_j]=0$. Therefore
\begin{equation}\label{eq:Astar}
\mathsf A^*q=\sum_j u_jq_j(W_j):\calG\longrightarrow\Hpi.
\end{equation}
The joint lower bound is $\ip q{\mathsf A\mathsf A^*q}_{\calG}\geq c_0\norm q_{\calG}^2$. The operator $\mathsf A\mathsf A^*$ is bounded, self-adjoint, and positive. The lower bound gives closed range and zero kernel. A self-adjoint operator's range is dense when its kernel is zero, so its range is all of $\calG$ and its inverse is bounded. Define
\[
\mathsf L=\mathsf A^*(\mathsf A\mathsf A^*)^{-1}:\calG\longrightarrow\Hpi.
\]
It satisfies \(\mathsf A\mathsf L=I\) and 
\(
\norm{\mathsf Lb}_2^2
=\ip b{(\mathsf A\mathsf A^*)^{-1}b}_{\calG}
\leq c_0^{-1}\norm b_{\calG}^2.
\) Take $s$ in the right side of \eqref{eq:joint-closure}. Choose $h_m\in\mathcal H_Y$ such that $\mathbf T h_m\to\mathsf A s$, and set
\(
s_m=s+\mathsf L(\mathbf T h_m-\mathsf A s).
\)
Every $s_m$ is in $\Hpi$, and
\(
\mathsf A s_m=\mathsf A s+\mathbf T h_m-\mathsf A s=\mathbf T h_m.
\)
Thus $s_m\in\mathcal S_{\mathrm{lin}}$. Furthermore,
\[
\norm{s_m-s}_2\leq c_0^{-1/2}\norm{\mathbf T h_m-\mathsf A s}_{\calG}\longrightarrow0.
\]
This proves the reverse inclusion and completes the closure proof.
\end{proof}

\subsubsection{Proof of Proposition~\ref{prop:tangent}}
\begin{proof}[Proof of Proposition~\ref{prop:tangent}]
Along a regular path admitting square-integrable
bridge derivatives, differentiating the bridge restrictions gives
\(
\E[u_js\mid W_j]=T_j\dot\phi_j,
\qquad j=2,\ldots,J.
\)
Thus, such scores belong to $\mathcal S_{\mathrm{lin}}$. Assumption~\ref{ass:tangent}(ii) ensures that these regular paths generate the tangent space and that the closed span of their scores coincides with the closure of the linearized score set. Consequently,
\(
\calT=\overline{\mathcal S_{\mathrm{lin}}}.
\)

It remains to characterize this closure. By Assumption~\ref{ass:tangent}(i), the score-to-moment map $s\mapsto\E[u_js\mid W_j]$ is continuous from
$L^2(P_0)$ to $L^2(P_{W_j})$, since conditional Cauchy--Schwarz gives
\[
\begin{aligned}
\left\|\E[u_js\mid W_j]\right\|_2^2
&\leq
\E\left[
  \E[u_j^2\mid W_j]\E[s^2\mid W_j]
\right]\\
&\leq C_0\|s\|_2^2.
\end{aligned}
\]
Moreover, Assumption~\ref{ass:tangent}(iii) ensures that the conditions for joint analytic closure are satisfied, either through closed operator ranges, joint residual nondegeneracy, or a direct verification of the closure identity. Therefore, Lemma~\ref{lem:joint-closure} implies
\[
\overline{\mathcal S_{\mathrm{lin}}}
=
\left\{
s\in\Hpi:
\E[u_js\mid W_j]\in\overline{\Ran(T_j)}
\text{ for every }j
\right\}.
\]
Combining the two equalities establishes
\eqref{eq:tangent-characterization}.
\end{proof}


\subsubsection{Proof of Proposition~\ref{prop:corrections}}\label{app:complement}

\begin{proof}[Proof of Proposition~\ref{prop:corrections}]
For each $j$, the defining adjoint identity gives
\begin{align*}
q\perp\Ran(T_j)
&\iff\ip q{T_jh}=0\quad\forall h\in L^2(P_{Y_j})\\
&\iff\ip{T_j^*q}h=0\quad\forall h\in L^2(P_{Y_j})\\
&\iff T_j^*q=0.
\end{align*}
Taking orthogonal complements yields
\begin{equation}\label{eq:range-null}
\overline{\Ran(T_j)}=\Null(T_j^*)^\perp=\calN_j^\perp.
\end{equation}

By Proposition~\ref{prop:tangent}, $s\in\calT$ exactly when $s\in\Hpi$ and
\[
\E[v_j(W_j)\E[u_js\mid W_j]]=0
\quad\text{for every }v_j\in\calN_j\text{ and every }j.
\]
Iterated expectation converts this to $\E[s u_jv_j(W_j)]=0$. Being orthogonal to every generating direction is equivalent to being orthogonal to every finite sum and then to its $L^2$ closure. Hence $\calT=\Hpi\cap\calV^\perp$.

Each residual sum is in $\Hpi$ because $\E[u_j\mid W_j]=0$ and $Z$ is part of $W_j$. Conditional expectation is continuous in $L^2$, so every limit in $\calV$ remains conditionally centered given $Z$. Thus $\calV$ is a closed subspace of $\Hpi$, and taking the orthogonal complement within $\Hpi$ gives $\calT^{\perp_{\Hpi}}=\calV$.

For any centered $f\in L_0^2(P_0)$,
\(
f=\{f-\E[f\mid Z]\}+\E[f\mid Z].
\)
The first term is in $\Hpi$, the second in $L_0^2(P_Z)$, and the two are orthogonal. Combining this decomposition with the preceding result gives
\begin{equation}\label{eq:full-complement}
\calT^\perp=L_0^2(P_Z)\oplus\calV.
\end{equation}

For the statement about all representers, let $\widetilde\psi\in\Hpi$ represent the same target derivative as $\psib$. Subtracting their derivative identities gives $\E[(\widetilde\psi-\psib)s]=0$ for every score. Continuity extends the identity to the closed score span $\calT$. Thus $\widetilde\psi-\psib\in\calT^{\perp_{\Hpi}}=\calV$. Conversely, adding any $g\in\calV$ preserves every derivative identity. Since elements of $\calV$ are mean zero, the resulting functions are influence-function representers. Taking $g=u_jv_j(W_j)$ gives \eqref{eq:single-gradient}.
\end{proof}

\subsection{Canonical-gradient projection and variance optimization}\label{app:efficiency}

For the operator proofs below, it is useful to name the map represented by the explicit sum in the main text. Define
\begin{equation}\label{eq:residual-map}
\mathsf R:\underbrace{\bigoplus_{j=2}^J\calN_j}_{\displaystyle\calN}\longrightarrow\Hpi,
\qquad
\mathsf R(v_2,\ldots,v_J)=\sum_{j=2}^J u_jv_j(W_j).
\end{equation}
The input is a tuple of admissible weight functions, one for each bridge. The output is a \emph{single scalar random variable}, the combined residual correction. For two auxiliary measurements this is simply $\mathsf R(v_2,v_3)=u_2v_2(W_2)+u_3v_3(W_3)$. 

The domain $\calN$ has norm $\norm v_{\calN}^2=\sum_j\norm{v_j}_2^2$. Since $(\sum_{j=2}^Jx_j)^2\leq(J-1)\sum_{j=2}^Jx_j^2$,
\[
\norm{\mathsf Rv}_2^2\leq(J-1)\sum_j\E[u_j^2v_j(W_j)^2]
\leq(J-1)C_0\sum_j\norm{v_j}_2^2.
\]
Thus the map is bounded. Its image has closure $\calV=\overline{\Ran(\mathsf R)}$. Even when component ranges are closed, a sum of closed subspaces need not be closed in an infinite-dimensional problem, so this final closure must be retained.

The adjoint $\mathsf R^*:\Hpi\to\calN$ is
\begin{equation}\label{eq:Rstar}
(\mathsf R^*f)_j=\Pi_{\calN_j}\E[u_jf\mid W_j].
\end{equation}
Indeed, the conditional moment belongs to $L^2(P_{W_j})$ by the upper residual-variance bound, and for $v\in\calN$,
\begin{align*}
\ip{\mathsf Rv}f
&=\sum_j\E[v_j(W_j)\E[u_jf\mid W_j]]\\
&=\sum_j\ip{v_j}{\Pi_{\calN_j}\E[u_jf\mid W_j]}
=\ip v{\mathsf R^*f}_{\calN}.
\end{align*}

\subsubsection{Proof of Theorem~\ref{thm:efficiency}}

\begin{proof}[Proof of Theorem~\ref{thm:efficiency}]
The benchmark IF represents the derivative on $\Hpi$, hence on its subspace $\calT$. The orthogonal-complement result gives the decomposition $\Hpi=\calT\oplus\calV$. Let $g_*=\Pi_{\calV}\psib$. Then $\psib-g_*\in\calT$ and, for every $s\in\calT$,
\(
\E[(\psib-g_*)s]=\E[\psib s]=\dot\tau(s).
\)
Therefore it is both a derivative representer and an element of the tangent space. It is the canonical gradient.

Proposition~\ref{prop:corrections} has already identified all representers in $\Hpi$ as $\psib+\calV$. It remains to minimize their variance. For an admissible tuple $v$, write
\(
\psib+\mathsf Rv=\psieff+(g_*+\mathsf Rv).
\)
The first term is in $\calT$ and the second in $\calV$, so their inner product is zero. Expanding the squared norm gives
\[
\norm{\psib+\mathsf Rv}_2^2
=\norm{\psieff}_2^2+\norm{g_*+\mathsf Rv}_2^2,
\]
which is \eqref{eq:excess-risk}. Taking $v=0$ gives $\Ve=\Vb-\norm{g_*}_2^2$. Since $\mathsf R(\calN)$ is dense in $\calV$, a sequence of admissible tuples can approximate $-g_*$, proving the infimum characterization in \eqref{eq:bound}.

The excess risk equals zero if and only if $\mathsf Rv=-g_*$. This proves the attainment statement. If two tuples attain it, their difference is in $\Null(\mathsf R)$, so their combined random corrections coincide.
\end{proof}

The general uniqueness argument can also be stated directly. If $\psi\in\calT$ and $\widetilde\psi$ both represent the derivative, then $d=\widetilde\psi-\psi\in\calT^\perp$. Hence
\(
\E[\widetilde\psi^2]=\E[(\psi+d)^2]=\E[\psi^2]+\E[d^2].
\)

\subsubsection{Proof of Corollary~\ref{cor:gain}}\label{app:gains}

\begin{proof}[Proof of Corollary~\ref{cor:gain}]
Theorem~\ref{thm:efficiency} gives $\Vb-\Ve=\norm{\Pi_{\calV}\psib}_2^2$. This is zero exactly when $\psib\perp\calV$. By the definition of $\calV$, that orthogonality is equivalent to
$\E[\psib u_jv_j(W_j)]=0$ for every generating direction. Since
\(
\E[\psib u_jv_j(W_j)]
=\ip{c_j}{v_j},
\)
it holds for every $v_j\in\calN_j$ exactly when $c_j\in\calN_j^\perp=\overline{\Ran(T_j)}$. If a generating direction $r$ has nonzero covariance with $\psib$, minimizing over its scalar coefficient gives the strictly positive reduction. Conversely, if every generating covariance vanishes, no combination or $L^2$ limit can have nonzero covariance.
\end{proof}

\subsection{Estimation}\label{app:implementation}

All fits in this subsection are computed on one training fold. The main-text vectors $p_j(Y_j)\in\R^{d_j}$ and $b_j(W_j)\in\R^{m_j}$ specify the bridge and instrument bases. Each coordinate of $p_j$ is square-integrable under $P_{Y_j}$, and each coordinate of $b_j$ is square-integrable under $P_{W_j}$. A finite-dimensional bridge approximation is $p_j^\top\theta_j$, and a correction is $b_j^\top\beta_j$.

\paragraph{The bridge loss.}
The identity
\(
\sup_f\{\ip f g-\tfrac12\norm f_2^2\}
=\tfrac12\norm g_2^2
\)
follows by writing its objective as $\tfrac12\norm g_2^2-\tfrac12\norm{f-g}_2^2$. Substituting $g(W_j)=\E[\phi_j(Y_j)-Y_1\mid W_j]$ proves \eqref{eq:minimax}. For a finite test basis $b_j$, the empirical minimax objective is a quadratic minimum-distance criterion for the sample moments involving $b_j(W_j)\{p_j(Y_j)^\top\theta_j-Y_1\}$. Suitable penalties stabilize its inverse problem. 

\paragraph{The finite-dimensional correction program.}
Let $d_v=\sum_{j=2}^J m_j$ be the total number of correction coefficients, and stack
\(
\beta=(\beta_2^\top,\ldots,\beta_J^\top)^\top\in\R^{d_v},\text{ and }
\widehat M=\operatorname{diag}(\widehat M_2,\ldots,\widehat M_J).
\)
The block-diagonal matrix $\widehat M$ has $\sum_jd_j$ rows and $d_v$ columns. Its feasible set is $\ker\widehat M$. Let
\begin{equation}\label{eq:basis-H}
\widehat H(O)=
(\widehat u_2b_2(W_2)^\top,\ldots,\widehat u_Jb_J(W_J)^\top)^\top\in\R^{d_v}.
\end{equation}
Each block is the fitted bridge discrepancy multiplied by the instrument basis for that bridge. Thus $\widehat H^\top\beta$ is exactly the combined estimated residual correction in the objective. Define the training second-moment matrix, cross-moment vector, and penalty matrix by
\begin{align}
\widehat S&=\mathbb P_{\mathrm{tr}}[\widehat H\widehat H^\top],&
\widehat c&=\mathbb P_{\mathrm{tr}}[\widehat H\widehat\psi_{\mathrm b}],\notag\\
\widehat G&=\operatorname{diag}\{\mathbb P_{\mathrm{tr}}[b_j(W_j)b_j(W_j)^\top]:j=2,\ldots,J\}.&&\label{eq:sample-matrices}
\end{align}
The off-diagonal blocks of $\widehat S$ retain cross-bridge dependence. The block structure of $\widehat G$ reflects the sum of squared weight-function penalties.

Let the columns of $\widehat B\in\R^{d_v\times r}$ be a basis of $\ker\widehat M$, where $r=d_v-\operatorname{rank}(\widehat M)$. Every feasible coefficient vector can be written $\beta=\widehat B\zeta$ for some $\zeta\in\R^r$. Substituting into \eqref{eq:sample-optimization} gives a constant plus
\(
2\zeta^\top\widehat B^\top\widehat c
+\zeta^\top\widehat B^\top(\widehat S+\lambda_n\widehat G)\widehat B\zeta.
\)
Differentiation yields the solution
\begin{equation}\label{eq:sample-solution}
\widehat\beta=-\widehat B
\left\{\widehat B^\top(\widehat S+\lambda_n\widehat G)\widehat B\right\}^{\dagger}
\widehat B^\top\widehat c.
\end{equation}

On a finite measurement support, when $p_j$ contains all category indicators,
\[
T_j^*v_j=0\iff\E[p_j(Y_j)v_j(W_j)]=0\iff M_j\beta_j=0.
\]
When $L^2(P_{Y_j})$ is infinite-dimensional, a finite basis tests only its own span. To see the untested error explicitly, let a population approximation be
\(
\phi_{0j}(Y_j)=p_j(Y_j)^\top\theta_{0j}+r_j(Y_j),
\)
where $r_j$ denotes the approximation residual. With $\widehat\phi_j=p_j^\top\widehat\theta_j$ and $\widehat v_j=b_j^\top\widehat\beta_j$,
\begin{align*}
\E[\widehat v_j(\widehat\phi_j-\phi_{0j})]
&=(\widehat\theta_j-\theta_{0j})^\top M_j\widehat\beta_j-\E[\widehat v_jr_j]\\
&=(\widehat\theta_j-\theta_{0j})^\top(M_j-\widehat M_j)\widehat\beta_j-\E[\widehat v_jr_j],
\end{align*}
since $\widehat M_j\widehat\beta_j=0$. Therefore
\begin{equation}\label{eq:sieve-bias}
\begin{split}
|\E[\widehat v_j(\widehat\phi_j-\phi_{0j})]|
\leq{}&\norm{\widehat\theta_j-\theta_{0j}}\,
\norm{M_j-\widehat M_j}_{\mathrm{op}}\,
\norm{\widehat\beta_j}\\
&+\norm{\widehat v_j}_2\norm{r_j}_2.
\end{split}
\end{equation}
Here vector norms are Euclidean and $\norm\cdot_{\mathrm{op}}$ is the induced matrix operator norm. A growing-basis argument must bound it as well as moment estimation error and coefficient growth.

\subsubsection{Proof of Theorem~\ref{thm:estimator}}\label{app:inference}



\begin{proof}[Proof of Theorem~\ref{thm:estimator}]
Write $m_*=\tau_0+\psieff$, so $P_0m_*=\tau_0$. Let $\widehat m_k$ denote the fitted contribution for fold $k$, let $\Pn^{(k)}$ be its evaluation average, and set $w_k=|I_k|/n$. The exact decomposition is
\begin{align}
\widehat\tau-\tau_0
={}&(\Pn-P_0)m_*\notag\\
&+\sum_{k=1}^K w_k(\Pn^{(k)}-P_0)(\widehat m_k-m_*)\notag\\
&+\sum_{k=1}^K w_k(P_0\widehat m_k-\tau_0).\label{eq:crossfit-decomposition}
\end{align}
The leading term is the empirical mean of the fixed square-integrable EIF. The ordinary i.i.d. central limit theorem gives its $\sqrt n$ limit $N(0,\Ve)$.

For a fixed fold, evaluation observations are independent of the training data used to fit $\widehat m_k$. Conditional on that training data,
\begin{align*}
&\Var\!\left[\sqrt n\,w_k(\Pn^{(k)}-P_0)(\widehat m_k-m_*)\,\middle|\,\mathrm{training}\right]\\
&\qquad=\frac{nw_k^2}{|I_k|}\Var_{P_0}(\widehat m_k-m_*)
\leq w_k\norm{\widehat m_k-m_*}_2^2=o_p(1).
\end{align*}
Conditional Chebyshev's inequality, applied on events where the random norm bound is small, makes each centered fold term $o_p(1)$ after multiplication by $\sqrt n$. There are finitely many folds, so the sum is $o_p(1)$.

The last terms of \eqref{eq:crossfit-decomposition} are $o_p(n^{-1/2})$ by the exact bias identity and \eqref{eq:bias-condition}. This proves the expansion and normal limit in \eqref{eq:CLT}.

For variance estimation, conditional Markov's inequality and score consistency yield
\[
\frac1n\sum_{i=1}^n\{\widehat m_i-m_*(O_i)\}^2=o_p(1).
\]
The ordinary law of large numbers gives $\Pn[\psieff^2]\to_p\Ve$. Expanding the estimated squares, Cauchy--Schwarz bounds cross-products between the estimated-score error and the true EIF. Centering by $\widehat\tau$ instead of $\tau_0$ is negligible since the expansion implies $\widehat\tau\to_p\tau_0$. Therefore $\widehat V\to_p\Ve$. Slutsky's theorem proves the stated confidence-interval coverage.
\end{proof}


\subsection{Finite-support regularity}\label{app:finite}

Let $\{y_{j1},\ldots,y_{jd_j}\}$ be the support of $Y_j$ and use the full indicator basis
\[
p_j(y)=\big(\ind\{y=y_{j1}\},\ldots,\ind\{y=y_{jd_j}\}\big)^\top.
\]
It spans all of $L^2(P_{Y_j})$, so $d_j$ is the number of measurement categories. Let $b_j(W_j)$ contain indicators of all $m_j$ support points of $W_j$. Use $p$ for the finite vector of joint cell probabilities. Define the conditional expectation matrix and benchmark conditional-mean vector
\begin{equation}\label{eq:finite-T}
\mathsf T_j(p)_{w,\cdot}=\E_p[p_j(Y_j)^\top\mid W_j=w],\qquad
g_j(p)_w=\E_p[Y_1\mid W_j=w].
\end{equation}
With the full indicator basis, $\mathsf T_j(p)_{w,r}=P_p(Y_j=y_{jr}\mid W_j=w)$ and $\theta_{j,r}=\phi_j(y_{jr})$. Thus the operator is an ordinary conditional-probability matrix, and the coefficients are the unrestricted bridge values at the categories. The bridge system is
\begin{equation}\label{eq:finite-system}
\mathsf T_j(p)\theta_j=g_j(p).
\end{equation}
All entries are smooth functions of $p$ near a positive-probability law.

\begin{assumption}[Observable rank and covariance conditions]\label{ass:finite}
Each conditional-probability matrix $T_j(p_0)$ has full column rank.
For each $j$, choose a set $I_j$ of $d_j$ rows for which
$T_{j,I_j}(p_0)$ is invertible, and define, in a sufficiently small
neighborhood of $p_0$,
\[
 \theta_j(p)=T_{j,I_j}(p)^{-1}g_{j,I_j}(p),
 \qquad
 F_j(p)=g_{j,I_j^c}(p)-T_{j,I_j^c}(p)\theta_j(p).
\]
Let $F(p)$ stack these remaining bridge restrictions. Here $p$ is
expressed in free joint-probability coordinates that respect the
known treatment probabilities. There are a neighborhood
$\mathcal U$ of $p_0$ and a subvector $F^\circ$ of $F$ such that
\[
 F(p)=0
 \quad\Longleftrightarrow\quad
 F^\circ(p)=0
 \qquad\text{for every }p\in\mathcal U,
\]
and $DF^\circ(p_0)$ has full row rank. The selected system may be
empty when the bridge restrictions hold throughout $\mathcal U$.
\end{assumption}

The full-column-rank condition identifies each bridge locally.
The condition on $F^\circ$ requires the observable bridge model
to admit a locally equivalent system of equations with independent
derivatives. 

\begin{proposition}[Realizability of finite-support scores]\label{prop:finite-tangent}
Under Assumption~\ref{ass:finite}, the observable bridge model is locally a smooth submanifold of the known-randomization probability simplex. Its tangent space is
\begin{equation}\label{eq:finite-tangent}
\calT=\{s:\E[s\mid Z]=0,\ \E[u_js\mid W_j]\in\Ran(T_j)\text{ for every }j\}.
\end{equation}
Every direction on the right is realized by a regular local density path. Thus Assumption~\ref{ass:tangent} holds through its closed-range alternative.
\end{proposition}
\begin{proof}
Full column rank and positivity persist in a sufficiently small neighborhood of $p_0$. The pivot formulas for $\theta_j(p)$ are smooth there. A probability vector in this neighborhood admits all bridges if and only if $F(p)=0$, equivalently $F^\circ(p)=0$.
The implicit function theorem therefore makes the observable bridge model a smooth local manifold with tangent space $\ker DF^\circ(p_0)$ in free probability coordinates.

Because $F^\circ$ is a subvector of $F$,
$\ker DF(p_0)\subseteq\ker DF^\circ(p_0)$. Conversely, every vector in $\ker DF^\circ(p_0)$ is the derivative of a smooth curve on $\{p:F^\circ(p)=0\}$. All components of $F$ vanish along that curve by local equivalence, so differentiation gives the reverse inclusion. Thus
\[
 \ker DF^\circ(p_0)=\ker DF(p_0).
\]

An allowed free-coordinate perturbation $\dot p$ induces a change in each joint-cell probability; write $\dot p(o)$ for this induced change. The corresponding score is $s(o)=\dot p(o)/p_0(o)$, which satisfies $E[s\mid Z]=0$. Differentiating the bridge equations at $p_0$ gives
\[
 E[u_js\mid W_j]=T_j\dot\theta_j,
 \qquad j=2,\ldots,J.
\]
The pivot rows determine $\dot\theta_j$ uniquely. The remaining rows hold exactly when $DF_j(p_0)\dot p=0$. Consequently the tangent space is
\[
 \mathcal T=
 \left\{s:E[s\mid Z]=0,\quad
 E[u_js\mid W_j]\in\mathcal R(T_j)
 \text{ for every }j\right\}.
\]
Every such direction is realized by a smooth curve on the local manifold. The curve has positive cell probabilities for a sufficiently small path parameter, and its bridges are differentiable by the pivot
formulas. On fixed finite positive support, the probability path is quadratic-mean differentiable with score $s$. All relevant moments are finite, and all operator ranges are finite dimensional and closed. These facts establish Assumption \ref{ass:tangent} through its closed-range alternative.
\end{proof}

A useful sufficient check for full row rank of the reduced derivative is that the full joint residual map $q\mapsto\sum_j u_jq_j(W_j)$ has zero kernel. In finite dimensions its adjoint $\mathsf A$ is then onto. 

\paragraph{A fully specified categorical estimator.}
On each training sample, use empirical joint-cell frequencies to form
$\widehat T_j$, $\widehat g_j$, and
$\widehat D_j=\operatorname{diag}\{
\mathbb P_{\mathrm{tr}}(W_j=w):w\in\operatorname{supp}(W_j)\}$.
With the full indicator bases $p_j$ and $b_j$, estimate
\[
 \widehat\theta_j=
 (\widehat T_j^{\top}\widehat D_j\widehat T_j)^{-1}
 \widehat T_j^{\top}\widehat D_j\widehat g_j,
 \text{ and }
 \widehat\varphi_j(y)=p_j(y)^{\top}\widehat\theta_j.
\]
Use empirical treatment-arm means for $\widehat\mu(z)$. Define
\begin{align*}
 \widehat u_j(O)&=Y_1-\widehat\varphi_j(Y_j),
 &\widehat\psi_b(O)&=a(Z)\{Y_1-\widehat\mu(Z)\},\\
 \widehat H(O)&=
 \bigl(\widehat u_2 b_2(W_2)^{\top},\ldots,
       \widehat u_J b_J(W_J)^{\top}\bigr)^{\top},
 &\widehat M_j&=\mathbb P_{\mathrm{tr}}[p_jb_j^{\top}],\\
 \widehat M&=\operatorname{diag}
       (\widehat M_2,\ldots,\widehat M_J),
 &\widehat G&=\operatorname{diag}_{j=2,\ldots,J}
       \{\mathbb P_{\mathrm{tr}}[b_jb_j^{\top}]\}.
\end{align*}
Choose a deterministic sequence satisfying
\begin{equation}
 \lambda_n>0,\qquad \lambda_n\longrightarrow0,
 \qquad \sqrt n\,\lambda_n\longrightarrow\infty.
 \label{eq:categorical-ridge-rate}
\end{equation}
For example, $\lambda_n=n^{-1/4}$ satisfies these conditions. Solve
\[
 \widehat\beta=
 \operatorname*{arg\,min}_{\widehat M\beta=0}
 \left\{
 \mathbb P_{\mathrm{tr}}
   [(\widehat\psi_b+\widehat H^{\top}\beta)^2]
 +\lambda_n\beta^{\top}\widehat G\beta
 \right\},
 \qquad
 \widehat v_j(w)=b_j(w)^{\top}\widehat\beta_j.
\]
The constraints are imposed exactly. The minimizer is unique whenever
all required categories are observed, because then $\widehat G$ is
positive definite. If a required category is absent, a bridge inverse
does not exist, or $\widehat M$ does not have full row rank, set all
fitted bridges, weights, and arm means to zero for that training sample
and use the benchmark-only held-out contribution $a(Z)Y_1$.

For $i\in I_k$, let
\[
 \widehat m_i=m\bigl(O_i;
 \widehat\mu^{(-k)},\widehat\varphi^{(-k)},
 \widehat v^{(-k)}\bigr),
 \qquad
 \widehat\tau=\frac1n\sum_{i=1}^n\widehat m_i,
 \qquad
 \widehat V=\frac1n\sum_{i=1}^n(\widehat m_i-\widehat\tau)^2.
\]

\begin{theorem}[Attainment with categorical measurements]
\label{thm:finite-estimator}
Suppose Assumption~\ref{ass:finite} holds and $0<V_{\mathrm{eff}}<\infty$. Use full indicator bases for $Y_j$ and $W_j$, and a fixed number $K\geq2$ of folds formed independently of the observations, with
$|I_k|/n\to\kappa_k\in(0,1)$.  Among all feasible population minimizers, define $v^\dagger$ by
\[
 v^\dagger=
 \operatorname*{arg\,min}_{\substack{
       v_j\in\mathcal N_j,\ j=2,\ldots,J\\
       E[(\psi_b+\sum_{j=2}^J u_jv_j(W_j))^2]
                =V_{\mathrm{eff}}}}
       \sum_{j=2}^J\|v_j\|_2^2.
\]
This tuple exists and is unique. The unpenalized minimizing tuples
themselves need not be unique, but each has the same combined residual
correction, and
\[
 \psi_b+\sum_{j=2}^J u_jv_j^\dagger(W_j)=\psi_{\mathrm{eff}}
 \quad P_0\text{-almost surely}.
\]

For each fold, suppress its index and evaluate $E[\cdot]$ and
$\|\cdot\|_2$ on a fresh observation from $P_0$, holding the training
sample fixed. The following probability orders are with respect to
training-sample randomness:
\begin{align}
 \|\widehat\varphi_j-\varphi_{0j}\|_2
   &=O_p(n^{-1/2}), \qquad j=2,\ldots,J,
       \label{eq:categorical-bridge-rate}\\
 \sum_{j=2}^J\|\widehat v_j-v_j^\dagger\|_2
   &=O_p\!\left(\lambda_n+\frac{n^{-1/2}}{\lambda_n}\right)
     =o_p(1),
       \label{eq:categorical-weight-rate}\\
 E[m(O;\widehat\mu,\widehat\varphi,\widehat v)]-\tau_0
   &=O_p(n^{-1}),
       \label{eq:categorical-bias-rate}\\
 \|m(\cdot;\widehat\mu,\widehat\varphi,\widehat v)
       -(\tau_0+\psi_{\mathrm{eff}})\|_2
   &=O_p\!\left(\lambda_n+\frac{n^{-1/2}}{\lambda_n}\right)
     =o_p(1).
       \label{eq:categorical-score-rate}
\end{align}
Consequently,
\[
 \sqrt n(\widehat\tau-\tau_0)
   =\frac1{\sqrt n}\sum_{i=1}^n\psi_{\mathrm{eff}}(O_i)+o_p(1)
   \ \xrightarrow{d}\ N(0,V_{\mathrm{eff}}),
 \qquad
 \widehat V\xrightarrow{p}V_{\mathrm{eff}}.
\]
The interval
$\widehat\tau\pm z_{1-\alpha/2}\sqrt{\widehat V/n}$ has asymptotic
pointwise coverage $1-\alpha$, where $z_{1-\alpha/2}$ is the corresponding
standard-normal quantile. The estimator is regular at $P_0$ along
smooth local submodels of the observable bridge model and attains its
efficiency bound.

\end{theorem}

\begin{proof}
Fix a training fold, whose size is of order $n$.
Let $H$ stack the vectors $u_jb_j(W_j)$, and write
$S=E[HH^\top]$, $c=E[H\psi_b]$,
$G=\operatorname{diag}_{j=2,\ldots,J}\{E[b_jb_j^\top]\}$,
and $M=\operatorname{diag}(M_2,\ldots,M_J)$.
Their empirical counterparts use the fitted residuals and training-sample averages. Positivity, full column rank, and smoothness
of the categorical fits in the empirical cell frequencies give
\[
 \widehat\theta_j-\theta_{0j}=O_p(n^{-1/2}),\qquad
 \|\widehat\mu-\mu\|_\infty=O_p(n^{-1/2}),
\]
and root-$n$ consistency of $\widehat M,\widehat S,\widehat c,\widehat G$.
Finite-support norm equivalence proves the bridge rate.

Proposition \ref{prop:finite-tangent} verifies Assumption \ref{ass:tangent}. Full indicator bases span the entire admissible correction space, which is finite dimensional and closed. Theorem \ref{thm:efficiency} therefore gives a unique optimal combined correction. Its coefficient representations form a nonempty closed affine set; since $G\succ0$, minimizing $\beta^\top G\beta$ on this set gives the unique representation $\beta^\dagger$ corresponding to $v^\dagger$.

Since $M_j^\top=D_jT_j$, the matrix $M$ has full row rank. Smoothness of its null-space projector permits orthonormal bases $B$ and $\widehat B$ of $\ker M$ and $\ker\widehat M$ satisfying
$\widehat B-B=O_p(n^{-1/2})$. If $\ker M=\{0\}$, the correction claims are immediate; otherwise set \( S_B=B^\top SB\), \(G_B=B^\top GB\), and \(c_B=B^\top c\) with fitted counterparts defined analogously.
These inherit root-$n$ consistency, and $G_B\succ0$. For $x\in\ker S_B$, we have $E[(H^\top Bx)^2]=0$, hence $x^\top c_B=0$. Thus $c_B\in\mathcal R(S_B)$. Diagonalizing $G_B^{-1/2}S_BG_B^{-1/2}$ consequently shows that
\[
 \gamma_\lambda=-(S_B+\lambda G_B)^{-1}c_B,\qquad
 \gamma_\lambda-\gamma^\dagger=O(\lambda),\qquad
 \sup_{0<\lambda\leq1}\|\gamma_\lambda\|<\infty,
\]
where $\gamma^\dagger$ is the minimum-$G_B$-norm solution of $S_B\gamma=-c_B$ and $B\gamma^\dagger=\beta^\dagger$.

Write $\widehat\beta=\widehat B\widehat\gamma$.
Because $\widehat S_B\succeq0$ and $\widehat G_B\to_p G_B\succ0$,
the inverse of $\widehat S_B+\lambda_n\widehat G_B$ is
$O_p(\lambda_n^{-1})$ in operator norm.
Subtracting the penalized normal equations gives
\begin{align*}
 \widehat\gamma-\gamma_{\lambda_n}
 ={}&-(\widehat S_B+\lambda_n\widehat G_B)^{-1}\\
 &\quad\times\left[
   \widehat c_B-c_B+
   \{\widehat S_B-S_B+
      \lambda_n(\widehat G_B-G_B)\}\gamma_{\lambda_n}
 \right]
 =O_p(n^{-1/2}/\lambda_n).
\end{align*}
Combining this bound with the population ridge approximation and basis alignment yields
\[
 \widehat\beta-\beta^\dagger
 =O_p\!\left(\lambda_n+\frac{n^{-1/2}}{\lambda_n}\right)
 =o_p(1).
\]
This proves the weight rate and gives
$\widehat\beta=O_p(1)$ and $\max_j\|\widehat v_j\|_\infty=O_p(1)$.

Exact empirical balance and exact indicator representation of the true bridges imply
\begin{align*}
 E[\widehat v_j(W_j)
       \{\widehat\varphi_j(Y_j)-\varphi_{0j}(Y_j)\}]
 &=(\widehat\theta_j-\theta_{0j})^\top
       (M_j-\widehat M_j)\widehat\beta_j\\
 &=O_p(n^{-1}).
\end{align*}
The exact bias identity for $m$ therefore gives the asserted $O_p(n^{-1})$ bias.
Moreover,
\[
 \widehat v_j(W_j)\{Y_1-\widehat\varphi_j(Y_j)\}
       -v_j^\dagger(W_j)u_j
 =u_j\{\widehat v_j(W_j)-v_j^\dagger(W_j)\}
       -\widehat v_j(W_j)
          \{\widehat\varphi_j(Y_j)-\varphi_{0j}(Y_j)\}.
\]
Boundedness on the fixed finite support, the nuisance rates above,
and $\psi_b+\sum_j u_jv_j^\dagger(W_j)=\psi_{\mathrm{eff}}$
give the stated fitted-contribution rate.
Theorem \ref{thm:estimator} now yields asymptotic linearity, variance consistency, and confidence-interval coverage.

Finally, along any smooth model path $P_t$ through $P_0$ with score $s$, finite positive support gives local asymptotic normality and
contiguity under $P_{h/\sqrt n}^{\otimes n}$, also after adjoining the independent fold allocation. Contiguity transfers the asymptotic expansion, and Le Cam's third lemma shifts its limiting mean by $hE[\psi_{\mathrm{eff}}s]$.
This equals the first-order shift of the local target, so centering by $\tau(P_{h/\sqrt n})$ gives $N(0,V_{\mathrm{eff}})$ for every fixed $h$. This proves regularity and efficiency.
\end{proof}

\subsection{Riesz representation}\label{app:riesz}

The main text derivation starts from the benchmark derivative and uses the bridge restrictions to identify admissible variance reductions. The same residual form can be understood through Riesz representation and NPIV debiasing. 

At a fixed distribution $P_0$, consider the sensitivity of the plug-in expression $\E[a(Z)\phi_j(Y_j)]$ to a change $h$ in the candidate bridge. This is the linear functional
\begin{equation}\label{eq:riesz-definition}
 L_j(h)=\E[a(Z)h(Y_j)]=\ip{\ell_{0j}}h.
\end{equation}
The element $\ell_{0j}$ is its Riesz representer in $L^2(P_{Y_j})$. Existence follows from $|L_j(h)|\leq\norm a_2\norm h_2$, where $\E[a^2]=1/\{\pi(1-\pi)\}$.

\begin{proposition}[Riesz representer and adjoint solutions]\label{prop:riesz}
The representer in \eqref{eq:riesz-definition} is
\begin{equation}\label{eq:riesz}
 \ell_{0j}=\E[a(Z)\mid Y_j]=T_j^*a.
\end{equation}
Consequently $\ell_{0j}\in\Ran(T_j^*)$, and the complete set of solutions to $T_j^*q_j=\ell_{0j}$ is
\(q_j=a+v_j,\text{ where } v_j\in\calN_j.\)
\end{proposition}

Consider the NPIV-style estimating equation
\(
 a(Z)\phi_j(Y_j)+q_j(W_j)\{Y_1-\phi_j(Y_j)\}.
\)
Holding $P_0$ and $q_j$ fixed, its expectation changes under $\phi_j+\epsilon h$ at rate
\(
 \E[(a(Z)-q_j(W_j))h(Y_j)]
 =\ip{\ell_{0j}-T_j^*q_j}h.
\)
Solving the adjoint equation removes this sensitivity. Such adjoint corrections are central to NPIV influence-function and debiasing analyses \citep{severini2012efficiency, ichimura2022influence,bennett2026inference}. In the present model, substituting $q_j=a+v_j$ gives the exact identity
\begin{equation}\label{eq:npiv-cancellation}
 \begin{split}
 &a(Z)\phi_j(Y_j)+\{a(Z)+v_j(W_j)\}\{Y_1-\phi_j(Y_j)\}\\
 &\hspace{3em}=a(Z)Y_1+v_j(W_j)\{Y_1-\phi_j(Y_j)\}.
 \end{split}
\end{equation}
After centering for the known treatment law, its population representer is \eqref{eq:single-gradient}.

\begin{proof}[Proof of Proposition~\ref{prop:riesz}]
The treatment weight satisfies
\(
\E[a(Z)^2]
=\frac1{\pi(1-\pi)}.
\)
Thus $L_j(h)=\E[a(Z)h(Y_j)]$ is continuous by Cauchy--Schwarz. Iterated expectation gives
\begin{align*}
L_j(h)&=\E[\E\{a(Z)h(Y_j)\mid Y_j\}]\\
&=\E[h(Y_j)\E\{a(Z)\mid Y_j\}]\\
&=\ip{h}{\E[a(Z)\mid Y_j]}.
\end{align*}
Conditional Jensen's inequality gives $\E[\E(a\mid Y_j)^2]\leq\E[a^2]<\infty$, so the displayed representer belongs to $L^2(P_{Y_j})$. It is unique: if two elements of this space represent every pairing, their difference pairs to zero with itself.

We know $T_j^*q=\E[q(W_j)\mid Y_j]$. Because $a(Z)$ is $W_j$-measurable, the representer is $T_j^*a$. The null-space identity used in the proposition is already established in Proposition~\ref{prop:corrections}.

If $q_j$ is any solution of $T_j^*q_j=\ell_{0j}$, subtract $T_j^*a=\ell_{0j}$ and use linearity to obtain $q_j-a\in\calN_j$. Conversely, $T_j^*(a+v_j)=\ell_{0j}$ whenever $v_j\in\calN_j$. All such functions are square-integrable because the two summands are. This proves both inclusions in the affine solution set.

Finally, since $Z$ belongs to $W_j$,
\(
\E[a(Z)u_j]=\E[a(Z)\E[u_j\mid W_j]]=0,
\)
which proves the common target representation \eqref{eq:common-target}.
\end{proof}

\subsection{The source condition and the minimum-norm solution}\label{app:source}

Among all adjoint solutions, the minimum-$L^2(P_{W_j})$-norm solution is
\begin{equation}\label{eq:minnorm}
 q_j^\dagger=a-\Pi_{\calN_j}a
 =\Pi_{\overline{\Ran(T_j)}}a.
\end{equation}
The stronger source condition used in the analysis of \citet{bennett2026inference} has the form
\begin{equation}\label{eq:source-equivalence}
 \ell_{0j}\in\Ran(T_j^*T_j)
 \quad\Longleftrightarrow\quad
 q_j^\dagger=T_j\xi_j
 \text{ for some }\xi_j\in L^2(P_{Y_j}).
\end{equation}
It supports estimation of a debiasing function through a forward conditional-expectation problem. It is stronger than the adjoint-range property already supplied here by $a(Z)$, and is not implied merely by that weight being bounded.

Our variance objective instead concerns the completed influence function, including its covariance with the benchmark and all cross-bridge covariances. Minimizing $\E[q_j^2]$ is not generally equivalent to minimizing this variance. The population efficiency characterization therefore does not require \eqref{eq:source-equivalence}.

\section{Pre-treatment covariates}
\label{app:covariates}

This section extends Section~\ref{sec:measurement-efficiency} to an observed vector $X$ of pre-treatment covariates. Covariates have
two distinct roles. They can improve estimation of the benchmark contrast by explaining variation within treatment arms, without changing the measurement assumptions. They can also enter the conditioning vectors of valid bridge equations, potentially supplying additional restrictions on the observed distribution. 

\subsection{Model and covariate-adjusted benchmark}
\label{cov:setup}\label{cov:benchmark}

Observe independent and identically distributed draws $O=(X,Z,Y_1,\ldots,Y_J)$, where $X$ is pre-treatment and $P(Z=1\mid X)=\pi\in(0,1)$ is known and constant. Let the population distribution $P_X$ be unknown. With $a(Z)$ as in \eqref{eq:weight}, define
\begin{equation}\label{cov:conditional-means}
 \mu(z,x)=\E[Y_1\mid Z=z,X=x],\text{ and }
 \Delta(x)=\mu(1,x)-\mu(0,x).
\end{equation}
The marginal arm mean remains $\mu(z)=\E[\mu(z,X)]$, and the target is
\begin{equation}\label{cov:target}
 \tau_0=\E[\eta(1)-\eta(0)]=\E[\Delta(X)]=\E[a(Z)Y_1].
\end{equation}
Here $\Delta(x)$ is the conditional \emph{benchmark} contrast. Suppose every valid $W_j$ contains $(Z,X)$, and
\begin{equation}\label{cov:bridges}
 \E[u_j\mid W_j]=0,\qquad
 u_j=Y_1-\phi_{0j}(Y_j),\qquad
 \phi_{0j}\in L^2(P_{Y_j}),\quad j=2,\ldots,J.
\end{equation} Let $\calM_{0,X}$ impose the known assignment mechanism but otherwise leave the covariate and conditional outcome laws unrestricted, subject to the required integrability. Let $\calM_X\subseteq\calM_{0,X}$ additionally impose \eqref{cov:bridges}. The following proposition summarizes the usual covariate-adjusted randomized-trial results by \citet{hahn1998role,tsiatis2008covariate}.

\begin{proposition}[The covariate score space and benchmark gradient]
\label{cov:prop-benchmark}
The ambient tangent space and its orthogonal complement in $L_0^2(P_0)$ are
\begin{align}
 \mathcal H_{\pi,X}
 &=\{s\in L_0^2(P_0):\E[s\mid Z,X]=\E[s\mid X]\},
 \label{cov:ambient}\\
 \mathcal D_{\pi,X}
 &=\{(Z-\pi)b(X):b\in L^2(P_X)\}.
 \label{cov:design-normal}
\end{align}
For centered $f\in L^2(P_0)$,
\begin{equation}\label{cov:ambient-projection}
 \Pi_{\mathcal H_{\pi,X}}f
 =f-\E[f\mid Z,X]+\E[f\mid X].
\end{equation}
The benchmark canonical gradient and bound are
\begin{align}
 \psi_{\mathrm b,X}
 &=\Delta(X)-\tau_0+a(Z)\{Y_1-\mu(Z,X)\},
 \label{cov:benchmark-if}\\
 V_{\mathrm b,X}
 &=\Var\{\Delta(X)\}
   +\E\!\left[\frac{\sigma_1^2(X)}\pi+
                    \frac{\sigma_0^2(X)}{1-\pi}\right],
 \label{cov:benchmark-bound}
\end{align}
where $\sigma_z^2(x)=\Var(Y_1\mid Z=z,X=x)$.
\end{proposition}

\subsection{Bridge efficiency as a consequence of the main results}
\label{cov:tangent}\label{cov:efficiency}

Recall that \(T_jh=\E[h(Y_j)\mid W_j]\), \(T_j^*v=\E[v(W_j)\mid Y_j]\), and \(\calN_j=\Null(T_j^*)\). Their domains remain $L^2(P_{Y_j})$ and $L^2(P_{W_j})$. In particular, $\calN_j$ imposes $\E[v(W_j)\mid Y_j]=0$, not $\E[v(W_j)\mid Y_j,X]=0$. Define
\begin{equation}\label{cov:residual-space}
 \calV_X=\overline{\left\{\sum_{j=2}^J u_jv_j(W_j):v_j\in\calN_j\right\}},
 \text{ and } g_{*,X}=\Pi_{\calV_X}\psi_{\mathrm b,X}.
\end{equation}
All closures are in $L^2(P_0)$.

\begin{assumption}[Covariate version of local bridge regularity]
\label{cov:ass-tangent}
The bridges are locally identified in their stated spaces. Assumption~\ref{ass:tangent} holds with $\Hpi$ replaced by $\mathcal H_{\pi,X}$ and with the enlarged conditioning vectors in all moment and operator expressions. Thus its moment bounds, differentiable-bridge path condition, and joint closure requirement are maintained for the present model.
\end{assumption}

Only two checks are needed to apply the previous operator argument. For any square-integrable residual weights $q_j$ and any marginal covariate score $s_X$,
\begin{equation}\label{cov:transfer-checks}
 \E\!\left[\sum_j u_jq_j(W_j)\,\middle|\,Z,X\right]=0,
 \text{ and }
 \E[u_js_X(X)\mid W_j]=0.
\end{equation}
Both follow from \eqref{cov:bridges}, since $(Z,X)$ is included in every $W_j$. The first ensures that the joint score-to-moment adjoint is still the residual sum, with no additional ambient projection. The second shows that the bridge equations leave marginal-covariate scores unrestricted. Consequently the closure proof in Appendix~\ref{app:joint}, including its right-inverse construction when needed, and the complement proof in Appendix~\ref{app:complement} apply with the stated replacements. In particular,
\begin{equation}\label{cov:tangent-characterization}
 \calT_X=\{s\in\mathcal H_{\pi,X}:\E[u_js\mid W_j]
                 \in\overline{\Ran(T_j)}\ \forall j\}
       =\mathcal H_{\pi,X}\cap\calV_X^\perp.
\end{equation}
Here $\calT_X$ is the closed span of actual model scores. Also $\calV_X\subseteq\{r:\E[r\mid Z,X]=0\}$ and $L_0^2(P_X)\subseteq\calT_X$.

\begin{corollary}[Covariate-adjusted bridge efficiency]
\label{cov:thm-efficiency}
Under Assumption~\ref{cov:ass-tangent}, Theorem~\ref{thm:efficiency} gives
\begin{align}
 \psi_{\mathrm{eff},X}&=\psi_{\mathrm b,X}-g_{*,X},
 \label{cov:eif}\\
 V_{\mathrm{eff},X}
 &=V_{\mathrm b,X}-\norm{g_{*,X}}_2^2
  =\inf_{v_j\in\calN_j}\E\!\left[
       \left\{\psi_{\mathrm b,X}+\sum_j u_jv_j(W_j)\right\}^2\right].
 \label{cov:efficiency-bound}
\end{align}
The affine family of representers, excess-variance identity, and attainment characterization in that theorem hold with the same replacements.
\end{corollary}

There is a covariate-specific implication of \eqref{cov:transfer-checks}:
\begin{equation}\label{cov:variance-floor}
 \E[\psi_{\mathrm{eff},X}\mid Z,X]=\Delta(X)-\tau_0,
 \text{ and }
 V_{\mathrm{eff},X}\leq V_{\mathrm b,X}.
\end{equation}
Indeed, the benchmark residual and $g_{*,X}$ are both centered given $(Z,X)$, so they are orthogonal to $\Delta(X)-\tau_0$. The bridge restrictions can reduce the remaining within-$(Z,X)$ variation.

\subsection{The two sources of precision}
\label{cov:gains}

Define
\begin{equation}\label{cov:balance-function}
 b_0(X)=\frac{\mu(1,X)-\mu(1)}\pi
       +\frac{\mu(0,X)-\mu(0)}{1-\pi}.
\end{equation}
The benchmark gradient in the main text has the decomposition
\begin{equation}\label{cov:benchmark-decomposition}
 \psib=\psi_{\mathrm b,X}+(Z-\pi)b_0(X),
 \text{ and } (Z-\pi)b_0(X)\in\mathcal D_{\pi,X}.
\end{equation}
The two terms are orthogonal. Therefore
\begin{align}
 \Vb-V_{\mathrm b,X}
 &=\pi(1-\pi)\E[b_0(X)^2]
  =\frac{\Var\{(1-\pi)\mu(1,X)+\pi\mu(0,X)\}}{\pi(1-\pi)},
 \label{cov:ordinary-gain}\\
 \Vb-V_{\mathrm{eff},X}
 &=\underbrace{\pi(1-\pi)\E[b_0(X)^2]}_{\text{ordinary covariate adjustment}}
   +\underbrace{\norm{g_{*,X}}_2^2}_{\text{bridge information after adjustment}}.
 \label{cov:total-gain}
\end{align}
The first equality follows by substituting the two benchmark gradients; the second follows from \eqref{cov:efficiency-bound}. Ordinary adjustment is strictly beneficial precisely when the weighted prognostic combination in \eqref{cov:ordinary-gain} is nonconstant. 

For comparison with the original bridge-efficient estimator, let $W_j^0$ be the original conditioning vectors and suppose each enlarged $W_j$ contains $(W_j^0,X)$, with the same bridges. Let $\calV$ and $\Ve$ be the original space and bound, embedded in the enlarged observation space. Provided both tangent characterizations hold,
\begin{equation}\label{cov:old-bound-comparison}
 \Ve-V_{\mathrm{eff},X}
 =\pi(1-\pi)\E[b_0(X)^2]
   +\norm{(\Pi_{\calV_X}-\Pi_{\calV})\psib}_2^2\geq0.
\end{equation}

\subsection{Estimation and inference with covariates}
\label{cov:estimation}\label{cov:inference}

Replace the unadjusted benchmark contribution by
\begin{equation}\label{cov:benchmark-contribution}
 m_{\mathrm b,X}(O;\widetilde\mu)
 =\widetilde\mu(1,X)-\widetilde\mu(0,X)
   +a(Z)\{Y_1-\widetilde\mu(Z,X)\},
\end{equation}
and define
\begin{equation}\label{cov:contribution}
 m_X(O;\widetilde\mu,\phi,v)
 =m_{\mathrm b,X}(O;\widetilde\mu)
   +\sum_{j=2}^J v_j(W_j)\{Y_1-\phi_j(Y_j)\}.
\end{equation}
Use the same fixed, nonvanishing folds as in Section~\ref{subsec:estimator}; all regressions, bridge fits, correction fits, and tuning for an observation use only its training sample. For $i\in I_k$, let
$\widehat m_{X,i}=m_X(O_i;\widehat\mu^{(-k)},\widehat\phi^{(-k)},\widehat v^{(-k)})$
and set
\begin{equation}\label{cov:estimator}
 \widehat\tau_X=\frac1n\sum_i\widehat m_{X,i},\text{ and }
 \widehat V_X=\frac1n\sum_i(\widehat m_{X,i}-\widehat\tau_X)^2.
\end{equation}

Fit $\mu(z,x)$ within each treatment arm, allowing flexible functions of $X$. Fit bridges by the same conditional-moment criterion \eqref{eq:minimax}, now using the enlarged $W_j$. In the basis implementation, $p_j(Y_j)$ still spans bridge functions of $Y_j$ alone, whereas $b_j(W_j)$ can include covariates and their interactions with the other conditioning variables. The tested constraints remain
$\widehat M_j\beta_j=0$, with $\widehat M_j=\mathbb P_{\mathrm{tr}}[p_j(Y_j)b_j(W_j)^\top]$.
Use \eqref{eq:sample-optimization} and the matrix solution \eqref{eq:sample-solution} with benchmark input $a(Z)\{Y_1-\widehat\mu(Z,X)\}$. Their derivations, rank qualifications, and approximation-error requirements in Appendix~\ref{app:implementation} are unchanged. 

Let $\delta\mu_z(X)=\widetilde\mu(z,X)-\mu(z,X)$. The identity
\begin{equation}\label{cov:regression-error}
 m_{\mathrm b,X}(O;\widetilde\mu)-m_{\mathrm b,X}(O;\mu)
 =-(Z-\pi)\left\{\frac{\delta\mu_1(X)}\pi+
                     \frac{\delta\mu_0(X)}{1-\pi}\right\}
\end{equation}
shows exactly what changes in the inference argument. Its conditional expectation given $X$ is zero, and its squared norm is
\begin{equation}\label{cov:regression-error-norm}
 \pi(1-\pi)\E\!\left[
       \left\{\frac{\delta\mu_1(X)}\pi+
                    \frac{\delta\mu_0(X)}{1-\pi}\right\}^2\right].
\end{equation}
Thus $\E[m_{\mathrm b,X}(O;\widetilde\mu)\mid X]=\Delta(X)$ for every integrable candidate regression, and $L^2(P_X)$ consistency of the arm regressions suffices for this additional score error. 

Combining the benchmark cancellation with \eqref{eq:bias-exact} yields
\begin{equation}\label{cov:exact-bias}
 \E[m_X(O;\widetilde\mu,\phi,v)]-\tau_0
 =-\sum_j\ip{T_j^*v_j}{\phi_j-\phi_{0j}}.
\end{equation}
Equation \eqref{cov:exact-bias} gives the product \eqref{eq:bias-weak}--\eqref{eq:bias-balance} with the enlarged operators.

\begin{corollary}[Inference with covariates]\label{cov:thm-estimation}
Suppose Assumption~\ref{cov:ass-tangent} holds and $0<V_{\mathrm{eff},X}<\infty$. Maintain the folds and integrability requirements of Theorem~\ref{thm:estimator}, and, conditional on each training sample, assume
\begin{align}
 \norm{m_X(\,\cdot\,;\widehat\mu^{(-k)},\widehat\phi^{(-k)},
                    \widehat v^{(-k)})-(\tau_0+\psi_{\mathrm{eff},X})}_2
 &=o_p(1),\label{cov:score-consistency}\\
 \left|\sum_j\ip{T_j^*\widehat v_j^{(-k)}}{
                   \widehat\phi_j^{(-k)}-\phi_{0j}}\right|
 &=o_p(n^{-1/2}).\label{cov:bias-rate}
\end{align}
Then
\begin{equation}\label{cov:clt}
 \sqrt n(\widehat\tau_X-\tau_0)
 =\frac1{\sqrt n}\sum_i\psi_{\mathrm{eff},X}(O_i)+o_p(1)
 \ \overset d\longrightarrow\ N(0,V_{\mathrm{eff},X}),
 \qquad \widehat V_X\to_p V_{\mathrm{eff},X}.
\end{equation}
The usual interval $\widehat\tau_X\pm z_{1-\alpha/2}\sqrt{\widehat V_X/n}$ has pointwise asymptotic coverage $1-\alpha$.
\end{corollary}

The cross-fitting decomposition, conditional variance bounds, and variance-consistency argument in Appendix~\ref{app:inference} apply to $m_X$ after \eqref{cov:exact-bias} has been established.

\subsection{Further extensions}
\label{cov:alternative-calibration}

\paragraph{Covariate-dependent calibration.}
Sometimes common calibration is not credible: the relation between an auxiliary measurement and
the benchmark may vary with observed baseline characteristics. Allowing $\phi_j^X(Y_j,X)\in L^2(P_{Y_j,X})$ changes the bridge domain. Replace $u_j$ by $u_j^X=Y_1-\phi_{0j}^X(Y_j,X)$ and the operators by
\begin{equation}\label{cov:heterogeneous-operators}
 T_j^Xh=\E[h(Y_j,X)\mid W_j],\text{ and }
 (T_j^X)^*v=\E[v(W_j)\mid Y_j,X].
\end{equation}
With $(Z,X)$ still in each valid $W_j$ and regularity for these new operators, the preceding primary-model results apply directly. Use bases $p_j(Y_j,X)$ and test $\E[v_j\mid Y_j,X]=0$; the benchmark gradient does not change. At a law admitting the same reference bridge in both regular models, $\Null((T_j^X)^*)\subseteq\Null(T_j^*)$, so the nested-space result gives
$V_{\mathrm{eff},X}^{\mathrm{common}}\leq V_{\mathrm{eff},X}^{\mathrm{covariate\text{-}dependent}}$.

\paragraph{Known covariate-dependent assignment.}\label{cov:known-e}
Suppose $P(Z=1\mid X)=e(X)$ is known and
$\epsilon\leq e(X)\leq1-\epsilon$ almost surely for some
$\epsilon>0$. Define
\(
 a_e(Z,X)=\frac Z{e(X)}-\frac{1-Z}{1-e(X)}.
\)
The ambient score space is still
$\{s\in L_0^2(P_0):\E[s\mid Z,X]=\E[s\mid X]\}$,
and its orthogonal complement is
$\{(Z-e(X))b(X):b\in L^2(P_X)\}$. The overlap bound makes the
relevant weighted and unweighted coefficient norms equivalent.
For centered $f$, the ambient projection remains
$f-\E[f\mid Z,X]+\E[f\mid X]$.

With $\mu(z,x)=\E[Y_1\mid Z=z,X=x]$ and
$\Delta(x)=\mu(1,x)-\mu(0,x)$, the target is
$\tau_0=\E[\Delta(X)]=\E[a_e(Z,X)Y_1]$. Its benchmark canonical
gradient and variance are
\[
 \begin{split}
 \psi_{\mathrm b,e}
 &=\Delta(X)-\tau_0+a_e(Z,X)\{Y_1-\mu(Z,X)\},\\
 V_{\mathrm b,e}
 &=\Var\{\Delta(X)\}
   +\E\!\left[\frac{\sigma_1^2(X)}{e(X)}
              +\frac{\sigma_0^2(X)}{1-e(X)}\right].
 \end{split}
\]
When every valid bridge conditioning vector contains $(Z,X)$,
the residual sums remain centered given $(Z,X)$. The bridge
projection and cross-fitted estimation results therefore apply
with $a_e$ and the corresponding benchmark gradient. The exact
regression-error identity is
\[
 m_{\mathrm b,e}(O;\widetilde\mu)-m_{\mathrm b,e}(O;\mu)
 =-(Z-e(X))\left\{\frac{\delta\mu_1(X)}{e(X)}
                   +\frac{\delta\mu_0(X)}{1-e(X)}\right\},
\]
so its expectation is zero and $L^2$ regression consistency suffices
under the overlap bound. The comparison identities
\eqref{cov:ordinary-gain}--\eqref{cov:old-bound-comparison}
retain their stated constant-assignment scope.

\paragraph{Estimated propensity.}
The preceding results assume that the treatment-assignment
probability is known. We briefly consider the case where
$e_0(X)=P(Z=1\mid X)$ is estimated from the data.
Maintain the conditional randomization, benchmark
normalization, and bridge assumptions of the preceding
sections, and suppose that
$\epsilon\leq e_0(X)\leq 1-\epsilon$ almost surely
for some $\epsilon>0$.

For candidate propensity score $\widetilde e(X)$ and
outcome regressions $\widetilde\mu(z,X)$, define
\(
a_{\widetilde e}(Z,X)
=
\frac{Z}{\widetilde e(X)}
-
\frac{1-Z}{1-\widetilde e(X)}
\)
and
\begin{align}
m_{X,e}(O;\widetilde e,\widetilde\mu,\phi,v)
={}&
\widetilde\mu(1,X)-\widetilde\mu(0,X)
\nonumber\\
&+
a_{\widetilde e}(Z,X)
\{Y_1-\widetilde\mu(Z,X)\}
\nonumber\\
&+
\sum_{j=2}^J
v_j(W_j)\{Y_1-\phi_j(Y_j)\}.
\label{cov:estimated-e-score}
\end{align}
The estimator $\widehat\tau_e$ averages this contribution
using cross-fitted estimates of all nuisance functions,
as in Section~\ref{subsec:estimator}.

Let $\delta\mu_z(X)=\widetilde\mu(z,X)-\mu(z,X)$.
Conditional on the training sample, the exact bias
decomposition is
\begin{align}
&\E[m_{X,e}(O;\widetilde e,\widetilde\mu,\phi,v)]
-\tau_0
\nonumber\\
&=
\E\left[
\{\widetilde e(X)-e_0(X)\}
\left\{
\frac{\delta\mu_1(X)}{\widetilde e(X)}
+
\frac{\delta\mu_0(X)}{1-\widetilde e(X)}
\right\}
\right]
\nonumber\\
&\quad-
\sum_{j=2}^J
\left\langle
T_j^*v_j,\phi_j-\phi_{0j}
\right\rangle .
\label{cov:estimated-e-bias}
\end{align}
The first term is the familiar second-order remainder
from augmented inverse-probability weighting
\citep{chernozhukov2018double}. The second is the bridge remainder already analyzed in~\eqref{cov:exact-bias}.

Suppose the estimated propensity scores are bounded
away from zero and one. By Cauchy--Schwarz, the first
remainder is $o_p(n^{-1/2})$ whenever
\begin{equation}
\|\widehat e-e_0\|_2
\left\{
\|\widehat\mu(1,X)-\mu(1,X)\|_2
+
\|\widehat\mu(0,X)-\mu(0,X)\|_2
\right\}
=o_p(n^{-1/2}).
\label{cov:estimated-e-rate}
\end{equation}
If, in addition, the bridge-bias condition \eqref{cov:bias-rate} and whole-score consistency hold for the fitted contribution in~\eqref{cov:estimated-e-score}, the cross-fitting
argument of Theorem~\ref{thm:estimator} applies without modification.

Under the corresponding local tangent-space regularity conditions, the efficient influence function remains the covariate-adjusted bridge influence function evaluated at $e_0(X)$. Indeed, it is orthogonal to propensity-score
directions of the form $\{Z-e_0(X)\}b(X)$, while the bridge residuals are conditionally centered given $(Z,X)$. Consequently, the estimator remains asymptotically efficient when the preceding nuisance conditions hold and the fitted correction approaches the efficient correction.

\section{Supplementary Simulation Results}
\label{si:cbs-simulations}

\subsection{Five-category experiment}
\label{app:sim2-discrete-identification}
\label{app:sim2-measurement-identification}

This subsection gives additional information of Section~\ref{app:sim2-discrete-experiment}. Recall that both auxiliaries take values \(0,1,2,3,4\), with measurement probability matrix \(B=0.875I_5+0.025\mathbf1_5\mathbf1_5^{\mathsf T}\). Since \(\mathbb E(Y_1\mid\eta)=\eta\), the common auxiliary bridge values are
\(
\varphi=B^{-1}\eta=\frac1{70}(3,19,35,51,67)^{\mathsf T}.
\)
Writing \(p_z=(p_1(z),\ldots,p_5(z))^{\mathsf T}\) for the latent probabilities after averaging over \(X\), the operator for one auxiliary given the other and treatment is
\(
T_z=\operatorname{diag}(B^{\mathsf T}p_z)^{-1}
 B^{\mathsf T}\operatorname{diag}(p_z)B.
\)
The bridge is represented by unrestricted category-specific values; these population values happen to be equally spaced in this DGP. All latent probabilities are at least \(0.08\) at the simulated effects.

The fitted estimators use observed data only. Treatment-specific categorical cell means estimate \(\mu(z,x)=\mathbb E(Y_1\mid Z=z,X=x)\), giving
\(
\widehat m_{b,i}=\widehat\mu(1,X_i)-\widehat\mu(0,X_i)
+(4Z_i-2)\{Y_{i1}-\widehat\mu(Z_i,X_i)\}.
\)
The raw and projected correction features are
\(
H_j=b_j(W_j)\{Y_1-\widehat\phi_j(Y_j)\},
\) and 
\[
\widetilde H_j=H_j-(Z-\tfrac12)
\{\widehat{\mathbb E}(H_j\mid Z=1,X)-\widehat{\mathbb E}(H_j\mid Z=0,X)\}.
\]
Projection regressions use categorical covariate bases and zero ridge penalties. Correction coefficients are optimized jointly using all cross-auxiliary covariance terms, subject to empirical common-bridge calibration. Three-fold inner cross-validation selects the correction-penalty multiplier from \(\{0,0.01,0.1,1\}\). The complete held-out estimating equation, NSI estimator, and reported variance are
\[
\widehat m_i=\widehat m_{b,i}+\sum_{j=2}^3\widehat\beta_j^{\mathsf T}\widetilde H_{ij},
\qquad
\widehat\tau=\frac1n\sum_{i=1}^n\widehat m_i,
\]
\[
\widehat{\operatorname{Var}}(\widehat\tau)
=\frac1{n^2}\sum_{i=1}^n(\widehat m_i-\widehat\tau)^2.
\]
The benchmark sets all correction coefficients to zero and uses the same covariate regressions and outer folds.

\subsection{Finite-sample performance with nonlinear continuous measurements}
\label{app:sim2-continuous-experiment}

This section describes the finite-sample performance with nonlinear continuous measurements. It contains one benchmark, two auxiliary measurements, and a predetermined covariate. 

\paragraph{Continuous measurements.}
Independently draw \(Z_i\sim\operatorname{Bernoulli}(1/2)\), \(X_i\sim N(0,1)\), and \(U_i\sim N(0,0.0719)\), and set
\[
\eta_i=\frac12+0.09X_i+\theta(Z_i-\tfrac12)+U_i.
\]
Independent Gaussian measurement errors are also independent of \((Z_i,X_i,U_i)\), with variances \(v_1=0.042\) and \(v_2=v_3=0.024\). Set \(T_{ij}=\eta_i+\varepsilon_{ij}\) and use the fixed monotone transformation
\(
g(t)=\tfrac12+\sinh(t-\tfrac12).
\)
The observed measurements and common bridges are
\[
Y_{i1}=T_{i1},\qquad Y_{i2}=1+2g(T_{i2}),\qquad
Y_{i3}=-1+\tfrac12g(T_{i3}),
\]
\[
\begin{aligned}
\phi_2(y)&=\tfrac12+\operatorname{asinh}\bigl((y-1)/2-\tfrac12\bigr),\\
\phi_3(y)&=\tfrac12+\operatorname{asinh}\bigl(2(y+1)-\tfrac12\bigr).
\end{aligned}
\]
The benchmark satisfies \(\mathbb E(Y_1\mid\eta)=\eta\). The function \(g\) is a smooth increasing bijection, and \(\phi_j(Y_j)=T_j\). Its curvature is mild near the center: \(g(t)=t+(t-1/2)^3/6+O\{(t-1/2)^5\}\). We use \(W_2=(Z,Y_3)\) and \(W_3=(Z,Y_2)\). Conditional on treatment, each pair of calibrated auxiliaries is bivariate Gaussian with nonzero correlation.

\paragraph{Estimation and Monte Carlo design.}
Treatment-specific linear regressions estimate \(\mu(z,x)=\mathbb E(Y_1\mid Z=z,X=x)\), giving the benchmark contribution
\[
\widehat m_{b,i}=\widehat\mu(1,X_i)-\widehat\mu(0,X_i)
+(4Z_i-2)\{Y_{i1}-\widehat\mu(Z_i,X_i)\}.
\]
For a raw correction feature \(H=b(W_j)\{Y_1-\widehat\phi_j(Y_j)\}\), the projected feature is
\[
\widetilde H=H-(Z-\tfrac12)
\{\widehat{\mathbb E}(H\mid Z=1,X)-\widehat{\mathbb E}(H\mid Z=0,X)\}.
\]
Projection regressions use linear covariate bases and zero ridge penalties, with coefficients estimated from the relevant training sample. Joint correction optimization retains all cross-auxiliary covariance terms and imposes common calibration.

We use three-fold inner cross-validation to select tuning parameters. Each auxiliary uses a matched bridge/instrument candidate pair from degree-one polynomials, degree-two polynomials, or raw natural splines with four degrees of freedom. Bridge and correction multipliers are selected from \(0,0.01,0.1,1\), also allowing zero correction. Bridge selection uses the U-statistic conditional-moment criterion, a bounded natural-spline validation critic with eight degrees of freedom, and validation ridge \(10^{-8}\). Correction selection minimizes the held-out score second moment.

The complete held-out contribution, NSI estimator, and reported variance are
\[
\widehat m_i=\widehat m_{b,i}+\sum_{j=2}^3\widehat\beta_j^{\mathsf T}\widetilde H_{ij},
\qquad
\widehat\tau=\frac1n\sum_{i=1}^n\widehat m_i,
\]
\[
\widehat{\operatorname{Var}}(\widehat\tau)
=\frac1{n^2}\sum_{i=1}^n(\widehat m_i-\widehat\tau)^2.
\]
The benchmark estimator sets all correction coefficients to zero and uses the same covariate regressions and outer folds.

For each \(n\in\{400,800,1200,1600\}\), we use 1,000 replications at \(\theta=0\) and 500 at \(\theta=0.06\). This continuous experiment contains 6,000 datasets and 6,000 NSI fits. We report bias, empirical standard deviation, root mean squared reported standard error, mean estimated variance divided by empirical variance, coverage of nominal 95-percent normal intervals, and two-sided rejection of zero. 

\begin{table}[htbp]
\centering\small
\setlength{\tabcolsep}{4pt}
\caption{Finite-sample performance of NSI with nonlinear continuous measurements}
\label{tab:sim2-continuous-bridge}
\begin{tabular}{rrrrrrrr}
\hline
$n$ & $\theta$ & Bias & Emp. SD & RMS SE & Var. ratio & Cov. (\%) & Rej. (\%) \\
\hline
400 & 0.00 & 0.00065 & 0.0294 & 0.0287 & 0.956 & 94.40 & 5.60 \\
800 & 0.00 & -0.00066 & 0.0209 & 0.0202 & 0.931 & 93.90 & 6.10 \\
1200 & 0.00 & -0.00006 & 0.0164 & 0.0165 & 1.007 & 95.40 & 4.60 \\
1600 & 0.00 & 0.00027 & 0.0140 & 0.0143 & 1.033 & 94.80 & 5.20 \\
400 & 0.06 & -0.00015 & 0.0282 & 0.0290 & 1.053 & 95.80 & 56.60 \\
800 & 0.06 & 0.00181 & 0.0211 & 0.0203 & 0.923 & 94.00 & 85.40 \\
1200 & 0.06 & -0.00064 & 0.0162 & 0.0165 & 1.040 & 95.60 & 95.60 \\
1600 & 0.06 & -0.00026 & 0.0143 & 0.0143 & 0.996 & 94.40 & 98.40 \\
\hline
\end{tabular}
\par\medskip
\begin{minipage}{.98\linewidth}\footnotesize
\textit{Notes:} Emp. SD is the standard deviation across replications; RMS SE is the square root of the mean squared reported standard error. Var. ratio is $(\mathrm{RMS\ SE}/\mathrm{Emp.\ SD})^2$. Coverage concerns nominal 95-percent intervals for the true effect. Rejection measures size at $\theta=0$ (target five percent) and power at $\theta=0.06$.
\end{minipage}
\end{table}


\begin{figure}[tbp]
\centering
\includegraphics[width=\linewidth]{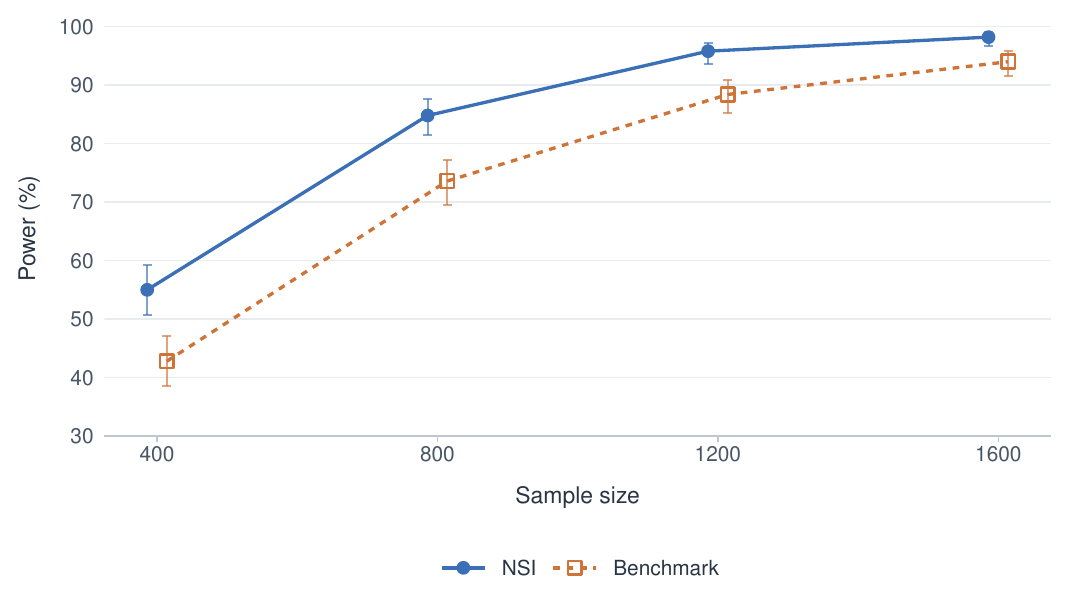}
\caption{Power with nonlinear continuous measurements at $\theta=0.06$, using covariate adjustment and common bridges. Points report rejection by two-sided five-percent tests over 500 replications per sample size. Bars are pointwise 95-percent Wilson Monte Carlo intervals. Blue solid lines and filled symbols denote NSI; orange dashed lines and open symbols denote the benchmark.}
\label{fig:sim2-power}
\end{figure}

\begin{figure}[tbp]
\centering
\includegraphics[width=\linewidth]{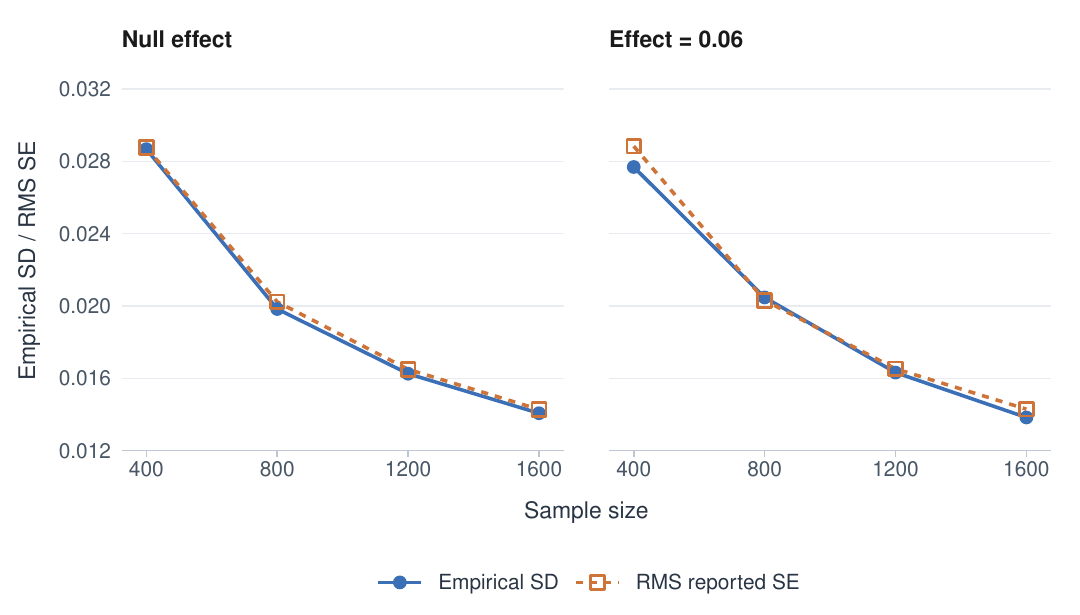}
\caption{NSI variance calibration with nonlinear continuous measurements, using covariate adjustment and common bridges. Blue solid lines and filled symbols show empirical standard deviations; orange dashed lines and open symbols show root mean squared reported standard errors. The null uses 1,000 replications per sample size and the alternative uses 500. Panels share a vertical scale.}
\label{fig:sim2-variance-calibration}
\end{figure}

\section{Supplementary Information for the Hajj Lottery Application}
\label{si:app}

This supplement reports the measurements, samples, estimates, and diagnostic checks for the six tolerance outcomes and beliefs in localized Muslim practices. All effects concern winning the visa lottery.

\subsection{Measurements and benchmark coding}

Table~\ref{tab:hajj-A1-measurement-inventory} lists the forty items used for the seven outcomes, with question labels, response coding, observed support, and response counts. Higher values follow the directions specified in the released data. The benchmark for each outcome is the item with the most distinct observed response values, with ties resolved by response count and then listed order. 

\input{hajj_application/tables/table_A1_measurement_inventory.tex}

\subsection{Analysis samples and missing responses}

Within each outcome, we retain applicants with complete data for all items in that battery and the regression covariates. Table~\ref{tab:hajj-A2-samples-missingness} reports sample sizes, application-party counts, and completion rates by lottery status, sex, and literacy. 

\input{hajj_application/tables/table_A2_samples_missingness.tex}

\subsection{NSI estimation and sensitivity analyses}

For outcome $r$, let $Y_{i1r}$ be the benchmark and $Y_{ijr}$, $j=2,\ldots,J_r$, the auxiliary measurements. The covariates $X_i$ consist of lottery-cell indicators. Lottery cells are the combinations of departure place, accommodation category, and party-size category in the released data. We use these indicators to estimate the benchmark's conditional mean separately in each lottery arm and to estimate the lottery-success probability. All outcomes and localized-belief batteries use this same adjustment specification. 

We use five-fold cross-fitting. The arm-specific benchmark means use ridge regression, and the propensity uses penalized logistic regression. Three inner party-level folds select penalty multipliers from $0.01,0.1,1,10$, divided by the square root of the relevant training sample size, using mean squared error for the benchmark means and binomial deviance for the propensity. 

Bridges use full categorical response indicators without a bridge ridge penalty. Three inner party-level folds select the joint correction penalty. Following the covariate procedure, correction fitting and validation use the weighted benchmark residual plus the joint bridge correction, excluding the estimated conditional-mean contrast. The criterion is the held-out mean of $[\widehat m_{ir}-\{\widehat\mu_{1r}(X_i)-\widehat\mu_{0r}(X_i)\}]^2$, with penalty multipliers $0,0.01,0.1,1,3,10,30,100,300,1000$, divided by the square root of the training sample size.

Analytic standard errors use application-party sums of the centered cross-fitted score
\(
\widehat\psi_{ir}=\widehat m_{ir}-\widehat\tau_r.
\)
For $G_r$ parties, the variance is $G_r/(G_r-1)$ times the sum of squared party influence totals divided by $N_r^2$. Intervals and two-sided tests use $t_{G_r-1}$.

Table~\ref{tab:hajj-covariate-tolerance} reports the NSI estimates with 95\% intervals; the main figure displays the same estimates and intervals.

\input{hajj_application/tables/table_tolerance_results.tex}




%% file: hajj_application/tables/table_A1_measurement_inventory.tex
\begin{landscape}
\begingroup
\scriptsize
\renewcommand{\baselinestretch}{1}\selectfont
\setlength{\tabcolsep}{2pt}
\begin{longtable}{>{\raggedright\arraybackslash}p{1.8cm}>{\raggedright\arraybackslash}p{2.6cm}>{\raggedright\arraybackslash}p{2.0cm}>{\raggedright\arraybackslash}p{5.8cm}>{\raggedright\arraybackslash}p{4.8cm}>{\raggedright\arraybackslash}p{1.0cm}>{\raggedright\arraybackslash}p{1.0cm}>{\raggedright\arraybackslash}p{1.0cm}}
\caption{Hajj latent-outcome measurement inventory}\label{tab:hajj-A1-measurement-inventory}\\
\toprule
Domain & Construct & Variable & Item & Coding / higher-value direction & Supp. & Item $N$ & Bench. \\ 
\midrule
\endfirsthead
\multicolumn{8}{c}{\tablename\ \thetable{} -- continued} \\ 
\toprule
Domain & Construct & Variable & Item & Coding / higher-value direction & Supp. & Item $N$ & Bench. \\ 
\midrule
\endhead
\midrule \multicolumn{8}{r}{Continued on next page} \\ \endfoot
\bottomrule \endlastfoot
Tolerance & Views of other countries & \nolinkurl{x_s15cq5} & Is your general view of Saudi people: & (Scale -2 to 2) 2=Very Positive, -2=Very Negative; family direction: more favorable views of people from other countries & 5 & 1,593 & Yes \\ 
Tolerance & Views of other countries & \nolinkurl{x_s15cq6} & Is your general view of Indonesian people: & (Scale -2 to 2) 2=Very Positive, -2=Very Negative; family direction: more favorable views of people from other countries & 5 & 1,583 &  \\ 
Tolerance & Views of other countries & \nolinkurl{x_s15cq7} & Is your general view of Turkish people: & (Scale -2 to 2) 2=Very Positive, -2=Very Negative; family direction: more favorable views of people from other countries & 5 & 1,583 &  \\ 
Tolerance & Views of other countries & \nolinkurl{x_s15cq8} & Is your general view of African people: & (Scale -2 to 2) 2=Very Positive, -2=Very Negative; family direction: more favorable views of people from other countries & 5 & 1,583 &  \\ 
Tolerance & Views of other countries & \nolinkurl{x_s15cq10} & Is your general view of European people: & (Scale -2 to 2) 2=Very Positive, -2=Very Negative; family direction: more favorable views of people from other countries & 5 & 1,583 &  \\ 
Tolerance & Views of other countries & \nolinkurl{x_s15cq9} & Is your general view of Chinese people: & (Scale -2 to 2) 2=Very Positive, -2=Very Negative; family direction: more favorable views of people from other countries & 5 & 1,585 &  \\ 
Tolerance & Views of other groups & \nolinkurl{x_s15b_b6} & In your opinion, overall how are people of a different sect compared to your people? & Binary: 0=Better or Worse, 1=Same; family direction: greater perceived equality across sectarian, religious, and ethnic groups & 2 & 1,604 & Yes \\ 
Tolerance & Views of other groups & \nolinkurl{x_s15b_c6} & In your opinion, overall how are people of a different religion compared to your people? & Binary: 0=Better or Worse, 1=Same; family direction: greater perceived equality across sectarian, religious, and ethnic groups & 2 & 1,604 &  \\ 
Tolerance & Views of other groups & \nolinkurl{x_s15b_a6} & In your opinion, overall how are people of a different ethnicity compared to your people? & Binary: 0=Better or Worse, 1=Same; family direction: greater perceived equality across sectarian, religious, and ethnic groups & 2 & 1,604 &  \\ 
Tolerance & Harmony & \nolinkurl{x_s14aq7} & Do you ever pray in the mosque of a different maslak than your own? & Binary: 1=Frequently, 0=Less Often/Never; family direction: greater perceived harmony and cross-sect practice & 2 & 1,463 & Yes \\ 
Tolerance & Harmony & \nolinkurl{x_s15cq2} & Do you believe that people of different sects of Islam can live in unity \& agreement (harmony) in a given society by  making agreements over their differences? & Binary: 1=Yes, 0=No; family direction: greater perceived harmony and cross-sect practice & 2 & 1,346 &  \\ 
Tolerance & Harmony & \nolinkurl{x_s15cq3} & Do you believe that people of different religions can live in unity \& agreement (harmony) in a given society by making agreements over their differences? & Binary: 1=Yes, 0=No; family direction: greater perceived harmony and cross-sect practice & 2 & 1,270 &  \\ 
Tolerance & Harmony & \nolinkurl{x_s15cq1} & Do you believe that people of different ethnicities (e.g. Pathan, Punjabi, Sindhi, Baluchi) can live in unity \& agreement (harmony) in a given society by  making agreements over their differences? & Binary: 1=Yes, 0=No; family direction: greater perceived harmony and cross-sect practice & 2 & 1,355 &  \\ 
Tolerance & Peaceful inclination & \nolinkurl{x_s10cq3} & What kind of support do you think Pakistan should give to those fighting against the Indian government in Kashmir? & 2= No or only political support, 1= Only military support, 0=Military and political support; family direction: more peaceful inclination & 3 & 1,236 & Yes \\ 
Tolerance & Peaceful inclination & \nolinkurl{x_s10bq4} & Do you believe the goals for which Ossama is fighting are correct? & Binary: 1=Not Correct at All/Slightly Incorrect, 0=Correct/Absolutely Correct; family direction: more peaceful inclination & 2 & 761 &  \\ 
Tolerance & Peaceful inclination & \nolinkurl{x_s10bq5} & Do you believe the methods Ossama uses in fighting are correct? & Binary: 1= Absolutely Never/Almost Never Correct, 0=To Small Extent/Some Extent/Strongly Correct; family direction: more peaceful inclination & 2 & 761 &  \\ 
Tolerance & Peaceful inclination & \nolinkurl{x_s10cq5} & How important do you believe peace with India is for Pakistan’s future? & Binary: 1=Important, 0= Not Important; family direction: more peaceful inclination & 2 & 1,155 &  \\ 
Tolerance & Peaceful inclination & \nolinkurl{x_s10cq4} & How would you view the establishment of the Line of Control/current borders as the final international boundary between India and Pakistan if it led to peace between the two countries? & Binary: 1= Very Favorably, 0=Favorably to Very Unfavorably; family direction: more peaceful inclination & 2 & 900 &  \\ 
Tolerance & Peaceful inclination & \nolinkurl{x_s10hq6} & Please tell me what you think about the correctness of the following: Family members physically punishing someone who has dishonored the family & Binary: 0=Correct, 1=Never Correct; family direction: more peaceful inclination & 2 & 1,459 &  \\ 
Tolerance & Peaceful inclination & \nolinkurl{x_s10hq7} & Please tell me what you think about the correctness of the following: Suicide attacks of military targets & Binary: 0=Always/Sometimes Correct, 1=Rarely/Never Correct; family direction: more peaceful inclination & 2 & 1,605 &  \\ 
Tolerance & Peaceful inclination & \nolinkurl{x_s10hq8} & Please tell me what you think about the correctness of the following: Attacks on civilian targets during war & Binary: 0=Always/Sometimes Correct, 1=Rarely/Never Correct; family direction: more peaceful inclination & 2 & 1,605 &  \\ 
Tolerance & Political Islam & \nolinkurl{x_s7q12a} & The government should force people to conform to all Islamic injunctions like praying, fasting, etc and punish those who violate any of these injunctions & Binary: 1=Agree Strongly/Agree, 0=Neutral/Disagree/Strongly Disagree; family direction: more support for a political role for Islam & 2 & 1,505 & Yes \\ 
Tolerance & Political Islam & \nolinkurl{x_s7q12f} & Religious leaders have a right to dispense justice on their own, even if it goes against the government & Binary: 1=Agree Strongly/Agree, 0=Neutral/Disagree/Strongly Disagree; family direction: more support for a political role for Islam & 2 & 1,315 &  \\ 
Tolerance & Political Islam & \nolinkurl{x_s7q12d} & Religious leaders should not directly influence government decisions & Binary: 1=Disagree Strongly/Disagree, 0=Neutral/Agree/Agree Strongly; family direction: more support for a political role for Islam & 2 & 1,357 &  \\ 
Tolerance & Political Islam & \nolinkurl{x_s7q12c} & It would be better for Pakistan if all politicians and government officials had strong religious beliefs & Binary: 1=Agree Strongly/Agree, 0=Neutral/Disagree/Strongly Disagree; family direction: more support for a political role for Islam & 2 & 1,393 &  \\ 
Tolerance & Political Islam & \nolinkurl{x_s7q11c} & In voting for a candidate for public office, how important would you consider the candidate’s religious beliefs to be? & Binary: 0=All other responses, 1=Very Important; family direction: more support for a political role for Islam & 2 & 1,498 &  \\ 
Tolerance & Views of the West & \nolinkurl{x_s10aq2} & Do you believe it is good for Pakistanis to adopt technologies developed in the West? & Binary: 0=Yes, 1=No; family direction: less favorable or more suspicious views of the West & 2 & 1,605 & Yes \\ 
Tolerance & Views of the West & \nolinkurl{x_s10aq3} & Do you believe it is good for Pakistanis to adopt Western social values? & Binary: 0=Yes, 1=No; family direction: less favorable or more suspicious views of the West & 2 & 1,572 &  \\ 
Tolerance & Views of the West & \nolinkurl{x_s10aq4} & How likely to do you think it is that events like the Sept 11th attack in the US and the 2005 London bombings were not done by Muslims but had a Western/Jewish hand behind them? & Binary: 0=Unlikely, 1=Very likely/likely; family direction: less favorable or more suspicious views of the West & 2 & 1,342 &  \\ 
Tolerance & Views of the West & \nolinkurl{x_s10aq5} & To what extent do you think the West takes into account the interests of countries like Pakistan? & Binary: 0=A Great deal/Fair amount, 1=Not too much/Not at all; family direction: less favorable or more suspicious views of the West & 2 & 1,109 &  \\ 
Religion & Belief in localized Muslim practices & \nolinkurl{x_s14dq6} & Is it required that a cap be worn when praying Salat? & Binary: 1=Yes, 0=No/Depends on school of Islamic thought; family direction: greater belief in localized Muslim practices & 2 & 1,605 & Yes \\ 
Religion & Belief in localized Muslim practices & \nolinkurl{x_s14aq15d} & Which of the following best describes your views regarding pirs? & Binary: 1=Genuine pirs do exist, 0=There is no concept of pirs in Islam or there are no genuine pirs today; family direction: greater belief in localized Muslim practices & 2 & 1,433 &  \\ 
Religion & Belief in localized Muslim practices & \nolinkurl{x_s14aq16b} & Which statement best describes your opinion and behavior with respect to: visiting genuine pirs to seek their advice & Binary: 1=It is useful, 0=it is not useful; family direction: greater belief in localized Muslim practices & 2 & 1,598 &  \\ 
Religion & Belief in localized Muslim practices & \nolinkurl{x_s14aq16a} & Which statement best describes your opinion and behavior with respect to: visiting mazars & Binary: 1=It is useful, 0=it is not useful; family direction: greater belief in localized Muslim practices & 2 & 1,603 &  \\ 
Religion & Belief in localized Muslim practices & \nolinkurl{x_s14aq16d} & Which statement best describes your opinion and behavior with respect to: doing dam or getting a ta’wiz & Binary: 1=It is useful, 0=it is not useful; family direction: greater belief in localized Muslim practices & 2 & 1,493 &  \\ 
Religion & Belief in localized Muslim practices & \nolinkurl{x_s14aq16e} & Which statement best describes your opinion and behavior with respect to: doing ”Chaliswaan” upon someone’s death & Binary: 1=It is important and should be done, 0= it is not important or should not be done; family direction: greater belief in localized Muslim practices & 2 & 1,542 &  \\ 
Religion & Belief in localized Muslim practices & \nolinkurl{x_s14aq16f} & Which statement best describes your opinion and behavior with respect to: participating in maulud mehfils & Binary: 1=It is useful, 0=it is not useful; family direction: greater belief in localized Muslim practices & 2 & 1,521 &  \\ 
Religion & Belief in localized Muslim practices & \nolinkurl{x_s14cq6} & Is dowry (jehej) mandatory according to Islam? & Binary: 1=Yes, 0=No; family direction: greater belief in localized Muslim practices & 2 & 1,562 &  \\ 
Religion & Belief in localized Muslim practices & \nolinkurl{x_s10dq7} & What priority should widowed women be given in getting remarried compared to unmarried women? & Binary: 1= Less or more priority, 0= Same priority; family direction: greater belief in localized Muslim practices & 2 & 1,424 &  \\ 
Religion & Belief in localized Muslim practices & \nolinkurl{x_s14dq7} & On the Day of Judgment, people will be called to bear witness to your actions in life. Do you believe you will be held entirely responsible for your actions in life, or do you believe that over and above being a witness it is possible for someone to intercede with God on your behalf? & Binary: 1=There could be an intercession, 0=I will be held entirely responsible; family direction: greater belief in localized Muslim practices & 2 & 1,604 &  \\ 
\end{longtable}
\noindent\footnotesize\textit{Notes:} Item coding follows the public Hajj codebook and oriented variables in the released data. The benchmark rule is fixed before estimation. The inventory covers the six tolerance outcomes and beliefs in localized Muslim practices. Support is the number of distinct observed response values; Bench. identifies the benchmark.
\endgroup
\end{landscape}

%% file: hajj_application/tables/table_A2_samples_missingness.tex
\begin{table}[!htbp]
\centering\small
\begin{center}\textbf{Survey completion in the original sampling frame}\end{center}
\begin{tabular}{llrrr}
\toprule
Sample & Lottery status & Selected & Completed & Completion rate \\ 
\midrule
Full selected sample & Total & 2,537 & 1,605 & 63.3\% \\ 
Full selected sample & Winner & 1,286 & 855 & 66.5\% \\ 
Full selected sample & Loser & 1,251 & 750 & 60.0\% \\ 
Restricted selected sample & Total & 1,995 & 1,295 & 64.9\% \\ 
Restricted selected sample & Winner & 1,032 & 679 & 65.8\% \\ 
Restricted selected sample & Loser & 963 & 616 & 64.0\% \\ 
\bottomrule
\end{tabular}
\begin{minipage}{0.9\linewidth}\footnotesize\textit{Notes:} Values reproduce Clingingsmith, Khwaja, and Kremer (2009), Table III. The public analysis file contains completed interviews only.\end{minipage}
\end{table}
\begingroup\scriptsize
\renewcommand{\baselinestretch}{1}\selectfont
\setlength{\tabcolsep}{2pt}
\begin{longtable}{>{\raggedright\arraybackslash}p{1.8cm}>{\raggedright\arraybackslash}p{2.4cm}rrcccccc>{\raggedright\arraybackslash}p{1.35cm}}
\caption{Familywise samples and patterned item response}\label{tab:hajj-A2-samples-missingness}\\
\toprule
Domain & Construct & $N$ & Parties & Loss & Win & Male & Female & Illit. & Lit. & Item $N$ range \\ 
\midrule\endfirsthead
\multicolumn{11}{c}{\tablename\ \thetable{} -- continued} \\ 
\toprule
Domain & Construct & $N$ & Parties & Loss & Win & Male & Female & Illit. & Lit. & Item $N$ range \\ 
\midrule\endhead
\bottomrule\endlastfoot
Tolerance & Views of other countries & 1,581 & 964 & 0.985 & 0.985 & 0.991 & 0.978 & 0.975 & 0.992 & 1,583--1,593 \\ 
Tolerance & Views of other groups & 1,604 & 969 & 0.999 & 1.000 & 1.000 & 0.999 & 0.998 & 1.000 & 1,604--1,604 \\ 
Tolerance & Harmony & 1,136 & 810 & 0.708 & 0.708 & 0.811 & 0.601 & 0.605 & 0.777 & 1,270--1,463 \\ 
Tolerance & Peaceful inclination & 487 & 402 & 0.301 & 0.305 & 0.433 & 0.169 & 0.115 & 0.430 & 761--1,605 \\ 
Tolerance & Political Islam & 1,245 & 852 & 0.764 & 0.786 & 0.864 & 0.684 & 0.654 & 0.857 & 1,315--1,505 \\ 
Tolerance & Views of the West & 1,012 & 744 & 0.621 & 0.639 & 0.771 & 0.484 & 0.419 & 0.773 & 1,109--1,605 \\ 
Religion & Belief in localized Muslim practices & 1,165 & 804 & 0.720 & 0.731 & 0.746 & 0.705 & 0.702 & 0.742 & 1,424--1,605 \\ 
\end{longtable}
\noindent\footnotesize\textit{Notes:} Loss and Win are family-completion rates among interviewed lottery losers and winners. Male/Female and Illit./Lit. show the corresponding demographic completion rates.
\endgroup

%% file: hajj_application/tables/table_tolerance_results.tex
\begin{table}[!htbp]\renewcommand{\baselinestretch}{1}\centering\small\selectfont
\caption{Lottery-cell-adjusted estimates for six tolerance outcomes}
\label{tab:hajj-covariate-tolerance}
\begin{threeparttable}\setlength{\tabcolsep}{3pt}
\begin{tabularx}{\linewidth}{>{\raggedright\arraybackslash}p{3.4cm}l>{\raggedright\arraybackslash}Xrr}
\toprule
Outcome & Instruments & Estimate [95\% CI] & SE & $p$ \\
\midrule
Views of other countries & Basic & $-0.0087$ [$-0.0421,0.0247$] & 0.0170 & 0.609 \\
Views of other countries & Extended & $0.0338$ [$0.0058,0.0619$] & 0.0143 & 0.018 \\
\addlinespace
Views of other groups & Basic & $0.0411$ [$0.0098,0.0724$] & 0.0159 & 0.010 \\
Views of other groups & Extended & $0.0721$ [$0.0260,0.1182$] & 0.0235 & 0.002 \\
\addlinespace
Harmony & Basic & $-0.0002$ [$-0.0019,0.0015$] & 0.0009 & 0.779 \\
Harmony & Extended & $0.0018$ [$0.0003,0.0033$] & 0.0008 & 0.022 \\
\addlinespace
Peaceful inclination & Basic & $-0.0047$ [$-0.0377,0.0283$] & 0.0168 & 0.778 \\
Peaceful inclination & Extended & $0.0028$ [$-0.0325,0.0381$] & 0.0179 & 0.875 \\
\addlinespace
Political Islam & Basic & $-0.0028$ [$-0.0064,0.0007$] & 0.0018 & 0.113 \\
Political Islam & Extended & $-0.0016$ [$-0.0089,0.0057$] & 0.0037 & 0.666 \\
\addlinespace
Views of the West & Basic & $0.0009$ [$-0.0079,0.0098$] & 0.0045 & 0.834 \\
Views of the West & Extended & $0.0035$ [$-0.0074,0.0144$] & 0.0056 & 0.530 \\
\bottomrule\end{tabularx}
\end{threeparttable}\end{table}

%% file: literature.bib
@article{chernozhukov2018double,
  title={Double/debiased machine learning for treatment and structural parameters},
  author={Chernozhukov, Victor and Chetverikov, Denis and Demirer, Mert and Duflo, Esther and Hansen, Christian and Newey, Whitney and Robins, James},
  journal={The Econometrics Journal},
  year={2018},
  volume={21},
  issue={1},
  publisher={Oxford University Press Oxford, UK}
}

@book{pearl2009causality,
  title={Causality},
  author={Pearl, Judea},
  year={2009},
  publisher={Cambridge university press}
}

@book{hernan2020causal,
  title={Causal inference: What If},
  author={Hern{\'a}n, Miguel A and Robins, James M},
  year={2020},
  publisher={CRC Boca Raton, FL}
}

@article{egami2022make,
  title={How to make causal inferences using texts},
  author={Egami, Naoki and Fong, Christian J and Grimmer, Justin and Roberts, Margaret E and Stewart, Brandon M},
  journal={Science Advances},
  volume={8},
  number={42},
  pages={eabg2652},
  year={2022},
  publisher={American Association for the Advancement of Science}
}

@article{vishnubhatla2026proxy,
  title={Proxy-Guided Measurement Calibration},
  author={Vishnubhatla, Saketh and Wan, Shu and Harrison, Andre and Raglin, Adrienne and Liu, Huan},
  journal={arXiv preprint arXiv:2603.09288},
  year={2026}
}

@article{landy2025causal,
  title={Causal Inference for Latent Outcomes Learned with Factor Models},
  author={Landy, Jenna M and Zorzetto, Dafne and De Vito, Roberta and Parmigiani, Giovanni},
  journal={arXiv preprint arXiv:2506.20549},
  year={2025}
}

@article{tchetgen2024introduction,
  title={An introduction to proximal causal inference},
  author={Tchetgen Tchetgen, Eric J and Ying, Andrew and Cui, Yifan and Shi, Xu and Miao, Wang},
  journal={Statistical Science},
  volume={39},
  number={3},
  pages={375--390},
  year={2024},
  publisher={Institute of Mathematical Statistics}
}

@article{severini2012efficiency,
  title={Efficiency bounds for estimating linear functionals of nonparametric regression models with endogenous regressors},
  author={Severini, Thomas A and Tripathi, Gautam},
  journal={Journal of Econometrics},
  volume={170},
  number={2},
  pages={491--498},
  year={2012},
  publisher={Elsevier}
}

@article{athey2026surrogate,
  title={The surrogate index: Combining short-term proxies to estimate long-term treatment effects more rapidly and precisely},
  author={Athey, Susan and Chetty, Raj and Imbens, Guido W and Kang, Hyunseung},
  journal={Review of Economic Studies},
  volume={93},
  number={4},
  pages={2284--2312},
  year={2026},
  publisher={Oxford University Press UK}
}

@article{chernozhukov2022locally,
  title={Locally robust semiparametric estimation},
  author={Chernozhukov, Victor and Escanciano, Juan Carlos and Ichimura, Hidehiko and Newey, Whitney K and Robins, James M},
  journal={Econometrica},
  volume={90},
  number={4},
  pages={1501--1535},
  year={2022},
  publisher={Wiley Online Library}
}

@article{ackerberg2012practical,
  title={A practical asymptotic variance estimator for two-step semiparametric estimators},
  author={Ackerberg, Daniel and Chen, Xiaohong and Hahn, Jinyong},
  journal={Review of Economics and Statistics},
  volume={94},
  number={2},
  pages={481--498},
  year={2012},
  publisher={The MIT Press}
}

@article{chen2025local,
  title={On Local Overidentification and Efficiency Gains in Modern Causal Inference and Data Combination},
  author={Chen, Xiaohong and Xie, Haitian},
  journal={arXiv preprint arXiv:2510.16683},
  year={2025}
}

@article{lubotsky2006interpretation,
  title={Interpretation of regressions with multiple proxies},
  author={Lubotsky, Darren and Wittenberg, Martin},
  journal={The Review of Economics and Statistics},
  volume={88},
  number={3},
  pages={549--562},
  year={2006},
  publisher={The MIT Press}
}

@article{freyberger2026normalizations,
  title={Normalizations and misspecification in skill formation models},
  author={Freyberger, Joachim},
  journal={Review of Economic Studies},
  volume={93},
  number={4},
  pages={2574--2604},
  year={2026},
  publisher={Oxford University Press UK}
}

@article{del2022identification,
  title={Identification of dynamic latent factor models of skill formation with translog production},
  author={Del Bono, Emilia and Kinsler, Josh and Pavan, Ronni},
  journal={Journal of Applied Econometrics},
  volume={37},
  number={6},
  pages={1256--1265},
  year={2022},
  publisher={Wiley Online Library}
}

@article{cunha2010estimating,
  title={Estimating the technology of cognitive and noncognitive skill formation},
  author={Cunha, Flavio and Heckman, James J and Schennach, Susanne M},
  journal={Econometrica},
  volume={78},
  number={3},
  pages={883--931},
  year={2010},
  publisher={Wiley Online Library}
}

@article{bond2013evolution,
  title={The evolution of the Black-White test score gap in Grades K--3: The fragility of results},
  author={Bond, Timothy N and Lang, Kevin},
  journal={Review of Economics and Statistics},
  volume={95},
  number={5},
  pages={1468--1479},
  year={2013},
  publisher={The MIT Press}
}

@article{guo2026proximal,
  title={Proximal causal inference for hidden outcomes},
  author={Guo, Helen and Ghassami, AmirEmad and Shpitser, Ilya and Ogburn, Elizabeth L},
  journal={arXiv preprint arXiv:2605.09849},
  year={2026}
}

@article{tsiatis2008covariate,
  title={Covariate adjustment for two-sample treatment comparisons in randomized clinical trials: a principled yet flexible approach},
  author={Tsiatis, Anastasios A and Davidian, Marie and Zhang, Min and Lu, Xiaomin},
  journal={Statistics in medicine},
  volume={27},
  number={23},
  pages={4658--4677},
  year={2008},
  publisher={Wiley Online Library}
}

@article{hahn1998role,
  title={On the role of the propensity score in efficient semiparametric estimation of average treatment effects},
  author={Hahn, Jinyong},
  journal={Econometrica},
  pages={315--331},
  volume={66},
  number={2},
  year={1998},
  publisher={JSTOR}
}

@article{ichimura2022influence,
  title={The influence function of semiparametric estimators},
  author={Ichimura, Hidehiko and Newey, Whitney K},
  journal={Quantitative Economics},
  volume={13},
  number={1},
  pages={29--61},
  year={2022},
  publisher={Wiley Online Library}
}

@article{chen2018overidentification,
  title={Overidentification in regular models},
  author={Chen, Xiaohong and Santos, Andres},
  journal={Econometrica},
  volume={86},
  number={5},
  pages={1771--1817},
  year={2018},
  publisher={Wiley Online Library}
}

@article{dikkala2020minimax,
  title={Minimax estimation of conditional moment models},
  author={Dikkala, Nishanth and Lewis, Greg and Mackey, Lester and Syrgkanis, Vasilis},
  journal={Advances in Neural Information Processing Systems},
  volume={33},
  pages={12248--12262},
  year={2020}
}

@article{bennett2026inference,
  title={Inference on strongly identified functionals of weakly identified functions},
  author={Bennett, Andrew and Kallus, Nathan and Mao, Xiaojie and Newey, Whitney K and Syrgkanis, Vasilis and Uehara, Masatoshi},
  journal={Journal of the Royal Statistical Society Series B: Statistical Methodology},
  volume={88},
  number={3},
  pages={998--1028},
  year={2026},
  publisher={Oxford University Press UK}
}

@article{clingingsmith2009estimating,
  title={Estimating the impact of the Hajj: religion and tolerance in Islam's global gathering},
  author={Clingingsmith, David and Khwaja, Asim Ijaz and Kremer, Michael},
  journal={The Quarterly Journal of Economics},
  volume={124},
  number={3},
  pages={1133--1170},
  year={2009},
  publisher={MIT Press}
}

@article{hu2008instrumental,
  title={Instrumental variable treatment of nonclassical measurement error models},
  author={Hu, Yingyao and Schennach, Susanne M},
  journal={Econometrica},
  volume={76},
  number={1},
  pages={195--216},
  year={2008},
  publisher={Wiley Online Library}
}

@article{chen2005measurement,
  title={Measurement error models with auxiliary data},
  author={Chen, Xiaohong and Hong, Han and Tamer, Elie},
  journal={The Review of Economic Studies},
  volume={72},
  number={2},
  pages={343--366},
  year={2005},
  publisher={Wiley-Blackwell}
}

@article{vanderweele2022statistical,
  title={A statistical test to reject the structural interpretation of a latent factor model},
  author={VanderWeele, Tyler J and Vansteelandt, Stijn},
  journal={Journal of the Royal Statistical Society: Series B (Statistical Methodology)},
  volume={84},
  number={5},
  pages={2032--2054},
  year={2022},
  publisher={Wiley Online Library}
}

@article{newey2003instrumental,
  title={Instrumental variable estimation of nonparametric models},
  author={Newey, Whitney K and Powell, James L},
  journal={Econometrica},
  volume={71},
  number={5},
  pages={1565--1578},
  year={2003},
  publisher={Wiley Online Library}
}

@article{horowitz2011applied,
  title={Applied nonparametric instrumental variables estimation},
  author={Horowitz, Joel L},
  journal={Econometrica},
  volume={79},
  number={2},
  pages={347--394},
  year={2011},
  publisher={Wiley Online Library}
}

@article{darolles2011nonparametric,
  title={Nonparametric instrumental regression},
  author={Darolles, Serge and Fan, Yanqin and Florens, Jean-Pierre and Renault, Eric},
  journal={Econometrica},
  volume={79},
  number={5},
  pages={1541--1565},
  year={2011},
  publisher={Wiley Online Library}
}

@article{miao2024confounding,
  title={A confounding bridge approach for double negative control inference on causal effects},
  author={Miao, Wang and Shi, Xu and Li, Yilin and Tchetgen Tchetgen, Eric J},
  journal={Statistical Theory and Related Fields},
  volume={8},
  number={4},
  pages={262--273},
  year={2024},
  publisher={Taylor \& Francis}
}

@article{santos2011instrumental,
  title={Instrumental variable methods for recovering continuous linear functionals},
  author={Santos, Andres},
  journal={Journal of Econometrics},
  volume={161},
  number={2},
  pages={129--146},
  year={2011},
  publisher={Elsevier}
}

@article{ai2012semiparametric,
  title={The semiparametric efficiency bound for models of sequential moment restrictions containing unknown functions},
  author={Ai, Chunrong and Chen, Xiaohong},
  journal={Journal of Econometrics},
  volume={170},
  number={2},
  pages={442--457},
  year={2012},
  publisher={Elsevier}
}

@article{carrasco2007linear,
  title={Linear inverse problems in structural econometrics estimation based on spectral decomposition and regularization},
  author={Carrasco, Marine and Florens, Jean-Pierre and Renault, Eric},
  journal={Handbook of econometrics},
  volume={6},
  pages={5633--5751},
  year={2007},
  publisher={Elsevier}
}

@article{ai2003efficient,
  title={Efficient estimation of models with conditional moment restrictions containing unknown functions},
  author={Ai, Chunrong and Chen, Xiaohong},
  journal={Econometrica},
  volume={71},
  number={6},
  pages={1795--1843},
  year={2003},
  publisher={Wiley Online Library}
}

@article{fu2025causal,
  title={Causal Inference for Experiments with Latent Outcomes: Key Results and Their Implications for Design and Analysis},
  author={Fu, Jiawei and Green, Donald P},
  journal={arXiv preprint arXiv:2505.21909},
  year={2025}
}

@book{kress1989linear,
  title={Linear integral equations},
  author={Kress, Rainer},
  volume={82},
  year={1989},
  publisher={Springer}
}

@article{miao2018identifying,
  title={Identifying causal effects with proxy variables of an unmeasured confounder},
  author={Miao, Wang and Geng, Zhi and Tchetgen Tchetgen, Eric J},
  journal={Biometrika},
  volume={105},
  number={4},
  pages={987--993},
  year={2018},
  publisher={Oxford University Press}
}

@article{chamberlain1987asymptotic,
  title={Asymptotic efficiency in estimation with conditional moment restrictions},
  author={Chamberlain, Gary},
  journal={Journal of econometrics},
  volume={34},
  number={3},
  pages={305--334},
  year={1987},
  publisher={Elsevier}
}

@article{schennach2022measurement,
  title={Measurement systems},
  author={Schennach, Susanne},
  journal={Journal of Economic Literature},
  volume={60},
  number={4},
  pages={1223--1263},
  year={2022},
  publisher={American Economic Association 2014 Broadway, Suite 305, Nashville, TN 37203-2425}
}

@book{bollen1989structural,
  title={Structural equations with latent variables},
  author={Bollen, Kenneth A},
  volume={210},
  year={1989},
  publisher={John Wiley \& Sons}
}

@article{rubin1980randomization,
  title={Randomization analysis of experimental data: The Fisher randomization test comment},
  author={Rubin, Donald B},
  journal={Journal of the American Statistical Association},
  volume={75},
  number={371},
  pages={591--593},
  year={1980},
  publisher={JSTOR}
}

@book{imbens2015causal,
  title={Causal inference in statistics, social, and biomedical sciences},
  author={Imbens, Guido W and Rubin, Donald B},
  year={2015},
  publisher={Cambridge university press}
}

@book{gerber2012field,
  title={Field experiments: Design, analysis, and interpretation},
  author={Gerber, Alan S and Green, Donald P},
  year={2012},
  publisher={W.W. Norton}
}

@book{angrist2009mostly,
  title={Mostly harmless econometrics: An empiricist's companion},
  author={Angrist, Joshua D and Pischke, J{\"o}rn-Steffen},
  year={2009},
  publisher={Princeton university press}
}

@article{stoetzer2022causal,
  title={Causal inference with latent outcomes},
  author={Stoetzer, Lukas F and Zhou, Xiang and Steenbergen, Marco},
  journal={American Journal of Political Science},
  year={2025},
  pages={624--640},
 volume={69},
 number={2},
  publisher={Wiley Online Library}
}

@article{anderson2008multiple,
  title={Multiple inference and gender differences in the effects of early intervention: A reevaluation of the Abecedarian, Perry Preschool, and Early Training Projects},
  author={Anderson, Michael L},
  journal={Journal of the American Statistical Association},
  volume={103},
  number={484},
  pages={1481--1495},
  year={2008},
  publisher={Taylor \& Francis}
}
